\documentclass[11pt,a4paper]{article}
\usepackage[a4paper]{geometry}
\usepackage{amssymb,latexsym,amsmath,amsfonts,amsthm}
\usepackage{graphicx}
\usepackage{epstopdf}
\usepackage{tikz}
\usepackage{pgflibraryshapes}
\usetikzlibrary{arrows,decorations.markings}
\usepackage{subfigure}
\usepackage{overpic}
\usepackage{fullpage}
\usepackage{comment,verbatim}
\usepackage{hyperref}
\usepackage{color}
\usepackage{mathrsfs}
\usepackage[title,titletoc]{appendix}
\usepackage{ulem}

\DeclareMathOperator{\diag}{diag}

\DeclareMathOperator{\CHF}{CHF}

\newcommand{\C}{\mathbb{C}}

\renewcommand{\Re}{\mathrm{Re}\,}
\renewcommand{\Im}{\mathrm{Im}\,}

\newcommand{\Pe}{\mathrm{Pe}}

\renewcommand{\vec}{\mathbf}
\newcommand{\ud}{\,\mathrm{d}}

\newcommand{\Boh}{\mathcal{O}}

\newtheorem{theorem}{Theorem}[section]
\newtheorem{lemma}[theorem]{Lemma}
\newtheorem{proposition}[theorem]{Proposition}

\newtheorem{corollary}[theorem]{Corollary}

\newtheorem{rhp}[theorem]{RH problem}

\theoremstyle{definition}

\theoremstyle{remark}

\newtheorem{remark}[theorem]{Remark}

\numberwithin{equation}{section}

\begin{document}

\title{On the deformed Pearcey determinant}

\author{Dan Dai\footnotemark[1], ~Shuai-Xia Xu\footnotemark[2] ~and Lun Zhang\footnotemark[3]}

\renewcommand{\thefootnote}{\fnsymbol{footnote}}
\footnotetext[1]{Department of Mathematics, City University of Hong Kong, Tat Chee
Avenue, Kowloon, Hong Kong. E-mail: \texttt{dandai@cityu.edu.hk}}
\footnotetext[2]{Institut Franco-Chinois de l'Energie Nucl\'{e}aire, Sun Yat-sen University,
Guangzhou 510275, China. E-mail: \texttt{xushx3@mail.sysu.edu.cn}}
\footnotetext[3] {School of Mathematical Sciences and Shanghai Key Laboratory for Contemporary Applied Mathematics, Fudan University, Shanghai 200433, China. E-mail: \texttt{lunzhang@fudan.edu.cn }}

\date{\today}

\maketitle

\begin{abstract}
In this paper, we are concerned with the deformed Pearcey determinant $\det\left(I-\gamma K^\Pe_{s,\rho}\right)$, where $0 \leq \gamma<1$ and $K^\Pe_{s,\rho}$ stands for the trace class operator acting on $L^2\left(-s, s\right)$ with the classical Pearcey kernel arising from random matrix theory. This determinant corresponds to the gap probability for the Pearcey process after thinning, which means each particle in the Pearcey process is removed independently with probability $1-\gamma$. We establish an integral representation of the deformed Pearcey determinant involving the Hamiltonian associated with a family of special solutions to a system of nonlinear differential equations. Together with some remarkable differential identities for the Hamiltonian, this allows us to obtain the large gap asymptotics, including the exact calculation of the constant term, which complements our previous work on the undeformed case (i.e., $\gamma=1$). It comes out that the deformed Pearcey determinant exhibits a significantly different asymptotic behavior from the undeformed case, which suggests a transition will occur as the parameter $\gamma$ varies. As an application of our results, we obtain the asymptotics for the expectation and variance of the counting function for the Pearcey process, and a central limit theorem as well.
\end{abstract}

\setcounter{tocdepth}{2} \tableofcontents

\section{Introduction}
Universality of local statistics of eigenvalues for large random matrices is one of the most fascinating phenomena in random matrix theory \cite{EY12,Forrester,Kui11,metha}. This means the local behaviors of the spectrum depend only on the symmetry type of the ensemble but not on the detailed information about the elements of the ensemble. For the classical Gaussian unitary ensemble (GUE), it is well-known that the eigenvalues form a determinantal point process. As the dimension of the matrix goes to infinity, the correlation function tends to the sine kernel in the bulk of the spectrum \cite{metha}, and to the Airy kernel at the edges of the spectrum \cite{TWAiry}. One encounters the same universal local statistics for a large class of random matrices, which particularly include the unitarily invariant ensembles \cite{DG07,DKMVZ99} and the Wigner matrices (i.e., Hermitian matrices with independent, identically distributed entries) \cite{EPRSY,ESY,Sosh99,TaoV11,TaoV10}, just to name a few.

A different local statistics will arise if we consider the following deformed GUE \cite{BH,BH1}
\begin{equation*}
A+\lambda M,
\end{equation*}
where $M$ is a GUE matrix, $A$ is a deterministic diagonal matrix (also known as the external source) and $\lambda$ is a real parameter. If the spectrum of $A$ is symmetric to the origin with a gap around 0, there will be a critical value of $\lambda$ at which the gap in the support of density closes and the density exhibits a cusp-like singularity at the origin \cite{Pastur}, i.e., the density behaves like $|x|^{\frac13}$ from both sides of the origin. This cubic root singularity leads to a new determinantal process characterized by the Pearcey kernel \cite{BK3,BH,BH1,TW,Zinn}.

The Pearcey kernel $K^\Pe$ is defined by (see \cite{BH,BH1})
\begin{align}\label{eq: pearcey kernel}
K^\Pe(x,y;\rho)&=\int_0^{\infty}\mathcal{P}(x+z)\mathcal{Q}(y+z)\ud z
\nonumber
\\
&=\frac{\mathcal{P}(x)\mathcal{Q}''(y)-\mathcal{P}'(x)\mathcal{Q}'(y)+\mathcal{P}''(x)\mathcal{Q}(y)-\rho
\mathcal{P}(x)\mathcal{Q}(y)}{x-y},
\end{align}
where $\rho\in\mathbb{R}$,
\begin{equation}\label{eq:pearcey integral}
\mathcal{P}(x)=\frac{1}{2\pi}\int_{-\infty}^\infty
e^{-\frac14 t^4-\frac{\rho}{2}t^2+itx} \ud t \qquad \text{and} \qquad
\mathcal{Q}(y)=\frac{1}{2\pi} \int_\Sigma e^{\frac14
t^4+\frac{\rho}{2}t^2+ity} \ud t.
\end{equation}
The contour $\Sigma$ in the definition of $\mathcal{Q}$ consists of the four
rays $\arg t=\frac{\pi}{4},\frac34\pi,\frac54\pi,\frac74\pi$, where the first and the
third rays are oriented from infinity to zero while the second and
the last rays are oriented outwards; see Figure \ref{fig: sigma} for an illustration.
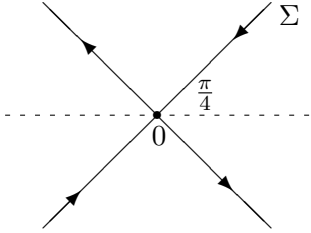
\begin{figure}[h]
\begin{center}
   \setlength{\unitlength}{1truemm}
   \vspace{-13mm}
   \begin{picture}(100,70)(-5,2)
        \dashline{0.8}(20,40)(60,40)
       \put(40,40){\line(1,1){15}}
       \put(40,40){\line(-1,-1){15}}
       \put(40,40){\line(-1,1){15}}
       \put(40,40){\line(1,-1){15}}
       \put(40,40){\thicklines\circle*{1}}
       \put(39.3,36){$0$}
       \put(45,42){$\frac{\pi}{4}$}
       \put(50,50){\thicklines\vector(-1,-1){.0001}}
       \put(30,50){\thicklines\vector(-1,1){.0001}}
       \put(50,30){\thicklines\vector(1,-1){.0001}}
       \put(30,30){\thicklines\vector(1,1){.0001}}
       \put(56,52){$\Sigma$}
\end{picture}
 \vspace{-23mm}
\caption{The contour $\Sigma$ in the definition of $\mathcal{Q}(y)$.}
\label{fig: sigma}
\end{center}
\end{figure}
The functions $\mathcal{P}$ and $\mathcal{Q}$ in \eqref{eq:pearcey integral} satisfy the following two third order differential equations
\begin{align}
\mathcal{P}'''(x)&=x\mathcal{P}(x)+\rho \mathcal{P}'(x), \label{eq:Pearcey1}
\\
\mathcal{Q}'''(y)&=-y\mathcal{Q}(y)+\rho \mathcal{Q}'(y),
\end{align}
respectively, and are also called Pearcey integrals \cite{Pear}.

Analogously to the universal sine and Airy point processes, the Pearcey process represents another canonical universality class in random matrix theory, as can be seen from its emergence in specific matrix models including large complex correlated Wishart matrices \cite{HHNa,HHNb}, a two-matrix model with special quartic potential \cite{GZ}, and in complex Hermitian Wigner-type matrices at the cusps under general conditions \cite{EKS} as well. The complex Hermitian Wigner-type matrices generalize the traditional Wigner matrices by dropping the requirement on the identical distribution of the entries. Moreover, a recent classification theorem regarding the singularities of the solution to the underlying Dyson equation shows that the limiting eigenvalue density therein has only square root or cubic root cusp singularities \cite{AjEK2017,AlEK2018}. Thus, the Pearcey process is the third and last universal statistics arising from this generalization of Wigner matrices. It is also worthwhile to mention that the Pearcey statistics have been related to the non-intersecting Brownian motions at the critical time \cite{AOV,AM,BK3} and to a combinatorial model on random partitions \cite{OR}.

Let $K^\Pe_{s,\rho}$ be the trace class operator acting on $L^2\left(-s, s\right)$, $s\geq 0$, with the Pearcey kernel \eqref{eq: pearcey kernel}, it is well-known that the associated Fredholm determinant $\det\left(I-K^\Pe_{s,\rho}\right)$ gives us the probability of finding no particles (also known as the gap probability) on the interval $(-s,s)$  in a determinantal point process on the real line characterized by the Pearcey kernel.  The nonlinear differential equations have been established for this gap probability in \cite{AM,BC1,BH,TW} under more general settings, while its transition to an Airy process and the large gap asymptotics can be found in \cite{ACV,BC1} and \cite{BH,DXZ20}, respectively.

In this paper, we intend to continue our investigation on the Pearcey determinant, initiated in \cite{DXZ20}, by considering
\begin{equation}\label{def:DeformedPearcey}
\det\left(I-\gamma K^\Pe_{s,\rho}\right),  \qquad 0 \leq \gamma<1,
\end{equation}
i.e., a deformed case. This determinant corresponds to the gap probability for the thinned Pearcey process, which means each particle in the Pearcey process is removed independently with probability $1-\gamma$. Thinning is a classical operation in the studies of point processes; cf. \cite{IPSS}. The thinned process is an intermediate process as it interpolates between the original point process (for $\gamma \to 1$) and an uncorrelated process (for $\gamma=0$) \cite{Kallen}. In the context of random matrix theory, they were first introduced by Bohigas and Pato in \cite{Boh06,Boh04} with motivations arising from nuclear physics and have attracted great interest recently. For the classical sine, Airy and Bessel point processes, the deformed distribution functions are all closely related to the Painlev\'{e} equations or the associated Hamiltonians, and exhibit significantly different asymptotic behaviors from the undeformed case (i.e., $\gamma=1$), which in particular implies transitions will occur as the parameter $\gamma$ varies; see \cite{BCI16,Bot:Buck2018,BDIK18,BDIK17,BDIK15,BIP19,Charlier18} for the relevant works. The critical value of $\gamma$ at $1$ can also be seen from the eigenvalues of the associated operators. For instance, let us consider the trace class operator $K_{s}$ acting on $L^2\left(-s, s\right)$ with the sine kernel $K_{\textrm{sin}}(x,y)=\frac{\sin (x-y)}{\pi (x-y)}$. Since $K_s$ is a positive projection operator, we have that the eigenvalues $\lambda_k(s)$, $k=0,1,2,\ldots$, of $K_s$ lie in $(0,1)$ and $\lambda_k(s) \to 1$ as $s\to\infty$ for each fixed $k$. By Lidskii's theorem, it follows that
\begin{equation*}
\det\left(I-\gamma K_{s}\right) = \prod_{k=0}^\infty (1- \gamma \lambda_k(s)).
\end{equation*}
As a consequence, one can see that \cite{Widom}: (i) when $\gamma <1$, the infinite product is exponentially small as  $0<1- \gamma \lambda_k(s)<1$ for all $k$; (ii) when $\gamma = 1$, the infinite product tends to zero much faster than the case $\gamma < 1$ because $1- \lambda_k(s)$ tends to 0 for any fixed $k$; (iii) when $\gamma > 1$, the determinant will vanish for a discrete set of $s$, where $\lambda_k(s)=1/\gamma$ for some $k$,  and this implies that the logarithm of the determinant blows up. This change of behaviors indicates that $\gamma=1$ is a critical phase transition point. It is expected that this phase transition will occur for many other kernels, including the Pearcey kernel, obtained in random matrix theory; cf. \cite{Bothner} for recent studies regarding the classical Airy and Bessel kernels. 

In the present work, we will fill in the gap in understanding the thinned Pearcey point process by working on the deformed Pearcey determinant \eqref{def:DeformedPearcey}. In particular, we are able to establish an integral representation of the deformed Pearcey determinant via the Hamiltonian for a system of differential equations. This, in turn, allows us to derive the large gap asymptotics including the exact evaluation of the constant term. Our results will be stated in the next section; see also \cite{Charlier20,CL21} for further extensions and applications.

%
%

\section{Statement of results}

\subsection{A system of differential equations and a family of special solutions}
The system of differential equations relevant to this work reads as follows:
 \begin{equation}\label{eq:SystemEq}
\left\{
 \begin{array}{ll}
  p_0'(s)=-\sqrt{2}p_3(s)q_2(s),\\
  q_0'(s)=\sqrt{2}p_2(s)q_1(s),\\
  q_1'(s)=q_2(s)-\frac{2}{s}p_2(s)q_ 1(s)q_2(s),\\
  q_2'(s)=\sqrt{2}p_0(s)q_1(s) +q_3(s)+\frac{2}{s}p_2(s)q_ 2(s)^2,\\
  q_3'(s)=sq_1(s)  + \sqrt{2}q_0(s)q_2(s)-\frac{2}{s}p_2(s)q_2(s)q_3(s), \\
  p_1'(s)=-\sqrt{2}p_0(s)p_2(s) -sp_3(s)+\frac{2}{s}p_1(s)p_2(s)q_ 2(s),\\
  p_2'(s)=-\sqrt{2}p_3(s)q_0(s) -p_1(s)-\frac{2}{s}p_2(s)^2q_ 2(s),\\
  p_3'(s)=-p_2(s)+\frac{2}{s}p_2(s)p_ 3(s)q_2(s),
 \end{array}
\right.
\end{equation}
where $p_i(s),q_i(s)$, $i=0,1,2,3$, are 8 unknown functions. By further imposing the condition
\begin{equation}\label{eq:sum0}
\sum_{k=1}^3 p_k(s) q_k(s)=0,
\end{equation}
it is readily to check that the Hamiltonian $H(s)=H(p_0,p_1,p_2,p_3,q_0,q_1,q_2,q_3;s)$ for the above system of differential equations is
given by

\begin{align}\label{pro:H}
H(s)& =  \sqrt{2}p_0(s) p_2(s)q_1(s)+\sqrt{2}p_3(s)q_0(s) q_2(s)+p_1(s)q_2(s) +p_2(s)q_3(s)+sp_3(s)q_1(s)
\nonumber
\\
& ~~~ +\frac{1}{2s}\left(p_1(s)q_1(s)-p_2(s)q_2(s)+p_3(s)q_3(s) \right)^2,
\end{align}
i.e., we have
\begin{equation}\label{eq:H-sys}
q_k'(s)=\frac{\partial H}{\partial p_k}, \qquad p_k'(s)=-\frac{\partial H}{\partial q_k},  \qquad k=0,1,2,3.
\end{equation}

Our first result concerns the existence of a family of special solution to the equations \eqref{eq:SystemEq} and \eqref{eq:sum0}.

\begin{theorem}\label{thm:specialsol}
For the real parameter $\rho\in\mathbb{R}$ and purely imaginary parameter
\begin{equation}\label{def:alpha}
\beta:=\frac{1}{2\pi i} \ln (1-\gamma)\in i \mathbb{R}_+, \qquad \gamma \in [0, 1),
\end{equation}
there exist solutions to the system of differential equations \eqref{eq:SystemEq} and \eqref{eq:sum0} such that the following asymptotic behaviors hold. As $s\to+\infty$, we have
\begin{align}
p_0(s)&=\frac{\sqrt{6}}{2}\beta is^{\frac23}+\frac{\sqrt{2}}{2} \left( \frac{\rho^3}{54}+\frac{\rho}{2} \right)+\Boh(s^{-\frac23}),  \label{thm:p-0-asy} \\
p_1(s)&= -\frac{2\sin(\beta\pi)}{3\pi }e^{\frac{1}{2}\theta_3(s)+\frac{2}{3}\beta \pi i} s^{\frac13}| \Gamma(1-\beta)| \left(\cos \left(\vartheta(s)-\frac{\pi }{3}\right)+ \sqrt{3} \beta i \cos \left(\vartheta(s)+\frac{\pi}{3} \right)\right) \nonumber \\
& ~~~\times \left(1+ \Boh(s^{-\frac23}) \right), \label{thm:p-1-asy} \\
p_2(s)&= \frac{2\sin(\beta\pi)}{3\pi }e^{\frac{1}{2}\theta_3(s)+\frac{2}{3}\beta \pi i} | \Gamma(1-\beta)|\cos (\vartheta(s)) \left( 1+\Boh(s^{-\frac23}) \right), \label{thm:p-2-asy} \\
p_3(s)&= -\frac{2\sin(\beta\pi)}{3\pi }e^{\frac{1}{2}\theta_3(s)+\frac{2}{3}\beta\pi i} s^{-\frac13}| \Gamma(1-\beta)| \cos \left( \vartheta(s)+\frac{\pi}{3} \right) \Big(1+\Boh(s^{-\frac23}) \Big), \label{thm:p-3-asy}
\\
q_0(s)&=-\frac{\sqrt{6}}{2}\beta is^{\frac23}+\frac{\sqrt{2}}{2}\left(-\frac{\rho^3}{54}+\frac{\rho}{2} \right)+\Boh(s^{-\frac23}),  \label{thm:q-0-asy} \\
q_1(s)&= 2ie^{-\frac{1}{2}\theta_3(s)-\frac{2}{3}\beta\pi i} s^{-\frac13}| \Gamma(1-\beta)|\sin \left(\vartheta(s)-\frac{\pi}{3} \right) \left( 1+\Boh(s^{-\frac23}) \right), \label{thm:q-1-asy} \\
q_2(s)&= -2ie^{-\frac{1}{2}\theta_3(s)-\frac{2}{3}\beta\pi i} | \Gamma(1-\beta)|\sin (\vartheta(s) ) \left(1+\Boh(s^{-\frac23}) \right), \label{thm:q-2-asy} \\
q_3(s)&= 2ie^{-\frac{1}{2}\theta_3(s)-\frac{2}{3}\beta\pi i} s^{\frac13}| \Gamma(1-\beta)| \left(\sin \left( \vartheta(s)+\frac{\pi}{3} \right) -\sqrt{3} \beta i \sin \left( \vartheta(s)-\frac{\pi}{3} \right) \right) \nonumber \\
& ~~~ \times \left(1+\Boh(s^{-\frac23}) \right), \label{thm:q-3-asy}
\end{align}
where  $\Gamma(z)$ is Euler's Gamma function, $\theta_3(s)=\theta_3(s;\rho)=\frac34 s^{\frac43}+\frac{\rho}{2}s^{\frac23}$ and
\begin{equation} \label{eq:vtheta}
\vartheta(s)=\vartheta(s;\beta)=-\frac{3 \sqrt{3}}{8} s^{\frac43} + \frac{\sqrt{3}\rho}{4} s^{\frac23}+\arg \Gamma(1-\beta) - \beta i\left(\frac{4}{3}\ln s+\ln \left(\frac{9}{2} \right)\right);
\end{equation}
as $s\to 0^{+}$, we have
\begin{align}
p_0(s) &= \frac{\sqrt{2}}{2}\left( \frac{\rho^3}{54}+\frac{\rho}{2} \right)+\Boh(s),  \label{thm:p-0-asy-0} \\
p_1(s) &= \Boh(s), \qquad p_2(s) = \Boh(1), \qquad p_3(s) = \Boh(s), \label{thm:p-1-asy-0} \\
q_0(s) &= \frac{\sqrt{2}}{2} \left(-\frac{\rho^3}{54}+\frac{\rho}{2} \right)+\Boh(s),  \label{thm:q-0-asy-0} \\
q_1(s) &= \Boh(1), \qquad q_2(s) = \Boh(s), \qquad q_3(s) = \Boh(1).\label{thm:q-1-asy-0}
\end{align}
Moreover, the functions $p_0(s)$ and $q_0(s)$ satisfy the following coupled differential system:
\begin{align}
p_0'''(s) & = \rho p_0'(s)-\frac{2\sqrt{2}q_0'(s)p_0'(s)^2}{s^2 \left( p_0(s)+q_0(s)-\frac{\rho}{\sqrt{2}} \right)} + \left(1+\frac{2\sqrt{2}}{s}p_0'(s)\right) \nonumber \\
& ~~~ \times \left( s \left( p_0(s)+q_0(s)-\frac{\rho}{\sqrt{2}} \right)+\frac{2q_0'(s)p_0''(s)+q_0''(s)p_0'(s)}{p_0(s)+q_0(s)-\frac{\rho}{\sqrt{2}}} \right. \nonumber \\
& \hspace{4.5cm} \left. -\frac{p_0'(s)q_0'(s)\left( 2q_0'(s)+p_0'(s) \right)}{\left( p_0(s)+q_0(s)-\frac{\rho}{\sqrt{2}} \right)^2}
\right) , \label{eq:eqforp-0}
\end{align}
and
\begin{align}
  q_0''(s)&= -p_0''(s) +\frac{p_0'(s)q_0'(s)}{p_0(s)+q_0(s)-\frac{\rho}{\sqrt{2}}} \left( 3+\frac{2\sqrt{2}}{s}\left( p_0'(s)-q_0'(s) \right) \right)  \nonumber \\
  & ~~~ +\sqrt{2}\left( p_0(s)+q_0(s) \right)\left( p_0(s)+q_0(s)-\frac{\rho}{\sqrt{2}} \right) . \label{eq:p0q0-1}
\end{align}
\end{theorem}

We note that the Hamiltonian \eqref{pro:H} is equivalent to that used in Br\'ezin and Hikami \cite{BH} via an elementary transformation
\begin{equation}
	H^{\mathtt{BH}}(u,v,P_0,P_1,P_2, Q_0,Q_1,Q_2; s)=-H \left( -\frac{u}{\sqrt{2}},-P_0,-P_1,-P_2, -\frac{v}{\sqrt{2}}, Q_0,Q_1,Q_2;s \right),
\end{equation}
where $H^{\mathtt{BH}},u,v, P_0,P_1,P_2, Q_0,Q_1,Q_2$ stand for the notations in \cite{BH}. In the special case $\rho=0$, the nonlinear differential equation \eqref{eq:eqforp-0} is first obtained in \cite[Equation (3.25)]{BH}, and  the second order differential equation \eqref{eq:p0q0-1} can be viewed as the first integral of the coupled third order differential equations \cite[Equations (3.25) and (3.26)]{BH}; see also \cite[Equation (3.26)]{BH} for the other third order nonlinear differential equation for $p_0$ and $q_0$.

\subsection{An integral representation of the deformed Pearcey determinant and large gap asymptotics}
By setting
\begin{equation}\label{def:Fnotation}
F(s;\gamma,\rho):=\ln \det\left(I-\gamma K^\Pe_{s,\rho}\right),
\end{equation}
it comes out that $F$ admits an elegant integral representation in terms of the Hamiltonian $H(s)$ \eqref{pro:H}. In view of the remarkable and well-known connections between several classical distribution functions in random matrix theory and the Painlev\'e equations \cite{BIP19,JM80,TWBessel,TWAiry}, our next theorem provides an analogous result for the deformed Pearcey determinant.
\begin{theorem}\label{thm:TW}
With the function $F(s;\gamma,\rho)$ defined in \eqref{def:Fnotation}, we have
\begin{equation}\label{thm: IntRep}
F(s;\gamma,\rho)=2\int_0^sH(\tau) \ud \tau, \qquad s\in (0,+\infty),
\end{equation}
where $H(s)$ is the Hamiltonian \eqref{pro:H} associated with the family of solutions specified in Theorem \ref{thm:specialsol}. Moreover, $H(s)$ satisfies the following asymptotic behaviors: as $s \to 0^+$,
\begin{equation} \label{thm:H-asy-0}
H(s)=\Boh(1),
\end{equation}
and $s \to +\infty$,
\begin{align} \label{thm:H-asy-infty}
H(s)= \sqrt{3} \beta i s ^{\frac13}- \frac{\rho \beta i}{\sqrt{3} s^{\frac13}}-\frac{4\beta^2}{3s} -\frac{2\sqrt{3} \beta i}{9s} \cos(2\vartheta(s))+\Boh(s^{-\frac53}),
\end{align}
where $\beta$ and $\vartheta(s)$ are defined in \eqref{def:alpha} and \eqref{eq:vtheta}, respectively.
\end{theorem}

The local behavior of $H$ near the origin \eqref{thm:H-asy-0} ensures the integral formula \eqref{thm: IntRep} is well-defined. Moreover, since the derivative of $H$ with respect to $s$ can be represented in terms of the functions $p_0(s)$ and  $q_0(s)$ (see \eqref{eq: dH-s} below), it then follows from \eqref{thm: IntRep} that, after an integration by parts, $F$ also admits a (complicated) integral representation only involving $p_0(s)$ and $q_0(s)$, which resembles the Tracy-Widom type formula for the classical distribution functions in random matrix theory \cite{TWBessel,TWAiry}.

A direct consequence of Theorem \ref{thm:TW} is that one can easily obtain the first few terms in the asymptotic expansion of $F(s;\gamma,\rho)$ as $s\to +\infty$, except for the constant term. In the literature, it is an important but challenging task to solve the so-called ``constant problem" in the large gap asymptotics \cite{Kra1}. By further exploring several differential identities for the Hamiltonian $H$ (see Proposition \ref{prop:H-diff} below), we have succeeded in resolving this problem for the present case and obtained the following large gap asymptotics.
\begin{theorem}\label{thm:FAsy}
With the function $F(s;\gamma,\rho)$ defined in \eqref{def:Fnotation}, we have, as $s\to +\infty$,
\begin{align}
  F(s;\gamma,\rho) = & \ \frac{3\sqrt{3} \beta i}{2} s ^{\frac43} - \sqrt{3} \rho \beta i  s^{\frac23}-\frac{8\beta^2}{3} \ln s \nonumber  \\
  & \ -2\beta^2\ln \left( \frac{9}{2} \right) +2\ln \left( G(1+\beta)G(1-\beta) \right)+\Boh(s^{-\frac23}), \qquad 0\leq \gamma <1, \label{main:F-asy}
\end{align}
uniformly for $\gamma$ and $\rho$ in any compact subset of $[0,1)$ and $\mathbb{R}$, respectively, where $\beta$ is given in \eqref{def:alpha} and $G(z)$ is the Barnes G-function.
\end{theorem}

If $\gamma=0$, we have $\beta=0$. Since $G(1)=1$, the large gap asymptotics \eqref{main:F-asy} reads $F(s;0,\rho)=\Boh(s^{-\frac23})$, which is consistent with the fact that $F(s;0,\rho)=0$. If $\gamma=1$, our recent work
\cite{DXZ20} shows that
\begin{equation} \label{F-gamma=1-asy}
F(s;1,\rho)= -\frac{9 s^{\frac83}}{2^{\frac{17}3}} + \frac{\rho s^2}{4} - \frac{\rho^2 s^{{\frac43}}}{2^{{\frac{10}3}}} - \frac{2}{9} \ln s +\frac{\rho^4}{216} + C + \Boh(s^{-\frac{2}{3}}), \qquad s \to +\infty,
\end{equation}
uniformly for $\rho$ in any compact subset of $\mathbb{R}$, where $C$ is an undetermined constant independent of $\rho$ and $s$; see \cite[Theorem 1.1]{DXZ20}. Clearly, the asymptotics of $F(s;\gamma,\rho)$ is not uniformly valid for $\gamma \in [0,1]$, as can be seen from the fact that the exponent of the leading term drops by half, i.e., $s^{\frac83}$ in \eqref{F-gamma=1-asy} reduces to $s^{\frac43}$ in \eqref{main:F-asy}. Such kind of phenomenon seems `universal' since it also appears in the deformed Airy, sine and Bessel determinants; cf. \cite{Bot:Buck2018,BIP19,Bud:Bus95,Charlier18}.
To describe the transition rigorously, a delicate uniform asymptotic analysis is needed to capture the behaviour of $F(s;\gamma,\rho)$ as $s \to +\infty$ and $\gamma \to 1$.
We will not  address the transition problem in this paper, but refer interested readers to \cite{BDIK18} where a complete transition picture for the  sine process is demonstrated.


Finally, it is worthwhile to point out that a combination of the proof of Theorem \ref{thm:TW} and our work \cite{DXZ20} on $F(s;1,\rho)$ shows that the integral representation \eqref{thm: IntRep} still holds for $\gamma=1$, but the asymptotic behavior of $H$ at $+\infty$ should be replaced by
\begin{equation}
H(s)=-\frac{3 s^{\frac53}}{2^{\frac{11}{3}}} + \frac{\rho s}{4} - \frac{\rho^2 s^{\frac13}}{3 \cdot 2^{\frac73}} - \frac{1}{9 s} + \Boh(s^{-\frac53});
\end{equation}
which is readily seen from \eqref{F-gamma=1-asy}.

\subsection{Applications}
By interpreting the Pearcey determinant \eqref{def:DeformedPearcey} as a moment generating function, we are able to obtain some useful information  about the Pearcey process. More precisely, let us denote by $N(s)$ to be the random variable for the number of particles in the Pearcey process falling into the interval $(-s,s)$. It is well-known that the following moment generating function of the occupancy number $N$
\begin{equation}
\mathbb{E}\left(e^{-2\pi \nu N(s)}\right)=\sum_{k=0}^{\infty}\mathbb{P}(N(s)=k)e^{-2 \pi \nu k}, \qquad \nu \geq 0,
\end{equation}
is equal to the Fredholm determinant $ \det\left(I-(1-e^{-2\pi \nu}) K^\Pe_{s,\rho}\right)$; cf. \cite{Johansson06,Sosh00}. This connection, together with Theorem \ref{thm:FAsy}, allows us to establish the expectation, variance and a central limit theorem for the counting function associated with the Pearcey process; see also \cite{Charlier18,CC20,SoshJSP} for the results about the sine, Airy, Bessel and other determinantal processes.
\begin{corollary}\label{cor:application}
As $s\to +\infty$, we have
\begin{equation}\label{eq:expecvar}
\mathbb{E}(N(s))=\mu(s)+\Boh\left( \frac{ \ln s }{s^{\frac23}} \right), \qquad  {\rm Var}(N(s))=\sigma(s)^2+\frac{1+\ln\left(\frac92\right)+\gamma_{{\rm E}}}{\pi^2} +\Boh\left( \frac{(\ln s)^2 }{s^{\frac23}} \right),
\end{equation}
where $\gamma_{{\rm E}}=-\Gamma'(1)\thickapprox 0.57721$ is Euler's constant,
\begin{equation}\label{def:musigma}
\mu(s)= \frac{3\sqrt{3}}{4\pi}s^{\frac43}-\frac{\sqrt{3}\rho}{2\pi}s^{\frac23},\qquad \sigma(s)^2=\frac{4}{3\pi^2}\ln s.
\end{equation}
Moreover, the random variable $\frac{N(s)-\mu(s)}{\sqrt{\sigma(s)^2}}$ converges in distribution to the normal law $\mathcal{N}(0,1)$ as $s\to +\infty$.
\end{corollary}
\begin{proof}
On the one hand, it is readily seen that, as $\nu \to 0$,
\begin{equation}\label{eq:momentexpan}
\mathbb{E}\left(e^{-2\pi \nu N(s)}\right)=1-2\pi \mathbb{E}(N(s))\nu+2\pi^2 \mathbb{E}(N(s)^2 )\nu^2+\Boh(\nu^3).
\end{equation}
On the other hand, by taking $\gamma=1-e^{-2 \pi \nu}$ in \eqref{main:F-asy}, we have
\begin{align}\label{eq:detexpansion1}
 &  \det\left(I-(1-e^{-2\pi \nu}) K^\Pe_{s,\rho}\right)
 \nonumber
 \\
 &=\left(\frac92\right)^{2\nu^2}G(1+\nu i)^2G(1-\nu i)^2 e^{-2\pi\mu(s)\nu+2\pi^2\sigma(s)^2\nu^2}(1+\Boh(s^{-\frac23})),\qquad s\to+\infty,
\end{align}
where the functions $\mu(s)$ and $\sigma(s)^2$ are defined in \eqref{def:musigma}. Recall the following definition of the Barnes G-function (see \cite[Equation 5.17.4]{DLMF}):
\begin{equation}\label{eq:G}
\ln (G(1+z))=\frac{z}{2}\ln(2\pi)-\frac{z(z+1)}{2}+z\ln(\Gamma(1+z))-\int_0^z \ln (\Gamma(1+x))\ud x, \qquad \Re z>-1,
\end{equation}
it follows that
\begin{equation}
G(1+z)=1+\frac{\ln(2\pi)-1}{2}z+\left(\frac{(\ln(2\pi)-1)^2}{8}-\frac{1+\gamma_{\textrm{E}}}{2}\right)z^2+\Boh(z^3),\qquad z\to 0,
\end{equation}
where $\gamma_{{\rm E}}$ is Euler's constant. This gives us the following expansion for the leading term in \eqref{eq:detexpansion1} as $\nu \to 0$:
\begin{align}
& \left(\frac92\right)^{2\nu^2}G(1+\nu i)^2G(1-\nu i)^2 e^{-2\pi\mu(s)\nu+2\pi^2\sigma(s)^2\nu^2} \nonumber \\
&  = 1-2\pi \mu(s)\nu+2\left(\ln \left(\frac92\right)+1+\gamma_\textrm{E}+\pi^2(\mu(s)^2+\sigma(s)^2)\right)\nu^2+\Boh(\nu^3). \label{eq:detexpansion2}
\end{align}
Since
\begin{equation}\label{eq:relationE}
\mathbb{E}\left(e^{-2\pi \nu N(s)}\right)=\det\left(I-(1-e^{-2\pi \nu}) K^\Pe_{s,\rho}\right),
\end{equation}
we have from \eqref{eq:momentexpan} that
\begin{align*}
\mathbb{E}(N(s)) &= -\frac{1}{2 \pi} \frac{\partial}{\partial \nu} \det\left(I-(1-e^{-2\pi \nu}) K^\Pe_{s,\rho}\right) \Big|_{\nu = 0}, \\
{\rm Var}(N(s)) &= \frac{1}{4 \pi^2} \left[ \frac{\partial^2}{\partial \nu^2} \det\left(I-(1-e^{-2\pi \nu}) K^\Pe_{s,\rho}\right) - \Big( \frac{\partial}{\partial \nu} \det\left(I-(1-e^{-2\pi \nu}) K^\Pe_{s,\rho}\right) \Big)^2 \right] _{\nu = 0}.
\end{align*}
Substituting \eqref{eq:detexpansion1} and \eqref{eq:detexpansion2} into the above formulas, we obtain  \eqref{eq:expecvar}. Note that the additional factors $\ln s$ and $(\ln s)^2$ appear in the error terms due to the derivative with respect to $\nu$; see also \eqref{eq:R-est-beta'} below.


To shown the central limit theorem, we observe from \eqref{eq:relationE} and \eqref{eq:detexpansion1} that
\begin{equation}
\mathbb{E}\left(e^{t\cdot \frac{N(s)-\mu(s)}{\sqrt{\sigma(s)^2}}}\right) \to e^{\frac{t^2}{2}}, \qquad s\to+\infty,
\end{equation}
which implies the convergence of $\frac{N(s)-\mu(s)}{\sqrt{\sigma(s)^2}}$ in distribution to the normal law $\mathcal{N}(0,1)$.

This completes the proof of Corollary \ref{cor:application}.
\end{proof}

The rest of this paper is devoted to the proofs of our main results. Following the general strategy established in \cite{Bor:Dei2002,DIZ97}, the proofs essentially boil down to the analysis of a $3\times 3$ Riemann-Hilbert (RH) problem.  In Section \ref{sec:RHP}, we recall an RH characterization of the Pearcey kernel given in \cite{BK3} and relate $\frac{\ud}{\ud s}F(s;\gamma, \rho)$ to a $3 \times 3$ RH problem for $\Phi$ with constant jumps. We then derive a Lax pair for $\Phi$ in Section \ref{sec:Lax}, and the system of differential equations follows from the associated compatibility conditions. Several remarkable differential identities for the Hamiltonian are also included therein for later use. We carry out a Deift-Zhou steepest descent analysis \cite{DZ93} on the RH problem for $\Phi$ as $s\to+\infty$ and $s\to 0^+$ in Section \ref{sec:AsyPhiinfty} and Section \ref{sec:AsyPhi0}, respectively. These asymptotic outcomes, together with the differential identities for the Hamiltonian, will finally lead to the proofs of our main results, as presented in Section \ref{sec:proofs}.

\section{An RH formulation} \label{sec:RHP}

\subsection{An RH characterization of the Pearcey kernel}

Our starting point is an alternative representation of the Pearcey kernel $K^\Pe$ via a $3 \times 3$ RH problem, as shown in \cite{BK3} and stated next.

\begin{rhp} \label{rhp: Pearcey}
We look for a $3 \times 3$ matrix-valued function $\Psi(z)=\Psi(z;\rho)$ satisfying
\begin{itemize}
\item[\rm (1)] $\Psi(z)$ is defined and analytic in $\C \setminus
\{\cup_{j=0}^5\Sigma_j \cup \{ 0 \} \}$,
where
\begin{equation}\label{def:sigmai}
\begin{aligned}
&\Sigma_0=(0,+\infty), \qquad \Sigma_1=e^{\frac{\pi i}{4}}(0,+\infty), \qquad \Sigma_2=e^{\frac{ 3 \pi i}{4}}(0,+\infty),
\\
&\Sigma_3=(-\infty, 0), \qquad \Sigma_4=e^{-\frac{3\pi i}{4}}(0,+\infty), \qquad \Sigma_5=e^{-\frac{\pi i}{4}}(0,+\infty),
\end{aligned}
\end{equation}
with the orientations as shown in Figure \ref{fig:Pearcey}.

\item[\rm (2)] For $z\in \Sigma_k$, $k=0,1,\ldots,5$, the limiting values
\[ \Psi_+(z) = \lim_{\substack{\zeta \to z \\\zeta\textrm{ on $+$-side of }\Sigma_k}}\Psi(\zeta), \qquad
   \Psi_-(z) = \lim_{\substack{\zeta \to z \\\zeta\textrm{ on $-$-side of }\Sigma_k}}\Psi(\zeta), \]
exist, where the $+$-side and $-$-side of $\Sigma_k$ are the sides
which lie on the left and right of $\Sigma_k$, respectively, when
traversing $\Sigma_k$ according to its orientation. These limiting
values satisfy the jump relation
\begin{equation}\label{jumps:M}
\Psi_{+}(z) = \Psi_{-}(z)J_{\Psi}(z),\qquad z\in \cup_{j=0}^5\Sigma_j,
\end{equation}
where
\begin{equation}\label{def:JPsi}
J_{\Psi}(z):= \left\{
        \begin{array}{ll}
          \begin{pmatrix} 0&1&0 \\ -1&0&0 \\ 0&0&1 \end{pmatrix}, & \qquad \hbox{$z\in \Sigma_0$,} \\
          \begin{pmatrix} 1&0&0 \\ 1&1&1 \\ 0&0&1 \end{pmatrix}, & \qquad \hbox{$z\in \Sigma_1$,} \\
          \begin{pmatrix} 1&0&0 \\ 0&1&0 \\ 1&1&1 \end{pmatrix}, & \qquad \hbox{$z\in \Sigma_2$,} \\
          \begin{pmatrix} 0&0&1 \\ 0&1&0 \\ -1&0&0 \end{pmatrix}, & \qquad \hbox{$z\in \Sigma_3$,} \\
          \begin{pmatrix} 1&0&0 \\ 0&1&0 \\ 1&-1&1 \end{pmatrix}, & \qquad \hbox{$z\in \Sigma_4$,} \\
          \begin{pmatrix} 1&0&0 \\ 1&1&-1 \\ 0&0&1 \end{pmatrix}, & \qquad \hbox{$z\in \Sigma_5$.}
        \end{array}
      \right.
\end{equation}

\item[\rm (3)]  As $z \to \infty$ and $\pm\Im z>0$, we have
\begin{equation}\label{eq:asyPsi}
 \Psi(z)=
\sqrt{\frac{2\pi}{3}} e^{\frac{\rho^2}{6}}i \Psi_0
\left(I+ \frac{\Psi_1}{z} +\mathcal \Boh(z^{-2}) \right)\diag \left(z^{-\frac13},1,z^{\frac13} \right)L_{\pm} e^{\Theta(z)},
\end{equation}
where
\begin{equation} \label{asyPsi:coeff}
\Psi_0 =  \begin{pmatrix}
1 & 0 & 0 \\
0 & 1  & 0 \\
\kappa_3(\rho)+\frac{2\rho}{3} & 0 & 1
\end{pmatrix}, \qquad
\Psi_1 = \begin{pmatrix}
0 & \kappa_3(\rho) & 0 \\
\widetilde\kappa_6(\rho) & 0 & \kappa_3(\rho) + \frac{\rho}{3} \\
0 & \widehat \kappa_6(\rho)  & 0
\end{pmatrix},
\end{equation}
with
\begin{equation} \label{poly:kappa3}
	\kappa_3(\rho) = \frac{\rho^3}{54} - \frac{\rho}{6},
\end{equation}
$$\widetilde\kappa_6(\rho) = \kappa_6(\rho) + \frac{\rho}{3} \kappa_3(\rho) - \frac{1}{3},\qquad \widehat\kappa_6(\rho) = \kappa_6(\rho) -\kappa_3(\rho)^2 + \frac{\rho^2}{9} - \frac{1}{3},$$
and
\begin{equation} \label{poly:kappa6}
	\kappa_6(\rho) = \frac{\rho^6}{5832} -  \frac{\rho^4}{162} - \frac{\rho^2}{72} + \frac{7}{36}.
\end{equation}
Moreover, $L_{\pm}$ are the constant matrices
\begin{align}\label{def:Lpm}
L_{+}=
\begin{pmatrix}
-\omega & \omega^2 & 1 \\ -1&1&1 \\ -\omega^2 & \omega & 1
\end{pmatrix},
\qquad
L_{-}=
\begin{pmatrix}
\omega^2 & \omega & 1 \\ 1&1&1 \\ \omega & \omega^2 & 1
\end{pmatrix},
\end{align}
with $\omega=e^{\frac{2\pi i}{3}}$, and $\Theta(z)$ is given by
\begin{align}\label{def:Theta}
\Theta(z)=\Theta(z;\rho)&= \begin{cases}
\diag (\theta_1(z;\rho),\theta_2(z;\rho),\theta_3(z;\rho)), & \text{$\Im z >0$,} \\
\diag (\theta_2(z;\rho),\theta_1(z;\rho),\theta_3(z;\rho)), & \text{$\Im z <0$,} \\
\end{cases}
\end{align}
with
\begin{equation} \label{eq: theta-k-def}
\theta_k(z;\rho)=\frac34 \omega^{2k}z^{\frac43}+\frac{\rho}{2}\omega^kz^{\frac23}, \qquad k=1,2,3.
\end{equation}

\item[\rm (4)] $\Psi(z)$ is bounded near the origin.
\end{itemize}
\end{rhp}

\begin{figure}[h]
\begin{center}
   \setlength{\unitlength}{1truemm}
   \begin{picture}(100,70)(-5,2)
       \put(40,40){\line(-1,-1){20}}
       \put(40,40){\line(-1,1){20}}
       \put(40,40){\line(-1,0){30}}
       \put(40,40){\line(1,0){30}}
       \put(40,40){\line(1,1){20}}
       \put(40,40){\line(1,-1){20}}

       \put(30,50){\thicklines\vector(1,-1){1}}
       \put(30,40){\thicklines\vector(1,0){1}}
       \put(30,30){\thicklines\vector(1,1){1}}
       \put(50,50){\thicklines\vector(1,1){1}}
       \put(50,40){\thicklines\vector(1,0){1}}
       \put(50,30){\thicklines\vector(1,-1){1}}

       \put(39,36.3){$0$}

       \put(20,15){$\Sigma_4$}
       \put(20,63){$\Sigma_2$}
       \put(3,40){$\Sigma_3$}
       \put(60,15){$\Sigma_5$}
       \put(60,63){$\Sigma_1$}
       \put(75,40){$\Sigma_0$}

       \put(22,44){$\Theta_2$}
       \put(22,34){$\Theta_3$}
       \put(55,44){$\Theta_0$}
       \put(55,34){$\Theta_5$}
       \put(38,53){$\Theta_1$}
       \put(38,22){$\Theta_4$}

       \put(40,40){\thicklines\circle*{1}}

   \end{picture}
   \vspace{-10mm}
   \caption{The jump contours $\Sigma_{k}$ and the regions $\Theta_k$, $k=0,1,\ldots,5$, for the RH problem for $\Psi$.}
   \label{fig:Pearcey}
\end{center}
\end{figure}
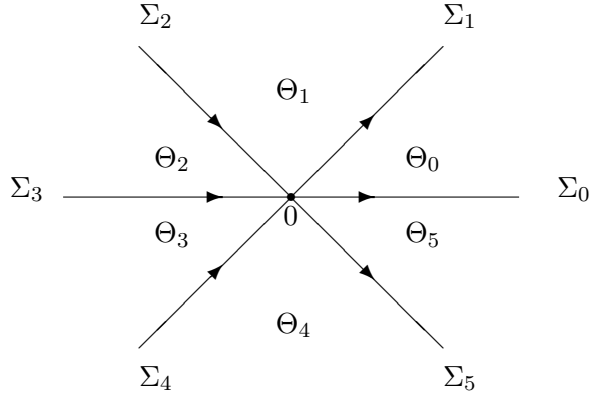

It is shown in \cite[Section 8.1]{BK3} that the above RH problem has a unique solution expressed in terms of solutions of the Pearcey differential equation
\eqref{eq:Pearcey1}. Indeed, note that \eqref{eq:Pearcey1} admits the following solutions:
\begin{equation}\label{def:pj}
\mathcal{P}_j(z)=\mathcal{P}_j(z;\rho)=\int_{\Gamma_j}e^{-\frac14 t^4-\frac{\rho}{2} t^2+it z}\ud t, \qquad j=0,1,\ldots,5,
\end{equation}
where
\begin{align*}
\Gamma_0&=(-\infty,+\infty), \quad && \Gamma_1=(i\infty, 0]\cup[0,\infty),
\\
\Gamma_2&=(i\infty,0]\cup[0,-\infty), \quad && \Gamma_3=(-i\infty, 0]\cup[0,-\infty),
\\
\Gamma_4&=(-i\infty,0]\cup[0,\infty), \quad && \Gamma_5=(-i\infty, i\infty).
\end{align*}
We then have
\begin{equation}\label{def:soltopsi}
\Psi(z)=\left\{
        \begin{array}{ll}
          \begin{pmatrix} -\mathcal{P}_2(z) & \mathcal{P}_1(z) & \mathcal{P}_5(z)\\ -\mathcal{P}_2'(z) & \mathcal{P}_1'(z) & \mathcal{P}_5'(z) \\-\mathcal{P}_2''(z) & \mathcal{P}_1''(z) & \mathcal{P}_5''(z) \end{pmatrix}, & \quad \hbox{$z\in \Theta_0$,} \\
          \begin{pmatrix}  \mathcal{P}_0(z) & \mathcal{P}_1(z) & \mathcal{P}_4(z)\\ \mathcal{P}_0'(z) & \mathcal{P}_1'(z) & \mathcal{P}_4'(z) \\ \mathcal{P}_0''(z) & \mathcal{P}_1''(z) & \mathcal{P}_4''(z) \end{pmatrix}, & \quad \hbox{$z\in \Theta_1$,} \\
          \begin{pmatrix} -\mathcal{P}_3(z) & -\mathcal{P}_5(z) & \mathcal{P}_4(z) \\ -\mathcal{P}_3'(z) & -\mathcal{P}_5'(z) & \mathcal{P}_4'(z) \\-\mathcal{P}_3''(z) & -\mathcal{P}_5''(z) & \mathcal{P}_4''(z) \end{pmatrix}, & \quad \hbox{$z\in \Theta_2$,} \\
          \begin{pmatrix} \mathcal{P}_4(z) & -\mathcal{P}_5(z) & \mathcal{P}_3(z)\\ \mathcal{P}_4'(z) & -\mathcal{P}_5'(z) & \mathcal{P}_3'(z) \\ \mathcal{P}_4''(z) & -\mathcal{P}_5''(z) & \mathcal{P}_3''(z)\end{pmatrix}, & \quad \hbox{$z\in \Theta_3$,} \\
          \begin{pmatrix} \mathcal{P}_0(z) & \mathcal{P}_2(z) & \mathcal{P}_3(z)\\ \mathcal{P}_0'(z) & \mathcal{P}_2'(z) & \mathcal{P}_3'(z) \\ \mathcal{P}_0''(z) & \mathcal{P}_2''(z) & \mathcal{P}_3''(z) \end{pmatrix}, & \quad \hbox{$z\in \Theta_4$,} \\
          \begin{pmatrix} \mathcal{P}_1(z) & \mathcal{P}_2(z) & \mathcal{P}_5(z)\\ \mathcal{P}_1'(z) & \mathcal{P}_2'(z) & \mathcal{P}_5'(z) \\ \mathcal{P}_1''(z) & \mathcal{P}_2''(z) & \mathcal{P}_5''(z) \end{pmatrix}, & \quad \hbox{$z\in \Theta_5$,}
        \end{array}
      \right.
\end{equation}
where $\Theta_k$, $k=0,1,\ldots,5$, is the region bounded by the rays $\Sigma_k$ and $\Sigma_{k+1}$ (with $\Sigma_6:=\Sigma_0$); see Figure \ref{fig:Pearcey} for an illustration.

Now, define
\begin{equation} \label{eq: tilde-psi}
\widetilde \Psi (z)=\widetilde \Psi (z;\rho)=
\begin{pmatrix}
\mathcal{P}_0(z) & \mathcal{P}_1(z) & \mathcal{P}_4(z)
\\
\mathcal{P}_0'(z) & \mathcal{P}_1'(z) & \mathcal{P}_4'(z)
\\
\mathcal{P}_0''(z) & \mathcal{P}_1''(z) & \mathcal{P}_4''(z)
\end{pmatrix}, \qquad z\in \mathbb{C},
\end{equation}
that is, $\widetilde \Psi$ is the analytic extension of the restriction of $\Psi$ on the region $\Theta_1$ to the whole complex plane. The Pearcey kernel \eqref{eq: pearcey kernel} then admits the following equivalent representation in terms of $\widetilde \Psi$ (see \cite[Equation (10.19)]{BK3}):
\begin{equation}\label{eq:PearceyRH}
K^\Pe(x,y;\rho)=\frac{1}{2 \pi i(x-y)}
\begin{pmatrix}
0 & 1 & 1
\end{pmatrix}
\widetilde \Psi(y;\rho)^{-1}\widetilde \Psi(x;\rho)
\begin{pmatrix}
1
\\
0
\\
0
\end{pmatrix}, \qquad x,y \in \mathbb{R}.
\end{equation}
As a consequence, it is readily seen that
\begin{equation}\label{eq:tildeKdef}
\gamma K^\Pe(x,y;\rho) = \frac{\vec{f}(x)^t\vec{h}(y)}{x-y},
\end{equation}
where
\begin{equation}\label{def:fh}
\vec{f}(x)=\begin{pmatrix}
f_1
\\
f_2
\\
f_3
\end{pmatrix}:=\widetilde \Psi(x)
\begin{pmatrix}
1
\\
0
\\
0
\end{pmatrix}, \qquad
\vec{h}(y)=\begin{pmatrix}
h_1
\\
h_2
\\
h_3
\end{pmatrix}
:=
\frac{\gamma}{2 \pi i}
\widetilde \Psi(y)^{-t} \begin{pmatrix}
0
\\
1
\\
1
\end{pmatrix}.
\end{equation}

\subsection{An RH problem related to $\frac{\ud }{\ud s}F$}

With the function $F$ defined in \eqref{def:Fnotation}, we have
\begin{align}
\frac{\ud}{\ud s}F(s;\gamma,\rho)&=\frac{\ud}{\ud s} \ln \det(I-\gamma K_{s,\rho}^\Pe)
=-\textrm{tr}\left((I-\gamma K_{s,\rho}^\Pe)^{-1} \gamma
\frac{\ud}{\ud s}K_{s,\rho}^\Pe\right) \nonumber \\
&=-R(s,s)-R(-s,-s), \label{eq:derivatives}
\end{align}
where $R(u,v)$ stands for the kernel of the resolvent operator, that is,
$$R=\left(I-\gamma K_{s,\rho}^\Pe\right)^{-1}-I=\gamma K_{s,\rho}^\Pe\left(I-\gamma K_{s,\rho}^\Pe\right)^{-1}=\gamma \left(I-\gamma K_{s,\rho}^\Pe\right)^{-1}K_{s,\rho}^\Pe.$$
In view of \eqref{eq:tildeKdef}, we have that the kernel of the operator $\gamma K_{s,\rho}^\Pe$ is integrable for all $0\leq \gamma \leq 1$ in the sense of \cite{IIKS90}, which also implies that its resolvent kernel is integrable as well; cf. \cite{DIZ97,IIKS90}. Indeed, by setting
\begin{equation}\label{def:FH}
\vec{F}(u)=
\begin{pmatrix}
F_1 \\
F_2 \\
F_3
\end{pmatrix}:=\left(I-\gamma K_{s,\rho}^\Pe\right)^{-1}\vec{f}, \qquad \vec{H}(v)=\begin{pmatrix}
H_1 \\
H_2 \\
H_3
\end{pmatrix}
:=\left(I-\gamma K_{s,\rho}^\Pe\right)^{-1}\vec{h},
\end{equation}
we have
\begin{equation}\label{eq:resolventexpli}
R(u,v)=\frac{\vec{F}(u)^t\vec{H}(v)}{u-v}.
\end{equation}

We next establish a connection between the function $\frac{\ud}{\ud s} F(s;\gamma,\rho)$ and an RH problem with constant jumps, which is based on the fact that the resolvent kernel $R(u,v)$ is related to the following RH problem; see \cite{DIZ97}.
\begin{rhp} \label{rhp:Y}
We look for a $3 \times 3$ matrix-valued function
$Y(z)$ satisfying the following properties:
\begin{enumerate}
\item[\rm (1)] $Y(z)$ is defined and analytic in $\mathbb{C}\setminus [-s,s]$, where the orientation is taken from the left to the right.

\item[\rm (2)] For $x\in(-s,s)$, we have
\begin{equation}\label{eq:Y-jump}
 Y_+(x)=Y_-(x)(I-2\pi i \vec{f}(x)\vec{h}(x)^t),
 \end{equation}
where the functions $\vec{f}$ and $\vec{h}$ are defined in \eqref{def:fh}.
\item[\rm (3)] As $z \to \infty$,
\begin{equation}\label{eq:Y-infty}
 Y(z)=I+\frac{\mathsf{Y}_1}{z}+\mathcal \Boh(z^{-2}).
 \end{equation}

\item[\rm (4)] As $z \to \pm s$, we have $Y(z) = \mathcal \Boh(\ln(z \mp s))$.

\end{enumerate}
\end{rhp}
By \cite{DIZ97}, it follows that
\begin{equation}\label{eq:Yexpli}
Y(z)=I-\int_{-s}^s\frac{\vec{F}(w)\vec{h}(w)^t}{w-z}\ud w
\end{equation}
and
\begin{equation}\label{def:FH2}
\vec{F}(z)=Y(z)\vec{f}(z), \qquad \vec{H}(z)=Y(z)^{-t}\vec{h}(z).
\end{equation}

\begin{remark}
Since $\det(I-\gamma K_{s,\rho}^\Pe)$ stands for the gap probability for the thinned Pearcey process, it is strictly positive and hence invertible. This in turn guarantees the solvability of the RH problem for $Y$.
\end{remark}

Recall the RH problem \ref{rhp: Pearcey} for $\Psi$, we make the following undressing transformation to arrive at an RH problem with constant jumps. To proceed, the four rays $\Sigma_k$, $k=1,2,4,5$, emanating from the origin are replaced by their parallel lines emanating from some special points on the real line. More precisely, we replace $\Sigma_1$ and $\Sigma_5$ by their parallel rays $\Sigma_1^{(s)}$ and $\Sigma_5^{(s)}$ emanating from the point $s$, replace $\Sigma_2$ and $\Sigma_4$ by their parallel rays $\Sigma_2^{(s)}$ and $\Sigma_4^{(s)}$ emanating from the point $-s$. Furthermore, these rays, together with the real axis, divide the complex plane into six regions $\texttt{I-VI}$, as illustrated in Figure \ref{fig:X}.

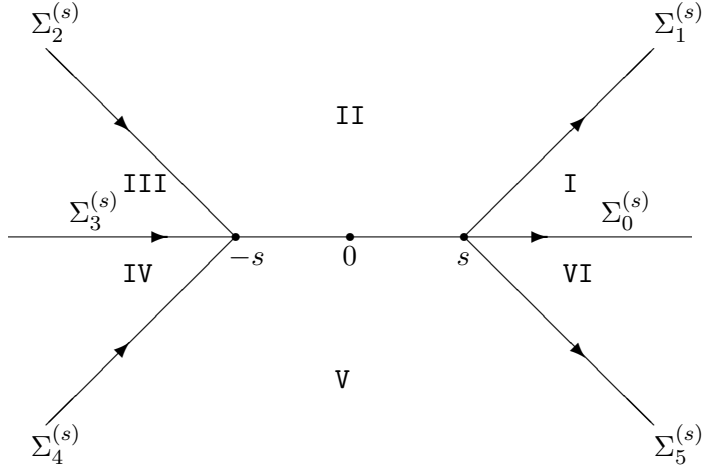
\begin{figure}[h]
\begin{center}
   \setlength{\unitlength}{1truemm}
   \begin{picture}(100,70)(-5,2)
       \put(25,40){\line(-1,0){30}}
       \put(55,40){\line(1,0){30}}

       \put(25,40){\line(1,0){30}}

       \put(25,40){\line(-1,-1){25}}
       \put(25,40){\line(-1,1){25}}

       \put(55,40){\line(1,1){25}}
       \put(55,40){\line(1,-1){25}}

       \put(15,40){\thicklines\vector(1,0){1}}
       \put(65,40){\thicklines\vector(1,0){1}}

       \put(10,55){\thicklines\vector(1,-1){1}}
       \put(10,25){\thicklines\vector(1,1){1}}
       \put(70,25){\thicklines\vector(1,-1){1}}
       \put(70,55){\thicklines\vector(1,1){1}}

       \put(39,36.3){$0$}
       \put(-2,11){$\Sigma_4^{(s)}$}

       \put(-2,67){$\Sigma_2^{(s)}$}
       \put(3,42){$\Sigma_3^{(s)}$}
       \put(80,11){$\Sigma_5^{(s)}$}
       \put(80,67){$\Sigma_1^{(s)}$}
       \put(73,42){$\Sigma_0^{(s)}$}

       \put(10,46){$\texttt{III}$}
       \put(10,34){$\texttt{IV}$}
       \put(68,46){$\texttt{I}$}
       \put(68,34){$\texttt{VI}$}
       \put(38,55){$\texttt{II}$}
       \put(38,20){$\texttt{V}$}

       \put(40,40){\thicklines\circle*{1}}
       \put(25,40){\thicklines\circle*{1}}
       \put(55,40){\thicklines\circle*{1}}

       \put(24,36.3){$-s$}
       \put(54,36.3){$s$}

   \end{picture}
   \caption{Regions $\texttt{I-VI}$ and the contours $\Sigma_{k}^{(s)}$, $k=0,1,\ldots,5$, for the RH problem for $\Phi$.}
   \label{fig:X}
\end{center}
\end{figure}

We now define $\Phi(z)=\Phi(z;s)$ as follows:
\begin{align}\label{eq:YtoX}
 \Phi(z) = \frac{\Psi_0^{-1}}{\sqrt{\frac{2\pi}{3}} e^{\frac{\rho^2}{6}}i} \begin{cases}
   Y(z)\Psi(z), & \hbox{ $z \in \texttt{I}\cup \texttt{III}\cup \texttt{IV} \cup \texttt{VI}$,} \\
         Y(z)\widetilde \Psi(z), & \hbox{ $z \in \texttt{II}$,} \\
Y(z)\widetilde \Psi(z)
\begin{pmatrix}
1 & -1 & -1
\\
0 & 1 & 0
\\
0 & 0 & 1
\end{pmatrix}, & \hbox{$z \in \texttt{V}$,}
 \end{cases}
\end{align}
where $\widetilde \Psi(z)$ is defined in \eqref{eq: tilde-psi} and the constant matrix $\Psi_0$ is given in \eqref{asyPsi:coeff}. Then, $\Phi$ satisfies the following RH problem.

\begin{proposition}\label{rhp:X}
The function $\Phi(z)=\Phi(z;s)$ defined in \eqref{eq:YtoX} has the following properties:
\begin{enumerate}
\item[\rm (1)] $\Phi(z)$ is defined and analytic in $\mathbb{C}\setminus \{\cup^5_{j=0}\Sigma_j^{(s)} \cup[-s,s] \}$, where
\begin{equation}\label{def:sigmais}
\begin{aligned}
&\Sigma_0^{(s)}=(s,+\infty), ~~ &&\Sigma_1^{(s)}=s+e^{\frac{\pi i}{4}}(0,+\infty),  &&&\Sigma_2^{(s)}=-s+e^{\frac{ 3 \pi i}{4}}(0,+\infty),
\\
&\Sigma_3^{(s)}=(-\infty, -s), ~~ &&\Sigma_4^{(s)}=-s+e^{-\frac{3\pi i}{4}}(0,+\infty),  &&&\Sigma_5^{(s)}=s+e^{-\frac{\pi i}{4}}(0,+\infty),
\end{aligned}
\end{equation}
with the orientations from the left to the right; see Figure \ref{fig:X} for an illustration.

\item[\rm (2)] $\Phi$ satisfies the jump condition
\begin{equation}\label{eq:X-jump}
 \Phi_+(z)=\Phi_-(z)J_\Phi(z), \qquad z\in \cup^5_{j=0}\Sigma_j^{(s)} \cup(-s,s),
\end{equation}
where
\begin{equation}\label{def:JX}
J_\Phi(z):=\left\{
 \begin{array}{ll}
          \begin{pmatrix} 0 & 1 &0 \\ -1 & 0 &0 \\ 0&0&1 \end{pmatrix}, & \qquad \hbox{$z\in \Sigma_0^{(s)}$,} \\
          \begin{pmatrix} 1&0&0 \\ 1&1&1 \\ 0&0&1 \end{pmatrix},  & \qquad  \hbox{$z\in \Sigma_1^{(s)}$,} \\
          \begin{pmatrix} 1&0&0 \\ 0&1&0 \\ 1&1&1 \end{pmatrix},  & \qquad \hbox{$z\in \Sigma_2^{(s)}$,} \\
          \begin{pmatrix} 0 &0&1 \\ 0&1&0 \\ -1&0&0 \end{pmatrix}, & \qquad  \hbox{$z\in \Sigma_3^{(s)}$,} \\
          \begin{pmatrix} 1&0&0 \\ 0&1&0 \\ 1&-1&1 \end{pmatrix}, & \qquad  \hbox{$z\in \Sigma_4^{(s)}$,} \\
          \begin{pmatrix} 1&0&0 \\ 1&1&-1 \\ 0&0&1 \end{pmatrix}, & \qquad  \hbox{$z\in \Sigma_5^{(s)}$,} \\
          \begin{pmatrix} 1&1-\gamma&1-\gamma\\ 0&1&0 \\ 0&0&1 \end{pmatrix}, & \qquad  \hbox{$z\in (-s,s)$,}
        \end{array}
      \right.
 \end{equation}

\item[\rm (3)]As $z \to \infty$ and $\pm \Im z>0$, we have
\begin{equation}\label{eq:asyX}
\Phi(z)=
\left(I+ \frac{\mathsf{\Phi}_1}{z} + \frac{\mathsf{\Phi}_2}{z^2} +\mathcal \Boh(z^{-3}) \right)\diag \left(z^{-\frac13},1,z^{\frac13} \right)L_{\pm} e^{\Theta(z)},
\end{equation}
for some functions $\mathsf{\Phi}_1$ and $\mathsf{\Phi}_2$, where $L_{\pm}$ and $\Theta(z)$ are given in \eqref{def:Lpm} and \eqref{def:Theta}, respectively, and
\begin{equation} \label{X1-formula}
	\mathsf{\Phi}_1 = \Psi_1 + \Psi_0^{-1} \mathsf{Y}_1 \Psi_0
\end{equation}
with $\Psi_1$ and $\mathsf{Y}_1$ given in \eqref{asyPsi:coeff} and \eqref{eq:Y-infty}.

\item[\rm (4)] As $z \to s$, we have
\begin{align} \label{eq:X-near-s}
	\Phi(z) =& \ \widehat{\Phi}_1(z) \begin{pmatrix}
	1 & -\frac{\gamma}{2\pi i} \ln(z-s) & -\frac{\gamma}{2\pi i} \ln(z-s) \\
	0 & 1 & 0 \\
	0& 0 & 1
	\end{pmatrix}
\nonumber
\\
&\times
\begin{cases}
	I, & z \in \textrm{\texttt{II}}, \\
	\begin{pmatrix}
1 & -1 & -1
\\
0 & 1 & 0
\\
0 & 0 & 1
\end{pmatrix}, & z \in \textrm{\texttt{V}},
	\end{cases}
\end{align}
where the principal branch is taken for $\ln(z-s)$, and $\widehat{\Phi}_1(z)$ is analytic at $z=s$ satisfying the following expansion
\begin{equation}\label{eq: Phi-expand-s}
 \widehat{\Phi}_1(z) =\Phi_0^{(0)}(s)\left(I+\Phi_1^{(0)}(s)(z-s)+\Boh((z-s)^2)\right ),\qquad z\to s,
 \end{equation}
for some functions $\Phi_0^{(0)}(s)$ and $\Phi_1^{(0)}(s)$.

\item[\rm (5)] As $z \to -s$, the behavior of $\Phi(z)$ is determined by the following symmetric relation
\begin{equation} \label{eq: Xz+-}
  \Phi(z) = -\diag(1,-1,1)  \Phi(-z) \begin{pmatrix}
-1 & 0 & 0
\\
0 & 0 & 1
\\
0 & 1 & 0
\end{pmatrix}
\end{equation}
and \eqref{eq:X-near-s}.
%
\end{enumerate}
\end{proposition}

\begin{proof}
For the jump condition \eqref{eq:X-jump} and \eqref{def:JX}, we only need to check the jump over $(-s,s)$, while the other jumps follow directly from \eqref{eq:YtoX} and the RH problem \ref{rhp: Pearcey} for $\Psi$. By \eqref{def:fh} and \eqref{eq:Y-jump}, we have, for $-s<x<s$,
\begin{equation}
Y_-(x)^{-1}Y_+(x) = I-2\pi i  \vec{f}(x)\vec{h}(x)^t=\widetilde\Psi(x)
\begin{pmatrix}
1 & -\gamma & -\gamma
\\
0 & 1  & 0
\\
0 & 0 & 1
\end{pmatrix}\widetilde\Psi(x)^{-1}.
\end{equation}
This, together with \eqref{eq:YtoX}, implies that for $x\in(-s,s)$,
\begin{align}
J_\Phi(x) = \Phi_-(x)^{-1} \Phi_+(x) &= \begin{pmatrix}
1 & 1 & 1
\\
0 & 1 & 0
\\
0 & 0 & 1
\end{pmatrix} \widetilde \Psi(x)^{-1} Y_-(x)^{-1} Y_+(x) \widetilde \Psi(x)
\nonumber
\\
&= \begin{pmatrix}
1 & 1-\gamma & 1-\gamma
\\
0 & 1  & 0
\\
0 & 0 & 1
\end{pmatrix},
\end{align}
as desired.

Finally, the symmetric relation \eqref{eq: Xz+-} follows from the properties of the jump matrices $J_\Phi(z)$ in \eqref{def:JX} and the large $z$ behavior of $\Phi$ given in \eqref{eq:asyX}; see also the proof of a similar relation in \cite[Equation (2.52)]{DXZ20}.

This completes the proof of Proposition \ref{rhp:X}.
\end{proof}

The connection between the above RH problem and the derivative of $F(s;\gamma,\rho)$ is revealed in the following proposition.
\begin{proposition}\label{prop:derivativeandX}
With $F$ defined in \eqref{def:Fnotation}, we have
\begin{align}
\frac{\ud}{\ud s} F(s;\gamma,\rho)&=\frac{\ud}{\ud s} \ln \det\left(I- \gamma K_{s,\rho}^{\Pe}\right) \nonumber
\\
&=-\frac{\gamma}{ \pi i} \left[ \left(\Phi_1^{(0)} (s) \right)_{21}+ \left(\Phi_1^{(0)} (s) \right)_{31}
\right],
\label{eq:derivativeinsX-2}
\end{align}
where $\Phi_1^{(0)} (s)$ is given in \eqref{eq: Phi-expand-s} and $(M)_{ij}$ stands for the $(i,j)$-th entry of a matrix $M$.
\end{proposition}
\begin{proof}

The proof of is similar to that of \cite[Equation (2.46)]{DXZ20}. For $z\in \texttt{II}$, we see from \eqref{def:fh}, \eqref{def:FH2} and \eqref{eq:YtoX} that
\begin{equation}\label{eq:FHinX}
\vec{F}(z)=Y(z)\vec{f}(z)=Y(z)\widetilde \Psi(z)\begin{pmatrix}
1
\\
0
\\
0
\end{pmatrix}=\sqrt{\frac{2\pi}{3}} e^{\frac{\rho^2}{6}} i \Psi_0 \Phi(z)
\begin{pmatrix}
1
\\
0
\\
0
\end{pmatrix}
\end{equation}
and
\begin{align}
\vec{H}(z)&=Y(z)^{-t}\vec{h}(z)= \frac{\Psi_0^{-t}}{\sqrt{\frac{2\pi}{3}} e^{\frac{\rho^2}{6}} i } \Phi(z)^{-t}\widetilde \Psi(z)^{t} \cdot \frac{\gamma }{2 \pi i} \widetilde \Psi(z)^{-t} \begin{pmatrix}
0
\\
1
\\
1
\end{pmatrix} \nonumber \\
&=\frac{\gamma }{2 \pi i} \frac{\Psi_0^{-t}}{\sqrt{\frac{2\pi}{3}} e^{\frac{\rho^2}{6}} i }  \Phi(z)^{-t}
\begin{pmatrix}
0
\\
1
\\
1
\end{pmatrix}.
\end{align}
Combining the above formulas and \eqref{eq:resolventexpli}, we obtain
\begin{equation} \label{eq: R-Xij}
	R(z,z) = \frac{\gamma}{2 \pi i} \left[ \left(\Phi(z)^{-1}\Phi'(z)\right)_{21}+\left(\Phi(z)^{-1}\Phi'(z)\right)_{31}
\right], \qquad  z\in \texttt{II}.
\end{equation}
This, together with \eqref{eq:derivatives} and \eqref{eq: Xz+-}, implies that
\begin{equation}
 \frac{\ud}{\ud s} F(s;\gamma,\rho) = -\frac{\gamma}{ \pi i} \left[\lim_{z \to s} \left(\left(\Phi(z)^{-1}\Phi'(z)\right)_{21}+\left(\Phi(z)^{-1}\Phi'(z)\right)_{31}\right)
\right], \qquad  z\in \texttt{II}.
\end{equation}
With the aid of the local behavior of $\Phi(z)$ as $z \to s$ given in \eqref{eq:X-near-s} and \eqref{eq: Phi-expand-s}, we arrive at \eqref{eq:derivativeinsX-2}.

This completes the proof of Proposition \ref{prop:derivativeandX}.
\end{proof}

\section{Lax pair equations and differential identities for the Hamiltonian}\label{sec:Lax}
In this section, we will derive a Lax pair for $\Phi(z;s)$ from the associated RH problem, which will lead to the proof of differential equations satisfied by $p_0$ and $q_0$ given in Theorem \ref{thm:specialsol}.
Several useful differential identities for the Hamiltonian $H$ will also be presented for later use.

\subsection{Lax pair equations}
We start with the following Lax pair for $\Phi(z;s)$.
\begin{proposition}\label{pro:Lax pair}
  For the function $\Phi(z) = \Phi(z;s)$ defined in \eqref{eq:YtoX}, we have
  \begin{equation} \label{eq:Lax pair}
    \frac{\partial}{\partial z}\Phi(z;s) = L(z;s)\Phi(z;s), \quad  \frac{\partial}{\partial s}\Phi(z;s) = U(z;s)\Phi(z;s),
  \end{equation}
  where
   \begin{equation} \label{eq:L}
L(z;s)=\begin{pmatrix}
      0 & 0 & 0 \\
      0& 0 & 0 \\
      z   & 0& 0
    \end{pmatrix}+A_0(s)+\frac{A_1(s)}{z-s}+\frac{A_2(s)}{z+s},
  \end{equation}
  and
   \begin{equation} \label{eq:U}
U(z;s)=
     - \frac{A_1(s)}{z-s}+\frac{A_2(s)}{z+s}
  \end{equation}
   with
  \begin{equation} \label{eq:A-k}
    A_0(s)= \begin{pmatrix}
      0 & 1 & 0 \\
      \sqrt{2} p_0(s)& 0 & 1 \\
     0   & \sqrt{2} q_0(s)& 0
    \end{pmatrix}, \qquad   A_1(s)  = \begin{pmatrix}
      q_1(s) \\
      q_2(s) \\
    q_3(s)
    \end{pmatrix}
    \begin{pmatrix}
      p_1(s) &  p_2(s) & p_3(s)
    \end{pmatrix}, \end{equation}
    and $A_2(s)$ being related to $A_1(s)$ through the following relation
   \begin{equation}\label{eq:A-sym}
A_2(s)=\diag(1,-1,1)A_1(s)\diag(1,-1,1).
\end{equation}
Moreover, the functions $p_i(s)$ and $q_i(s)$, $i=0,1,2,3$, in \eqref{eq:A-k} satisfy the equations \eqref{eq:SystemEq} and \eqref{eq:sum0}, and $p_0(s)$ and $q_0(s)$ satisfy the coupled differential system
\eqref{eq:eqforp-0} and \eqref{eq:p0q0-1}.
\end{proposition}
\begin{proof}
The proof is based on the RH problem for $\Phi$ given in Proposition \ref{rhp:X}. Note that all jumps in the RH problem for $\Phi$ are independent of $z$ and $s$,  it follows that coefficient matrices
   \begin{equation}
   L(z;s):=\frac{\partial}{\partial z}\Phi(z;s)\cdot \Phi(z;s)^{-1}, \qquad U(z;s):=\frac{\partial}{\partial s}\Phi(z;s)\cdot \Phi(z;s)^{-1}
   \end{equation}
are analytic in the complex $z$ plane except for possible isolated singularities at $z=\pm s $ and $z = \infty$. Moreover, in view of the symmetry relation between $\Phi(z)$ and $\Phi(-z)$ in \eqref{eq: Xz+-}, it is easy to check that $L(z;s)$ and $U(z;s)$  satisfy the following symmetry relations
\begin{align}
  L(z;s)&=-\diag(1,-1,1)L(-z;s)\diag(1,-1,1), \label{eq:L-sym} \\
  U(z;s)&=\diag(1,-1,1)U(-z;s)\diag(1,-1,1). \label{eq:U-sym}
\end{align}
In what follows, we calculate these two functions explicitly one by one.

From the large $z$ behavior of $\Phi$ given in \eqref{eq:asyX}, we have, as $z \to \infty$,
\begin{align}\label{eq:L-expand}
   &L(z;s)=   \begin{pmatrix}
      0 & 0 & 0 \\
      0 & 0 & 0 \\
      1 & 0 & 0
    \end{pmatrix}z + \begin{pmatrix}
      0 & 1 & 0 \\
      \frac{\rho}{3} & 0 & 1 \\
      0  & \frac{\rho}{3} & 0
    \end{pmatrix}   + \left[ \mathsf{\Phi}_1, \begin{pmatrix}
      0 & 0 & 0 \\
      0 & 0 & 0 \\
      1 & 0 & 0
    \end{pmatrix}\right]  + \frac{\mathsf{L}_1}{z}  + \Boh(z^{-2}),
  \end{align}
where
\begin{equation}
  \mathsf{L}_1 = \begin{pmatrix}
      -\frac{1}{3} & 0 & \frac{\rho}{3} \\
      0 & 0 & 0 \\
      0 & 0 & \frac{1}{3}
    \end{pmatrix} +\left[ \mathsf{\Phi}_2, \begin{pmatrix}
      0 & 0 & 0 \\
      0 & 0 & 0 \\
      1 & 0 & 0
    \end{pmatrix}\right]+\left[ \begin{pmatrix}
      0 & 0 & 0 \\
      0 & 0 & 0 \\
      1 & 0 & 0
    \end{pmatrix} \mathsf{\Phi}_1,\mathsf{\Phi}_1\right]+\left[ \mathsf{\Phi}_1, \begin{pmatrix}
      0 & 1 & 0 \\
      \frac{\rho}{3} & 0 & 1 \\
      0 & \frac{\rho}{3} & 0
    \end{pmatrix}\right]
\end{equation}
and $[A,B]$ denotes the commutator of two matrices, i.e., $[A,B] = AB - BA$. This gives us the leading term in \eqref{eq:L}  and
\begin{equation} \label{eq:A_0-Phi_1}
A_0(s) =  \begin{pmatrix}
      0 & 1 & 0 \\
      \frac{\rho}{3} & 0 & 1 \\
      0  & \frac{\rho}{3} & 0
    \end{pmatrix}  + \left[ \mathsf{\Phi}_1, \begin{pmatrix}
      0 & 0 & 0 \\
      0 & 0 & 0 \\
      1 & 0 & 0
    \end{pmatrix}\right].
\end{equation}
On account of the symmetric relation \eqref{eq:L-sym}, we further observe that
\begin{equation}
 A_0(s) = -\diag(1,-1,1) A_0(s) \diag(1,-1,1),
\end{equation}
which implies
\begin{equation} \label{eq:A0-entries}
\left(A_0(s)\right)_{11} = \left(A_0(s)\right)_{13}=\left(A_0(s)\right)_{22}=\left(A_0(s)\right)_{31}=\left(A_0(s)\right)_{33}=0.
\end{equation}
Let the functions $p_0(s)$ and $q_0(s)$ be defined as
\begin{equation}\label{def: p-0}
p_0(s)=\frac{1}{\sqrt{2}}\left( \frac{\rho}{3} +(\mathsf{\Phi}_1)_{23} \right), \qquad q_0(s)=\frac{1}{\sqrt{2}} \left( \frac{\rho}{3}-(\mathsf{\Phi}_1)_{12} \right),
\end{equation}
we then obtain the expression of $A_0(s)$ in \eqref{eq:A-k} by combining the above two formulas and \eqref{eq:A_0-Phi_1}. If $z\to s$, it is easily seen from \eqref{eq:X-near-s} that $L(z;s)$ has a simple pole at $z=s$, which comes from the term $\ln(z-s)$. Thus,
$$L(z;s)\sim \frac{A_1(s)}{z-s},\qquad z \to s$$
with
 \begin{equation}\label{eq: L-expand-s}
 A_1(s) = -\frac{\gamma}{2\pi i}\Phi_0^{(0)}(s) \begin{pmatrix} 0&1&1\\ 0&0&0 \\ 0&0&0 \end{pmatrix} \Phi_0^{(0)}(s)^{-1},
 \end{equation}
where $\Phi_0^{(0)}(s)$ is given in \eqref{eq: Phi-expand-s}.
By setting
 \begin{equation}\label{def: q-k}
 \begin{pmatrix}
      q_1(s) \\
      q_2(s) \\
    q_3(s)
    \end{pmatrix}=\Phi_0^{(0)}(s) \begin{pmatrix}
      1\\
    0 \\
   0
    \end{pmatrix} \quad \textrm{and} \quad \begin{pmatrix}
      p_1(s) \\
      p_2(s) \\
    p_3(s)
    \end{pmatrix}=-\frac{\gamma}{2\pi i}\Phi_0^{(0)}(s)^{-t}\begin{pmatrix}
      0\\
    1 \\
   1
    \end{pmatrix},
\end{equation}
we obtain the expression of $A_1(s)$ in \eqref{eq:A-k}. The formula for $A_2(s)$ in \eqref{eq:A-sym} then follows from the symmetry relation \eqref{eq:L-sym}. This finishes the computation of $L(z;s)$ in \eqref{eq:L}. It is also worthwhile to point out that \eqref{eq: L-expand-s} implies
 \begin{equation}\label{eq: const1}
\mathtt{Tr} A_1(s) = \sum_{k=1}^3q_k(s)p_k(s)=0.
\end{equation}

The computation of $U(z;s)$ is similar. It is easy to check that
\begin{equation} \label{eq: U-large-z}
U(z;s)=\Boh(z^{-1}), \qquad z \to \infty;\quad  U(z;s)\sim - \frac{A_1(s)}{z-s},\qquad z\to s,
\end{equation}
where $A_1(s)$ is given in \eqref{eq: L-expand-s}. This, together with \eqref{eq:U-sym}, gives us the expression of $U(z;s)$ in \eqref{eq:U}.

Next, we show the functions $p_i(s)$ and $q_i(s)$, $i=0,1,2,3$, in \eqref{eq:A-k} satisfy the equations \eqref{eq:SystemEq} and \eqref{eq:sum0}. By  \eqref{eq: const1}, we have already proved \eqref{eq:sum0}. To see other differential equations, we recall that the compatibility condition
$$\frac{\partial^2}{\partial z \partial s}\Phi(z;s)=\frac{\partial^2}{\partial s \partial z}\Phi(z;s)$$
for the Lax pair \eqref{eq:Lax pair} is the zero curvature relation
\begin{equation}\label{eq: comp-condition}
\frac{\partial }{\partial s}L(z;s)-\frac{\partial}{\partial z} U(z;s)=A_0'(s) + \frac{A_1'(s)}{z-s} + \frac{A_2'(s)}{z-s}=[U,L].
\end{equation}
Inserting \eqref{eq:L} and \eqref{eq:U} into the above formula, we obtain upon taking $z \to \infty$,
\begin{equation}\label{eq:d-A-0}
A_0'(s)=\left[A_2(s)-A_1(s),\begin{pmatrix}
      0 & 0 & 0 \\
      0 & 0 & 0 \\
      1 & 0 & 0
    \end{pmatrix} \right]=\begin{pmatrix}
      0 & 0 & 0 \\
      -2p_3(s)q_2(s) & 0 & 0 \\
      0 & 2p_2(s)q_1(s) & 0
    \end{pmatrix},
\end{equation}
which leads to
 \begin{equation}\label{eq:d-q-0}
\left\{
 \begin{array}{ll}
p_0'(s)=-\sqrt{2}p_3(s)q_2(s),\\
q_0'(s)=\sqrt{2}p_2(s)q_1(s).
\end{array}
\right.
\end{equation}

If we calculate the residue at $z = s$ on the both sides of \eqref{eq: comp-condition}, it is readily seen that
\begin{equation}\label{eq:d-A-1}
A_1'(s)=-[A_1(s),M(s)],
\end{equation}
where
\begin{align}\label{eq: M}
  M(s) & =\begin{pmatrix}
      0 & 0 & 0 \\
      0 & 0 & 0 \\
      s & 0 & 0
    \end{pmatrix}+ A_0(s)+\frac{1}{s}\left(A_2(s)-A_1(s) \right)\nonumber\\
    &=\begin{pmatrix}
      0 & 1-\frac{2p_2(s)q_ 1(s)}{s}& 0 \\
      \sqrt{2}p_0(s)-\frac{2p_1(s)q_2(s)}{s} & 0 & 1-\frac{2p_3(s)q_ 2(s)}{s} \\
      s &  \sqrt{2}q_0(s)-\frac{2p_2(s)q_3(s)}{s} & 0
    \end{pmatrix}.
\end{align}
Here, to use the symmetric relation \eqref{eq:A-sym} to simplify the subsequent calculations, we have replaced $[A_1,A_2]$ in the $s^{-1}$-term by $[A_1, A_2-A_1]$ with the fact $[A_1,A_1]=0$. To this end, we observe from
\eqref{eq:Lax pair}--\eqref{eq:U} that, on the one hand,
\begin{equation}
	\left(\frac{\partial }{\partial z}\Phi(z;s)+\frac{\partial}{\partial s}\Phi(z;s) \right)\Phi(z;s)^{-1}=L(z;s)+U(z;s)\sim M(s)+\frac{A_1(s)}{s}, \qquad z \to s.
\end{equation}
On the other hand, substituting \eqref{eq: Phi-expand-s} into the left-hand side of the above equation, it follows from a straightforward calculation that
\begin{equation}
\left(\frac{\partial }{\partial z}\Phi(z;s)+\frac{\partial}{\partial s}\Phi(z;s) \right)\Phi(z;s)^{-1} \sim \frac{\ud}{\ud s} \Phi_0^{(0)}(s) \cdot  \Phi_0^{(0)}(s) ^{-1}, \qquad  z \to s.
\end{equation}
Thus, we have from the above two formulas that
\begin{equation}
\frac{\ud}{\ud s}\Phi_0^{(0)}(s) = \left( M(s)+\frac{A_1(s)}{s} \right) \Phi_0^{(0)}(s).
\end{equation}
Recall the definition of $q_k(s)$, $k=1,2,3$, given in \eqref{def: q-k}, we take the first column in the above formula and obtain
\begin{equation}\label{eq:d-Phi-0}
	\begin{pmatrix}
  q_1'(s)  & q_2'(s) & q_3'(s)
\end{pmatrix}^{t}=M(s) \begin{pmatrix}
  q_1(s)  & q_2(s) & q_3(s)
\end{pmatrix}^{t}.
	\end{equation}
Here we have made use of the fact that
$$A_1(s) \begin{pmatrix}
  q_1(s)  & q_2(s) & q_3(s)
\end{pmatrix}^{t}  =  \begin{pmatrix}
  q_1(s)  & q_2(s) & q_3(s)
\end{pmatrix}^{t}  \left(\sum_{k=1}^3q_k(s)p_k(s) \right) = \begin{pmatrix}
  0  & 0 & 0
\end{pmatrix};$$
see \eqref{eq:A-k} and \eqref{eq: const1}. The equation \eqref{eq:d-Phi-0} then gives us
  \begin{equation}\label{eq:d-q}
\left\{
 \begin{array}{ll}
 q_1'(s)=q_2(s)-\frac{2}{s}p_2(s)q_ 1(s)q_2(s),\\
  q_2'(s)=\sqrt{2}p_0(s)q_1(s) +q_3(s)+\frac{2}{s}p_2(s)q_ 2(s)^2,\\
  q_3'(s)=sq_1(s)  + \sqrt{2}q_0(s)q_2(s)-\frac{2}{s}p_2(s)q_2(s)q_3(s) .
 \end{array}
\right.
\end{equation}

To show the last three equations in \eqref{eq:SystemEq}, we see from \eqref{eq:d-A-1} and \eqref{eq:A-k} that
\begin{align*}
A_1'(s)&=\begin{pmatrix}
      q_1'(s) \\
      q_2'(s) \\
    q_3'(s)
    \end{pmatrix} \begin{pmatrix}
      p_1(s) \\
      p_2(s) \\
    p_3(s)
    \end{pmatrix}^t + \begin{pmatrix}
      q_1(s) \\
      q_2(s) \\
    q_3(s)
    \end{pmatrix} \begin{pmatrix}
      p_1'(s) \\
      p_2'(s) \\
    p_3'(s)
    \end{pmatrix}^t
    =-[A_1(s),M(s)]
    \nonumber
    \\
    & = -\begin{pmatrix}
      q_1(s) \\
      q_2(s) \\
    q_3(s)
    \end{pmatrix} \begin{pmatrix}
      p_1(s) \\
      p_2(s) \\
    p_3(s)
    \end{pmatrix}^t M(s) + M(s) \begin{pmatrix}
      q_1(s) \\
      q_2(s) \\
    q_3(s)
    \end{pmatrix} \begin{pmatrix}
      p_1(s) \\
      p_2(s) \\
    p_3(s)
    \end{pmatrix}^t.
\end{align*}
A combination of this formula and \eqref{eq:d-Phi-0} yields
 \begin{equation}\label{eq:d-p-v}
\begin{pmatrix}
  p_1'(s) & p_2'(s) & p_3'(s)
\end{pmatrix} = - \begin{pmatrix}
  p_1(s) & p_2(s) & p_3(s)
\end{pmatrix}M(s),
	\end{equation}	
which is equivalent to the following differential equations
 \begin{equation}\label{eq:d-p}
\left\{
 \begin{array}{ll}
 p_1'(s)=-\sqrt{2}p_0(s)p_2(s) -sp_3(s)+\frac{2}{s}p_1(s)p_2(s)q_ 2(s),\\
  p_2'(s)=-\sqrt{2}p_3(s)q_0(s) -p_1(s)-\frac{2}{s}p_2(s)^2q_ 2(s),\\
  p_3'(s)=-p_2(s)+\frac{2}{s}p_2(s)p_ 3(s)q_2(s).
 \end{array}
\right.
\end{equation}

Finally, we derive the coupled differential system \eqref{eq:eqforp-0} and \eqref{eq:p0q0-1}. For that purpose, we need to express the functions $p_k(s)$ and $q_k(s)$, $k =1,2,3$, in terms of $p_0(s)$ and $q_0(s)$.
In view of \eqref{eq:A_0-Phi_1} and \eqref{eq:A0-entries}, we have $(\mathsf{\Phi}_1)_{13}=0$. Comparing  the $(1,3)$ entry of the $z^{-1}$-term on both sides of \eqref{eq:L-expand}, we get
\begin{equation}
  2p_3(s)q_1(s) = \frac{\rho}{3} + (\mathsf{\Phi}_1)_{12} - (\mathsf{\Phi}_1)_{23}.
\end{equation}
This, together with the definition of $p_0(s)$ and $q_0(s)$ in \eqref{def: p-0}, gives us
\begin{equation}\label{eq: const2}
p_3(s)q_1(s)=-\frac{1}{\sqrt{2}}\left( p_0(s)+q_0(s) -\frac{\rho}{\sqrt{2}} \right).
\end{equation}
Thus, by \eqref{eq:d-q-0} and \eqref{eq: const2},  we have
\begin{align}
  p_2(s)q_2(s) &=\frac{(p_2(s)q_1(s))(p_3(s)q_2(s))}{p_3(s)q_1(s)}=\frac{p_0'(s)q_0'(s)}{\sqrt{2}\left(p_0(s)+q_0(s)-\frac{\rho}{\sqrt{2}} \right)}, \label{eq:pq-2} \\
  p_2(s)q_3(s) &= \frac{(p_2(s)q_1(s))(p_3(s)q_3(s))}{p_3(s)q_1(s)}=-\frac{q_0'(s)p_3(s)q_3(s)}{p_0(s)+q_0(s)-\frac{\rho}{\sqrt{2}}}. \label{eq:p2q3}
\end{align}
Differentiating both sides of the first equation in \eqref{eq:d-q-0}, we obtain from \eqref{eq:d-q} and \eqref{eq:d-p} that
\begin{align*}
p_0''(s) &= -\sqrt{2} \left( p_3'(s) q_2(s) + p_3(s) q_2'(s) \right) \\
&= -2p_0(s)p_3(s)q_1(s) -\sqrt{2}p_3(s)q_3(s)+\sqrt{2}p_2(s)q_2(s)\left(1-\frac{4}{s}p_3(s)q_ 2(s) \right).
\end{align*}
Using \eqref{eq: const2}, \eqref{eq:pq-2} and \eqref{eq:d-q-0}, we are able to rewrite the right-hand side of the above formula in terms of $p_0(s)$ and $q_0(s)$, except for the $p_3(s)q_3(s)$ term:
 \begin{equation}\label{eq:dd-p-0}
p_0''(s)=\sqrt{2}p_0(s)\left(p_0(s)+q_0(s)-\frac{\rho}{\sqrt{2}} \right) - \sqrt{2}p_3(s)q_3(s)+\frac{p_0'(s)q_0'(s)\left(1+\frac{2\sqrt{2}}{s}p_0'(s) \right)}{p_0(s)+q_0(s)-\frac{\rho}{\sqrt{2}}}.
\end{equation}
By similar calculations, we also have
\begin{align}\label{eq:dd-q-0}
q_0''(s)&=-2q_0(s)p_3(s)q_1(s)-\sqrt{2}p_1(s)q_1(s)+\sqrt{2}p_2(s)q_2(s)\left(1-\frac{4}{s}p_2(s)q_1(s) \right), \nonumber \\
&=\sqrt{2}q_0(s)\left(p_0(s)+q_0(s)-\frac{\rho}{\sqrt{2}} \right) -\sqrt{2}p_1(s)q_1(s)+\frac{p_0'(s)q_0'(s)\left(1-\frac{2\sqrt{2}}{s}q_0'(s)\right)}{p_0(s)+q_0(s)-\frac{\rho}{\sqrt{2}}}.
\end{align}
Adding \eqref{eq:dd-q-0} to \eqref{eq:dd-p-0}, we then obtain \eqref{eq:p0q0-1} by the fact that  $p_2(s)q_2(s)=- p_1(s)q_1(s)-p_3(s)q_s(s) $ (see \eqref{eq: const1}) and \eqref{eq:pq-2}.

To show the equation \eqref{eq:eqforp-0}, we make use of \eqref{eq:d-q-0}, \eqref{eq:d-q},  \eqref{eq:d-p}, \eqref{eq: const2}  and \eqref{eq:p2q3} to get
\begin{align}
(p_3q_3)'(s)&=sp_3(s)q_1(s)+\sqrt{2}q_0(s)p_3(s)q_2(s)-p_2(s)q_3(s)
\nonumber \\
&=-\frac{s}{\sqrt{2}}\left( p_0(s)+q_0(s) -\frac{\rho}{\sqrt{2}} \right) -q_0(s)p_0'(s)+\frac{q_0'(s)p_3(s)q_3(s)}{p_0(s)+q_0(s)-\frac{\rho}{\sqrt{2}}}. \label{eq:d-pq-3}
\end{align}
Taking the derivative  on both sides of \eqref{eq:dd-p-0} and using the above formula, we obtain \eqref{eq:eqforp-0}.

This completes the proof of Proposition \ref{pro:Lax pair}.
\end{proof}

We note that the Hamiltonian defined in \eqref{pro:H} agrees with that derived from the general theory of Jimbo-Miwa-Ueno \cite{JM}. Indeed, from \cite[Equation (5.1)]{JM}, we have
\begin{align}\label{eq:H-F}
H(s)= -\frac{\gamma}{2\pi i} \mathtt{Tr}
\left(\Phi_1^{(0)}(s)\begin{pmatrix}
      0 & 1 & 1 \\
      0 & 0 & 0 \\
      0 & 0 & 0
    \end{pmatrix}\right)
= -\frac{\gamma}{2\pi i} \left[ \left(\Phi_1^{(0)} (s) \right)_{21}+ \left(\Phi_1^{(0)} (s) \right)_{31}
\right],
\end{align}
where $\Phi_1^{(0)} (s)$ given in \eqref{eq: Phi-expand-s} appears in the expansion of $\Phi(z)$ near $z =s$. In view of the first equation in the Lax pair \eqref{eq:Lax pair}, by sending $z\to s$, we obtain from  \eqref{eq:L} and \eqref{eq:X-near-s} that the $\Boh(1)$-term in the expansion yields
\begin{equation*}\label{eq:Phi-1}
\Phi_1^{(0)}(s)=\frac{\gamma}{2\pi i}\left[\Phi_1^{(0)}(s),\begin{pmatrix}
      0 & 1 & 1 \\
      0 & 0 & 0 \\
      0& 0 & 0
    \end{pmatrix} \right]+\Phi_0^{(0)}(s)^{-1}\left( \begin{pmatrix}
      0 & 0 & 0 \\
      0 & 0 & 0 \\
      s & 0 & 0
    \end{pmatrix}+ A_0(s)+\frac{1}{2s}A_2(s)\right) \Phi_0^{(0)}(s).
    \end{equation*}
Inserting the above formula into \eqref{eq:H-F} gives us
\begin{equation*}
H(s)=-\frac{\gamma}{2\pi i}  \begin{pmatrix}
0 & 1 & 1
\end{pmatrix} \Phi_0^{(0)}(s)^{-1}\left( \begin{pmatrix}
      0 & 0 & 0 \\
      0 & 0 & 0 \\
      s & 0 & 0
    \end{pmatrix}+ A_0(s)+\frac{1}{2s}A_2(s)\right) \Phi_0^{(0)}(s)\begin{pmatrix}
     1 \\
      0  \\
      0
    \end{pmatrix}.
\end{equation*}
On account of the expressions of $A_k(s)$, $k=0,2$, in \eqref{eq:A-k}--\eqref{eq:A-sym} and the definition of $p_k(s)$ and $q_k(s)$, $k=1,2,3$, in \eqref{def: q-k}, it is easily seen that
 \begin{align*}
H(s) & =  \begin{pmatrix}
p_1(s) \\ p_2(s) \\ p_3(s)
\end{pmatrix}^t \left( \begin{pmatrix}
      0 & 1 & 0 \\
      \sqrt{2}p_0(s) & 0 & 1 \\
      s & \sqrt{2}q_0(s)  & 0
    \end{pmatrix}+\frac{1}{2s}  \begin{pmatrix}
    q_1(s) \\
      -q_2(s)  \\
      q_3(s)
    \end{pmatrix}
    \begin{pmatrix}
    p_1(s) \\ -p_2(s) \\ p_3(s)
\end{pmatrix}^t \right) \begin{pmatrix}
    q_1(s) \\
      q_2(s)  \\
      q_3(s)
    \end{pmatrix}
    \\
     &= \sqrt{2}p_0(s) p_2(s)q_1(s)+\sqrt{2}p_3(s)q_0(s) q_2(s)+p_1(s)q_2(s) +p_2(s)q_3(s)+sp_3(s)q_1(s)
    \\
&~~~ +\frac{1}{2s}\left(p_1(s)q_1(s)-p_2(s)q_2(s)+p_3(s)q_3(s) \right)^2,
\end{align*}
which coincides with that defined in \eqref{pro:H}.

\subsection{Differential identities for the Hamiltonian}
For later use, we collect some remarkable differential identities for the Hamiltonian $H$ in what follows. In particular, some of these differential identities will be crucial in our the derivation of the large gap asymptotics for $F(s;\gamma,\rho)$, especially for the constant term.

\begin{proposition} \label{prop:H-diff}
With the Hamiltonian $H$ defined in \eqref{pro:H}, we have
\begin{align}
  \frac{\ud}{\ud s}H(s)&=p_3(s)q_1(s)-\frac{2}{s^2}p_2(s)^2q_2(s)^2  \nonumber \\
&= -\frac{1}{\sqrt{2}}\left( p_0(s)+q_0(s)-\frac{\rho}{\sqrt{2}} \right)-\frac{p_0'(s)^2q_0'(s)^2}{s^2\left(p_0(s)+q_0(s)-\frac{\rho}{\sqrt{2}}\right)^2}, \label{eq: dH-s}
\end{align}
and $H$ is related to the action differential $\sum_{k=0}^3p_kq_k'-H$ by
\begin{align} \label{eq: action-diff}
 & \sum_{k=0}^3p_k(s)q_k'(s)-H(s)
 \nonumber
 \\
 &=H(s) +\frac{1}{4}\frac{\ud}{\ud s}\left( 2p_0(s)q_0(s)+p_2(s)q_2(s)+2p_3(s)q_3(s)-3sH(s) \right).
\end{align}
Moreover, we also have the following differential identities with respect to the parameters $\gamma$ and $\rho$:
\begin{align}
  \frac{\partial}{\partial \gamma}\left(\sum_{k=0}^3p_k(s)q_k'(s)-H(s) \right)&=\frac{\ud}{\ud s}\left(\sum_{k=0}^3p_k(s)\frac{\partial}{\partial \gamma}q_k(s) \right), \label{eq: dH-gamma} \\
  \frac{\partial}{\partial \rho}\left(\sum_{k=0}^3p_k(s)q_k'(s)-H(s) \right)&=\frac{\ud}{\ud s}\left(\sum_{k=0}^3p_k(s)\frac{\partial}{\partial \rho}q_k(s) \right).  \label{eq: dH-rho}
\end{align}
\end{proposition}
	
\begin{proof}
The proofs of \eqref{eq: dH-s} and \eqref{eq: action-diff} are straightforward, where we have made use of \eqref{eq:SystemEq}, \eqref{eq: const2} and \eqref{eq:p2q3} in the derivation of \eqref{eq: dH-s}, and of \eqref{eq:SystemEq} and \eqref{eq:sum0} in the derivation of \eqref{eq: action-diff}. To see the differential identity with respect to the parameter $\gamma$, we have from \eqref{eq:H-sys} that
\begin{equation}
  \frac{\partial}{\partial \gamma} H(s) = \sum_{k=0}^3 \left( \frac{\partial H}{\partial p_k} \frac{\partial }{\partial \gamma}p_k(s) +  \frac{\partial H}{\partial q_k} \frac{\partial}{\partial \gamma}q_k(s) \right) = \sum_{k=0}^3  \left( q_k'(s) \frac{\partial}{\partial \gamma}p_k(s) - p_k'(s) \frac{\partial}{\partial \gamma}q_k(s) \right),
\end{equation}
which gives us \eqref{eq: dH-gamma}. The differential identity with respect to $\rho$ in \eqref{eq: dH-rho} can be proved in a similar manner, and we omit the details here.

This completes the proof of Proposition \ref{prop:H-diff}.
\end{proof}

\section{Asymptotic analysis of the RH problem for $\Phi(z;s)$ as $s \to +\infty$}\label{sec:AsyPhiinfty}
In this section, we shall perform a Deift-Zhou steepest descent analysis \cite{DZ93} to the RH problem for $\Phi$ as $s \to +\infty$. It consists of a series of explicit and invertible transformations which leads to an RH problem tending to the identity matrix as $s \to +\infty$.
	
\subsection{First transformation: $\Phi \to T$}
This transformation is a rescaling and normalization of the RH problem for $\Phi$, which is defined by
\begin{equation}\label{def:PhiToT}
T(z)=  \diag \left(s^{\frac13},1,s^{-\frac13} \right) \Phi(sz)e^{-\Theta(sz)},
\end{equation}
where $\Theta(z)$ is given in \eqref{def:Theta}. Then, $T(z)$ satisfies the following RH problem.

\begin{proposition}\label{rhp:T}
The function $T$ defined in \eqref{def:PhiToT} has the following properties:
\begin{enumerate}
\item[\rm (1)] $T(z)$ is defined and analytic in $\mathbb{C}\setminus \{\cup^5_{j=0}\Sigma_j^{(1)}\cup[-1,1]\}$, where the contours $\Sigma_j^{(1)}$ are defined in \eqref{def:sigmais} with $s=1$.

\item[\rm (2)] $T(z)$ satisfies the jump condition
\begin{equation}\label{eq:T-jump}
 T_+(z)=T_-(z)J_T(z), \qquad z\in \cup^5_{j=0}\Sigma_j^{(1)}\cup(-1,0)\cup(0,1),
\end{equation}
where
\begin{equation}\label{def:JT}
J_T(z):=\left\{
 \begin{array}{ll}
          \begin{pmatrix} 0 & 1 &0 \\ -1 & 0 &0 \\ 0&0&1 \end{pmatrix}, & \qquad \hbox{$z\in \Sigma_0^{(1)}$,} \\
          \begin{pmatrix} 1&0&0 \\ e^{\theta_2(sz)-\theta_1(sz)}&1&e^{\theta_2(sz)-\theta_3(sz)} \\ 0&0&1 \end{pmatrix},  & \qquad  \hbox{$z\in \Sigma_1^{(1)}$,} \\
          \begin{pmatrix} 1&0&0 \\ 0&1&0 \\ e^{\theta_3(sz)-\theta_1(sz)}&e^{\theta_3(sz)-\theta_2(zs)}&1 \end{pmatrix},  & \qquad \hbox{$z\in \Sigma_2^{(1)}$,} \\
          \begin{pmatrix} 0 &0&1 \\ 0&1&0 \\ -1&0&0 \end{pmatrix}, & \qquad  \hbox{$z\in \Sigma_3^{(1)}$,} \\
          \begin{pmatrix} 1&0&0 \\ 0&1&0 \\ e^{\theta_3(sz)-\theta_2(sz)}&-e^{\theta_3(sz)-\theta_1(sz)}&1 \end{pmatrix}, & \qquad  \hbox{$z\in \Sigma_4^{(1)}$,} \\
          \begin{pmatrix} 1&0&0 \\ e^{\theta_1(sz)-\theta_2(sz)}&1&-e^{\theta_1(sz)-\theta_3(sz)} \\ 0&0&1 \end{pmatrix}, & \qquad  \hbox{$z\in \Sigma_5^{(1)}$,} \\
          \begin{pmatrix} e^{\theta_2(sz)-\theta_1(sz)}&1-\gamma&(1-\gamma)e^{\theta_2(sz)-\theta_3(sz)}\\ 0&e^{\theta_1(sz)-\theta_2(sz)}&0 \\ 0&0&1 \end{pmatrix}, & \qquad  \hbox{$z\in  (0,1)$,}\\
           \begin{pmatrix} e^{\theta_{3,+}(sz)-\theta_{3,-}(sz)}&(1-\gamma) e^{\theta_{2,-}(sz)-\theta_{2,+}(sz)}&1-\gamma\\ 0&1&0 \\ 0&0& e^{\theta_{3,-}(sz)-\theta_{3,+}(sz)}
           \end{pmatrix}, & \qquad  \hbox{$z\in(-1,0)$,}
        \end{array}
      \right.
      \end{equation}
      with $\theta_k(z)=\theta_k(z;\rho)$, $k=1,2,3$, being defined in \eqref{eq: theta-k-def}.
\item[\rm (3)]As $z \to \infty$ and $\pm \Im z>0$, we have
\begin{equation}\label{eq:asyT}
T(z)=
\left(I+ \frac{\mathsf{T}_1}{z} +\mathcal \Boh(z^{-2}) \right)\diag \left(z^{-\frac13},1,z^{\frac13} \right)L_{\pm},
\end{equation}
where
\begin{equation}
\mathsf{T}_1=\frac{\diag \left(s^{\frac13},1,s^{-\frac13} \right) \mathsf{\Phi}_1 \diag \left(s^{-\frac13},1,s^{\frac13} \right)}{s},
\end{equation}
with $\mathsf{\Phi}_1$ given in \eqref{X1-formula}.
\item[\rm (4)]
As $z \to \pm 1$, we have $T(z)=\Boh(\ln(z\mp 1))$.
\end{enumerate}
\end{proposition}
\begin{proof}
We only need to check the jump on $(-1,0)$, while the other claims follow directly from \eqref{def:PhiToT} and the RH problem for $\Phi$ given in Proposition \ref{rhp:X}.

If $z\in(-1,0)$, it is readily seen from  \eqref{def:PhiToT}, \eqref{eq:X-jump} and \eqref{def:JX} that
\begin{align}\label{eq:JT}
J_T(z)&=T_-(z)^{-1}T_+(z)
\nonumber \\
&=\diag\left(e^{\theta_{2,-}(sz)},e^{\theta_{1,-}(sz)},e^{\theta_{3,-}(sz)}\right)\Phi_-(sz)^{-1}\Phi_+(sz)
\nonumber \\
& \quad \times \diag\left(e^{-\theta_{1,+}(sz)},e^{-\theta_{2,+}(sz)},e^{-\theta_{3,+}(sz)}\right)
\nonumber \\
&=\begin{pmatrix} e^{\theta_{2,-}(sz)-\theta_{1,+}(sz)}&(1-\gamma) e^{\theta_{2,-}(sz)-\theta_{2,+}(sz)} & (1-\gamma)e^{\theta_{2,-}(sz)-\theta_{3,+}(sz)}
\\ 0& e^{\theta_{1,-}(sz)-\theta_{2,+}(sz)} &0
\\ 0&0& e^{\theta_{3,-}(sz)-\theta_{3,+}(sz)}
\end{pmatrix}.
\end{align}	
To this end, we observe from \eqref{eq: theta-k-def} that, if $z<0$,
\begin{align}\label{eq:thetarelations}
\theta_{3,+}(z)=\theta_{2,-}(z),\quad \theta_{1,+}(z)=\theta_{3,-}(z), \quad \theta_{1,-}(z)=\theta_{2,+}(z).
\end{align}
This, together with \eqref{eq:JT}, gives us the formula of $J_T$ on $(-1,0)$ as shown in \eqref{def:JT}.

This completes the proof of Proposition \ref{rhp:T}.
\end{proof}

\subsection{Second transformation: $T \to S$}
As $s\to +\infty$, it comes out that most of the entries involving $\theta_k(sz)$ in $J_T$ tend to zero exponentially fast, except some entries
when restricted on $(-1,0)\cup (0,1)$. More precisely, note that
\begin{equation}
\theta_2(sz)-\theta_1(sz)=\sqrt{3}i s^{\frac43} \left( \frac{3}{4}z^{\frac43}-\frac{\rho}{2s^{\frac23}}z^{\frac23} \right), \qquad z\in(0,1),
\end{equation}
thus, the $(1,1)$ and $(2,2)$ entries of $J_T(z)$ is highly oscillatory for large positive $s$ and $z\in(0,1)$. Similarly, since
\begin{equation}
\theta_{3,+}(sz)-\theta_{3,-}(sz)=-\sqrt{3}i s^{\frac43} \left(\frac{3}{4}|z|^{\frac43}-\frac{\rho}{2s^{\frac23}}|z|^{\frac23} \right), \qquad z\in(-1,0),
\end{equation}
we have that the $(1,1)$ and $(3,3)$ entries of $J_T(z)$ is highly oscillatory for large positive $s$ and $z\in(-1,0)$.

Following the spirit of steepest descent analysis, the second transformation involves the so-called lens opening. The goal of this step is to convert the highly oscillatory jumps into constant jumps on the original contours while creating extra jumps tending to the identity matrices exponentially fast for large positive $s$ on the new contours. This transformation is based on the following factorizations:
\begin{align}\label{eq:Deformation-1}
&  \begin{pmatrix} e^{\theta_2(sz)-\theta_1(sz)}&1-\gamma&(1-\gamma)e^{\theta_2(sz)-\theta_3(sz)}\\ 0&e^{\theta_1(sz)-\theta_2(sz)}&0 \\ 0&0&1 \end{pmatrix}
\nonumber \\
&=
\begin{pmatrix} 1&0&0\\ \frac{ e^{\theta_1(sz)-\theta_2(sz)}}{1-\gamma}&1&-e^{\theta_1(sz)-\theta_3(sz)} \\ 0&0&1 \end{pmatrix}
\begin{pmatrix}
0 &1-\gamma&0
\\
\frac{1}{\gamma-1} & 0  & 0
\\ 0&0&1
\end{pmatrix}
\nonumber
\\ &~~~\times
\begin{pmatrix} 1&0&0\\ \frac{e^{\theta_2(sz)-\theta_1(sz)}}{1-\gamma} &1&e^{\theta_2(sz)-\theta_3(sz)} \\ 0&0&1 \end{pmatrix}, \qquad z\in (0,1),
\end{align}
and
\begin{align}\label{eq:Deformation-2}
& \begin{pmatrix}
  e^{\theta_{3,+}(sz)-\theta_{3,-}(sz)}&(1-\gamma)e^{\theta_{2,-}(sz)-\theta_{2,+}(sz)} & 1-\gamma
  \\
  0&1&0
  \\
  0&0& e^{\theta_{3,-}(sz)-\theta_{3,+}(sz)}
  \end{pmatrix}
  \nonumber
  \\
&=
\begin{pmatrix} 1&0&0\\0&1&0 \\
\frac{e^{\theta_{3,-}(sz)-\theta_{2,-}(sz)}}{1-\gamma}&-e^{\theta_{3,-}(sz)-\theta_{1,-}(sz)}&1
\end{pmatrix} \begin{pmatrix} 0 &0&1-\gamma \\ 0&1&0 \\ \frac{1}{\gamma-1}&0&0 \end{pmatrix}
\nonumber
\\ &~~~\times
     \begin{pmatrix} 1&0&0\\0&1&0 \\  \frac{e^{\theta_{3,+}(sz)-\theta_{1,+}(sz)}}{1-\gamma} &e^{\theta_{3,+}(sz)-\theta_{2,+}(sz)}&1
     \end{pmatrix}, \qquad z\in(-1,0),
\end{align}
where we have made use of relations \eqref{eq:thetarelations} in  \eqref{eq:Deformation-2}.

We now set simply connected domains (the lenses) $\Omega_{1,\pm} (\Omega_{-1,\pm})$  on the $\pm$-side of $(0,1)$ ($(-1,0)$), with oriented boundaries $\partial \Omega_{1,\pm} \cup (0,1)$ ($\partial \Omega_{-1,\pm} \cup (-1,0)$) as shown in Figure \ref{fig:S}.

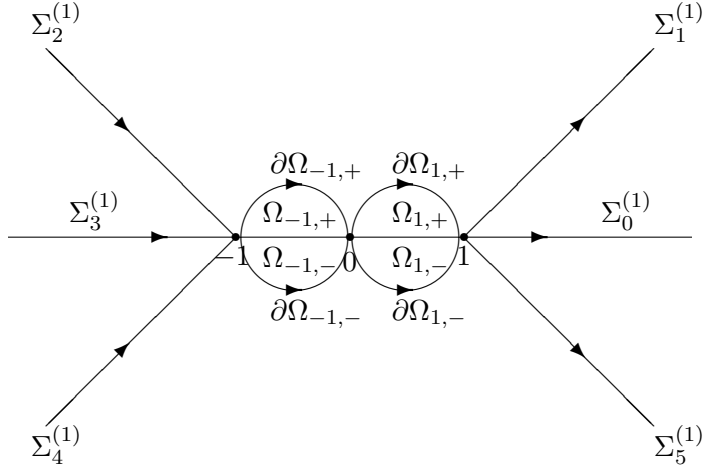
\begin{figure}[h]
\begin{center}
   \setlength{\unitlength}{1truemm}
   \begin{picture}(100,70)(-5,2)
       \put(25,40){\line(-1,0){30}}
       \put(55,40){\line(1,0){30}}

       \put(25,40){\line(1,0){30}}

       \put(25,40){\line(-1,-1){25}}
       \put(25,40){\line(-1,1){25}}
       \put(55,40){\line(1,1){25}}
       \put(55,40){\line(1,-1){25}}

       \put(15,40){\thicklines\vector(1,0){1}}
       \put(65,40){\thicklines\vector(1,0){1}}

       \put(10,55){\thicklines\vector(1,-1){1}}
       \put(10,25){\thicklines\vector(1,1){1}}
       \put(70,25){\thicklines\vector(1,-1){1}}
       \put(70,55){\thicklines\vector(1,1){1}}

       \put(39,35.5){$0$}
       \put(-2,11){$\Sigma_4^{(1)}$}

       \put(-2,67){$\Sigma_2^{(1)}$}
       \put(3,42){$\Sigma_3^{(1)}$}
       \put(80,11){$\Sigma_5^{(1)}$}
       \put(80,67){$\Sigma_1^{(1)}$}
       \put(73,42){$\Sigma_0^{(1)}$}


\put(32.7,40){\circle{18}}
\put(47.3,40){\circle{18}}

\put(45.5,42){$\Omega_{1,+}$}
\put(45.5,36){$\Omega_{1,-}$}

\put(28.5,42){$\Omega_{-1,+}$}
\put(28.5,36){$\Omega_{-1,-}$}
 \put(32.7,47){\thicklines\vector(1,0){1}}
 \put(47.3,47){\thicklines\vector(1,0){1}}
 \put(32.7,33){\thicklines\vector(1,0){1}}
  \put(47.3,33){\thicklines\vector(1,0){1}}
   \put(45.5,48.5){$\partial \Omega_{1,+}$}
    \put(45.5,29){$\partial \Omega_{1,-}$}
      \put(29.5,48.5){$\partial \Omega_{-1,+}$}
      \put(29.5,29){$\partial \Omega_{-1,-}$}

  \put(40,40){\thicklines\circle*{1}}

   \put(25,40){\thicklines\circle*{1}}
       \put(55,40){\thicklines\circle*{1}}

       \put(22,36.3){$-1$}
       \put(54,36.3){$1$}

   \end{picture}
   \caption{The lenses used for the transformation $T \to S$.}
   \label{fig:S}
\end{center}
\end{figure}

Based on the decompositions of $J_S$ given in \eqref{eq:Deformation-1}, \eqref{eq:Deformation-2} and also on the lenses just defined,
the second transformation reads
\begin{equation}\label{def:TtoS}
S(z)=T(z) \left\{
\begin{array}{ll}
\begin{pmatrix} 1&0&0\\ \frac{ e^{\theta_1(sz)-\theta_2(sz)}}{1-\gamma}&1&-e^{\theta_1(sz)-\theta_3(sz)} \\ 0&0&1 \end{pmatrix},  & \qquad z\in \Omega_{1,-},\\
  \begin{pmatrix} 1&0&0\\ \frac{e^{\theta_2(sz)-\theta_1(sz)}}{\gamma-1} &1&-e^{\theta_2(sz)-\theta_3(sz)} \\ 0&0&1 \end{pmatrix}, & \qquad z\in\Omega_{1,+},\\
   \begin{pmatrix} 1&0&0\\0&1&0 \\  \frac{e^{\theta_{3}(sz)-\theta_{1}(sz)}}{\gamma-1} &-e^{\theta_3(sz)-\theta_2(sz)}&1 \end{pmatrix}, & \qquad z\in\Omega_{-1,+},\\
\begin{pmatrix} 1&0&0\\0&1&0 \\ \frac{e^{\theta_{3}(sz)-\theta_{2}(sz)}}{1-\gamma}&-e^{\theta_3(sz)-\theta_1(sz)}&1 \end{pmatrix}, & \qquad z\in\Omega_{-1,-}, \\
I, & \qquad \mbox{$z$ outside the lenses}.
\end{array}
\right.
\end{equation}

It is then straightforward to check that $S(z)$ satisfies the following RH problem.

\begin{rhp}\label{rhp:S}
The function $S$ defined in \eqref{def:TtoS} has the following properties:
\begin{enumerate}
\item[\rm (1)] $S(z)$ is defined and analytic in $\mathbb{C} \setminus \Sigma_S$,
where
\begin{equation}\label{def:SigmaS}
\Sigma_S:=\cup^5_{j=0}\Sigma_j^{(1)}\cup[-1,1] \cup \partial \Omega_{1,\pm} \cup \partial \Omega_{-1,\pm}
\end{equation}

\item[\rm (2)] $S(z)$ satisfies the jump condition
\begin{equation}\label{eq:S-jump}
 S_+(z)=S_-(z)J_S(z), \qquad z\in \Sigma_S,
\end{equation}
where
\begin{equation}\label{def:JS}
J_S(z):=\left\{
 \begin{array}{ll}
           J_T(z), & \qquad  \hbox{$z\in \cup^5_{j=0}\Sigma_j^{(1)}$,} \\
    \begin{pmatrix} 1&0&0\\ \frac{ e^{\theta_1(sz)-\theta_2(sz)}}{1-\gamma}&1&-e^{\theta_1(sz)-\theta_3(sz)}  \\ 0&0&1 \end{pmatrix}, & \qquad  \hbox{$z\in \partial\Omega_{1,-}$,} \\
    \begin{pmatrix} 1&0&0\\ \frac{e^{\theta_2(sz)-\theta_1(sz)}}{1-\gamma} &1&e^{\theta_2(sz)-\theta_3(sz)}  \\ 0&0&1 \end{pmatrix}, & \qquad  \hbox{$z\in \partial\Omega_{1,+}$,} \\
    \begin{pmatrix} 1&0&0\\0&1&0 \\  \frac{e^{\theta_{3}(sz)-\theta_{1}(sz)}}{1-\gamma} &e^{\theta_3(sz)-\theta_2(sz)}&1 \end{pmatrix},& \qquad  \hbox{$z\in \partial\Omega_{-1,+}$,} \\
    \begin{pmatrix} 1&0&0\\0&1&0 \\ \frac{e^{\theta_{3}(sz)-\theta_{2}(sz)}}{1-\gamma}&-e^{\theta_3(sz)-\theta_1(sz)}&1 \end{pmatrix}& \qquad  \hbox{$z\in \partial\Omega_{-1,-}$,} \\
 \begin{pmatrix} 0&1-\gamma&0\\ \frac{1}{\gamma-1}&0&0 \\ 0&0&1 \end{pmatrix}, & \qquad  \hbox{$z\in  (0,1)$,}\\
           \begin{pmatrix} 0&0&1-\gamma\\ 0&1&0 \\ \frac{1}{\gamma-1}&0& 0\end{pmatrix}, & \qquad  \hbox{$z\in(-1,0)$,}
        \end{array}
      \right.
 \end{equation}
and where $J_T$ is defined in \eqref{def:JT}.
\item[\rm (3)]As $z \to \infty$ and $\pm \Im z>0$, we have
\begin{equation}\label{eq:asyS}
S(z)=
\left(I+ \frac{\mathsf{T}_1}{z} +\mathcal \Boh(z^{-2}) \right)\diag \left(z^{-\frac13},1,z^{\frac13} \right)L_{\pm},
\end{equation}
for some function $\mathsf{T}_1$.
\item[\rm (4)] As $z \to \pm 1$, we have $S(z)=\Boh(\ln(z\mp 1))$.
\end{enumerate}
\end{rhp}

\subsection{Global parametrix}
It is now easily seen that all the jump matrices of $S$ tend to the identity matrices exponentially fast as $s\to +\infty$, except the ones along the real axis.  We are then led to consider the following RH problem for the global parametrix $N$.

\begin{rhp}\label{rhp:N}
We look for a $3 \times 3$ matrix-valued function $N$ satisfying the following properties:
\begin{enumerate}
\item[\rm (1)] $N(z)$ is defined and analytic in $\mathbb{C}\setminus \mathbb{R}$.

\item[\rm (2)] $N(z)$ satisfies the jump condition
\begin{equation}\label{eq:N-jump}
 N_+(z)=N_-(z)J_N(z), \qquad z\in \mathbb{R},
\end{equation}
where
\begin{equation}\label{def:JN}
J_N(z)=J_S(z)=
\left\{
 \begin{array}{ll}
          \begin{pmatrix} 0 &0&1 \\ 0&1&0 \\ -1&0&0 \end{pmatrix}, & \qquad  \hbox{$z\in(-\infty,-1)$,} \\
          \begin{pmatrix} 0&0&1-\gamma\\ 0&1&0 \\ \frac{1}{\gamma-1}&0& 0\end{pmatrix}, & \qquad  \hbox{$z\in(-1,0)$,}\\
          \begin{pmatrix} 0&1-\gamma&0\\ \frac{1}{\gamma-1}&0&0 \\ 0&0&1 \end{pmatrix}, & \qquad  \hbox{$z\in (0,1)$,}\\
           \begin{pmatrix} 0 & 1 &0 \\ -1 & 0 &0 \\ 0&0&1 \end{pmatrix}, & \qquad \hbox{$z\in(1,+\infty)$.}
                  \end{array}
      \right.
 \end{equation}
\item[\rm (3)]As $z \to \infty$ and $\pm \Im z>0$, we have
\begin{equation}\label{eq:asyN}
N(z)=
\left(I+ \frac{\mathsf{N}_1}{z} +\mathcal \Boh(z^{-2})\right)\diag \left(z^{-\frac13},1,z^{\frac13} \right)L_{\pm},
\end{equation}
for some constant $\mathsf{N}_1$.
\end{enumerate}
\end{rhp}

The solution of this RH problem takes the following form:
\begin{equation}\label{eq:NSolution}
N(z)=C_N \diag \left(z^{-\frac13},1,z^{\frac13} \right)L_{\pm}\diag\left(d_1(z), d_2(z),d_3(z)\right),
\end{equation}
where the constant matrix $C_N$ and the scalar functions $d_k$, $k=1,2,3$, are to be determined.

We start with the following proposition concerning the properties of $d_k$.
\begin{proposition}\label{prop:dk}
The scalar functions $d_k(z), k=1,2,3$, in \eqref{eq:NSolution} are analytic in $\mathbb{C}\setminus\mathbb{R}$ and satisfy the following relations:
\begin{align}
d_{3,+}(x)&=d_{3,-}(x), \qquad  &&x\in(0,1)\cup(1+\infty),
\\
d_{1,+}(x)&=\frac{d_{2,-}(x)}{1-\gamma}, \qquad d_{2,+}(x)=d_{1,-}(x)(1-\gamma), \qquad &&x\in(0,1),
\\
d_{1,+}(x)&=d_{2,-}(x), \qquad d_{2,+}(x)=d_{1,-}(x),  \qquad  && x\in(1,+\infty),
\\
d_{2,+}(x)&=d_{2,-}(x),  \qquad  && x\in(-\infty,-1)\cup(-1,0), \\
d_{1,+}(x)&=\frac{d_{3,-}(x)}{1-\gamma}, \qquad d_{3,+}(x)=d_{1,-}(x)(1-\gamma), \qquad && x\in(-1,0),\\
d_{1,+}(x)&=d_{3,-}(x), \qquad d_{3,+}(x)=d_{1,-}(x), \qquad  && x\in(-\infty,-1).
\end{align}
\end{proposition}
\begin{proof}
We will only show the relations for $x\in(0,1)$, since the other relations can be proved similarly.

By \eqref{eq:NSolution} and \eqref{def:JN}, it follows that, if $x\in(0,1)$,
\begin{equation*}
L_+\diag(d_{1,+}(x),d_{2,+}(x),d_{3,+}(x))
=L_-\diag(d_{1,-}(x),d_{2,-}(x),d_{3,-}(x))
\begin{pmatrix} 0&1-\gamma&0\\ \frac{1}{\gamma-1}&0&0 \\ 0&0&1 \end{pmatrix}.
\end{equation*}
This, together with \eqref{def:Lpm}, implies that
\begin{align}
&\diag(d_{1,+}(x),d_{2,+}(x),d_{3,+}(x))
\nonumber
\\
&=L_+^{-1} L_-\diag(d_{1,-}(x),d_{2,-}(x),d_{3,-}(x))
\begin{pmatrix} 0&1-\gamma&0\\ \frac{1}{\gamma-1}&0&0 \\ 0&0&1 \end{pmatrix}
\nonumber
\\
& =\begin{pmatrix}
0 & -1 & 0
\\
1 & 0 & 0
\\ 0&0&1
\end{pmatrix}
\diag(d_{1,-}(x),d_{2,-}(x),d_{3,-}(x))
\begin{pmatrix} 0&1-\gamma&0\\ \frac{1}{\gamma-1}&0&0 \\ 0&0&1 \end{pmatrix}
\nonumber
\\
& =\diag\left(\frac{d_{2,-}(x)}{1-\gamma}, (1-\gamma)d_{1,-}(x),d_{3,-}(x)\right),
\end{align}
which gives us the relations on $(0,1)$.

This completes the proof of Proposition \ref{prop:dk}.
\end{proof}

To find the explicit expressions of $d_{k}$, $k=1,2,3$, we set
\begin{equation}\label{eq:dklambda}
\begin{array}{ll}
d_1(z)=\left\{
         \begin{array}{ll}
          \lambda(z^{\frac13}), & \quad \hbox{$\Im z>0$,} \\
           \lambda(\omega^{-1}z^{\frac13}), & \quad \hbox{$\Im z<0$,}
         \end{array}
       \right.
 \\
d_{2}(z)= \left\{\begin{array}{ll}
          \lambda(\omega^{-1} z^{\frac13}), & \quad \hbox{$\Im z>0$,} \\
          \lambda(z^{\frac13}), & \quad \hbox{$\Im z<0$,}
         \end{array}
       \right.
\\
d_{3}(z)=\lambda(\omega z^{\frac13}),
\end{array}
\end{equation}
for some function $\lambda$, where $-\pi <\arg z< \pi $ and recall that $\omega=e^{\frac{2\pi i}{3}}$. By further assuming that $d_k(z) \to 1 $, $k=1,2,3$, as $z\to \infty$, it is readily seen from \eqref{eq:dklambda} and Proposition \ref{prop:dk} that the function $\lambda$ solves the following scalar RH problem.
\begin{rhp}\label{rhp:lambda}
We look for a function $\lambda$ satisfying the following properties:
\begin{enumerate}
\item[\rm (1)] $\lambda(\xi)$ is defined and analytic in $\mathbb{C}\setminus\{(-1,1)\cup \omega^{-1} (-1, 1)\}$.

\item[\rm (2)] $\lambda(\xi)$ satisfies the jump condition
\begin{equation}\label{eq:lambdajump}
\lambda_+(\xi)=\lambda_{-}(\xi)\left\{
                                  \begin{array}{ll}
                                    1-\gamma, & \quad \hbox{$\xi  \in (-1,0)\cup e^{-\frac{2\pi i}{3}} (0, 1)$,} \\
                                    \frac{1}{1-\gamma}, &\quad  \hbox{$\xi  \in  e^{\frac{\pi i}{3}}(0,1)\cup (0, 1)$,}
                                  \end{array}
                                \right.
\end{equation}
where the orientations of the contours are shown in Figure \ref{fig:lambda}.

\item[\rm (3)]As $\xi \to \infty$, we have
$\lambda(\xi)\to 1$.
\end{enumerate}
\end{rhp}

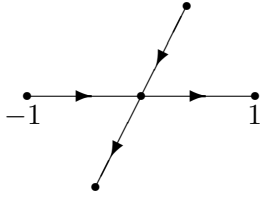
\begin{figure}[h]
\begin{center}
   \setlength{\unitlength}{1truemm}
   \vspace{-16mm}
   \begin{picture}(100,70)(-5,2)

       \put(40,40){\line(-1,-2){6}}
       \put(40,40){\line(1,2){6}}

       \put(25,40){\line(1,0){30}}

       \put(32.5,40){\thicklines\vector(1,0){1}}
       \put(47.5,40){\thicklines\vector(1,0){1}}
       \put(43,46){\thicklines\vector(-1,-2){1}}
       \put(37,34){\thicklines\vector(-1,-2){1}}

  \put(40,40){\thicklines\circle*{1}}
       \put(25,40){\thicklines\circle*{1}}
       \put(55,40){\thicklines\circle*{1}}
       \put(46,52){\thicklines\circle*{1}}
       \put(34,28){\thicklines\circle*{1}}

       \put(22,36.3){$-1$}
       \put(54,36.3){$1$}
   \end{picture}
   \vspace{-23mm}
   \caption{The contours for the RH problem for $\lambda$.}
   \label{fig:lambda}
\end{center}
\end{figure}

It is easily seen that the solution to the above RH problem is explicitly given by
\begin{equation}\label{eq:lambdaSolution}
\lambda(\xi)= \left(\frac{\xi-\omega^{-1}}{\xi+1}\right)^{\beta}\left(\frac{\xi+\omega^{-1}}{\xi-1}\right)^{\beta}= \left(\frac{\xi^2-\omega}{\xi^2-1}\right)^{\beta},\qquad \xi\in \mathbb{C}\setminus \{(-1,1)\cup \omega^{-1}(-1, 1)\},
\end{equation}
where $\beta$ is defined in \eqref{def:alpha} and the branch is chosen such that $\lambda(\xi) \to 1$ as $\xi \to \infty$.

A further effort finally gives us the following lemma.
\begin{lemma}\label{lem:N}
Let $\lambda$ be the function defined in \eqref{eq:lambdaSolution},  we have
\begin{align}\label{def:N}
N(z)=
C_N
\diag \left(z^{-\frac13},1,z^{\frac13} \right)
\left\{
                                                \begin{array}{ll}
                                                  L_+\diag\left(\lambda(z^{\frac13}), \lambda(\omega^{-1}z^{\frac13}),\lambda(\omega z^{\frac13})\right), & ~~\hbox{$\Im z>0$,} \\
                                                  L_-\diag\left(\lambda(\omega^{-1}z^{\frac13}), \lambda(z^{\frac13}),\lambda(\omega z^{\frac13})\right), & ~~\hbox{$\Im z <0$,}
                                                \end{array}
                                              \right.
\end{align}
solves the RH problem \ref{rhp:N} for $N$, where
\begin{equation}\label{eq:CN}
C_N=\begin{pmatrix}
1 & 0 & 0
\\
0 & 1 & 0
\\
\beta \omega (\omega-1) & 0 & 1
\end{pmatrix}= \begin{pmatrix}
1 & 0 & 0
\\
0 & 1 & 0
\\
- \sqrt{3} \beta i  & 0 & 1
\end{pmatrix},
\end{equation}
and
$\beta$ is defined in \eqref{def:alpha}. Moreover, the coefficient $\mathsf{N}_1$ in the expansion \eqref{eq:asyN} takes the following structure
\begin{equation}\label{eq:N1}
\mathsf{N}_1= \begin{pmatrix} 0 & \sqrt{3}\beta i &0 \\ * & 0 &\sqrt{3}\beta i \\ 0&*&0 \end{pmatrix},
\end{equation}
where $*$ denotes certain unimportant entry.
\end{lemma}
\begin{proof}
On account of our previous arguments, it remains to show \eqref{eq:CN}, which follows from the asymptotic condition \eqref{eq:asyN}. To that end, it is easily seen from \eqref{eq:lambdaSolution} that
\begin{equation}\label{eq:lambdainfty}
\lambda(\xi)=1+\frac{\beta(1-\omega)}{\xi^2}+\Boh(\xi^{-4}), \qquad \xi \to \infty.
\end{equation}
Thus, if $z\to \infty$ and $\Im z>0$, it follows that
\begin{align}
&\diag \left(z^{-\frac13},1,z^{\frac13} \right)L_{+}\diag\left(\lambda(z^{\frac13}), \lambda(\omega^{-1}z^{\frac13}),\lambda(\omega z^{\frac13})\right)
\nonumber
\\
&=\diag \left(z^{-\frac13},1,z^{\frac13} \right)L_{+}\left(I+\frac{\beta(1-\omega)}{z^{\frac23}}\diag(1,\omega^2,\omega^{-2})+\Boh(z^{-\frac43})\right). \label{eq:lambda-matrix-infinity}
\end{align}
Note that
\begin{align}
& \diag \left(z^{-\frac13},1,z^{\frac13} \right)L_{+}\diag(1,\omega^2,\omega^{-2})
\\
&=
\left(
\begin{pmatrix}
0 & 0 & 0
\\
0 & 0 & 0
\\
\omega & 0 & 0
\end{pmatrix}z^{\frac23}+\begin{pmatrix}
0 & \omega & 0
\\
0 & 0 & \omega
\\
0 & 0 & 0
\end{pmatrix}z^{-\frac13}\right)\diag \left(z^{-\frac13},1,z^{\frac13} \right)L_{+},
\end{align}
a combination of the above two formulas shows that
\begin{align}
&\diag \left(z^{-\frac13},1,z^{\frac13} \right)L_{+}\diag\left(\lambda(z^{\frac13}), \lambda(\omega^{-1}z^{\frac13}),\lambda(\omega z^{\frac13})\right)
\nonumber
\\
&=\left(\begin{pmatrix}
1 & 0 & 0
\\
0 & 1 & 0
\\
\beta \omega (1-\omega) & 0 & 1
\end{pmatrix}+\Boh(z^{-1})\right)\diag \left(z^{-\frac13},1,z^{\frac13} \right)L_{+}.
\end{align}
This, together with \eqref{eq:asyN} and \eqref{eq:dklambda}, gives us \eqref{eq:CN}. One gets the same result if we consider the asymptotics as $z\to \infty$ and $\Im z<0$.

Finally, the structure of $\mathsf{N}_1$ in \eqref{eq:N1} can be seen by expanding one more term in \eqref{eq:lambda-matrix-infinity} and a straightforward calculation.

This completes the proof of Lemma \ref{lem:N}.
\end{proof}

\subsection{Local parametrix near $1$}
Due to the fact that the convergence of the jump matrices to the identity matrices on $\Sigma_1^{(1)}$, $\Sigma_5^{(1)}$ and $\partial \Omega_{1,\pm}$ is not uniform near 1, we intend to find a function $P^{(1)}(z)$ satisfying an RH problem as follows:
\begin{rhp}\label{rhp:P-1}
We look for a $3 \times 3$ matrix-valued function  $P^{(1)}(z)$ satisfying
\begin{enumerate}
\item[\rm (1)]  $P^{(1)}(z)$ is defined and analytic in $D(1,\delta)\setminus \Sigma_S$, where $\Sigma_S$ is defined in \eqref{def:sigmais} and $D(z_0,\delta)$ stands for a fixed open disc centered at $z_0$ with radius $\delta>0$.

\item[\rm (2)]  $P^{(1)}(z)$ satisfies the jump condition
\begin{equation}\label{eq:P1-jump}
P^{(1)}_+(z)=P^{(1)}_-(z)J_S(z), \qquad z\in \Sigma_S \cap D(1,\delta),
\end{equation}
where $J_S(z)$ is given in \eqref{def:JS}.

\item[\rm (3)]
As $ s \to +\infty$, $P^{(1)}(z)$ matches $N(z)$ on the boundary $\partial D(1,\delta)$ of $D(1,\delta)$, i.e.,
\begin{equation}\label{eq:MatchingCondP1}
P^{(1)}(z)=\left(I+\Boh(s^{-\frac43}) \right)N(z),\qquad z\in \partial D(1,\delta).
\end{equation}
\end{enumerate}
\end{rhp}

To construct the local parametrix near $z=1$, we first introduce the following transformation to remove the $(2,3)$ entry of $J_S(z)$ in \eqref{def:JS}:
\begin{equation} \label{P1:P1-tilde}
P^{(1)}(z) = \widehat{P}^{(1)}(z)  \begin{cases}  \begin{pmatrix} 1& 0 &0 \\ 0 & 1&e^{\theta_2(sz)-\theta_3(sz)} \\ 0&0&1 \end{pmatrix}, &~~~ z \in \texttt{II} \cap D(1,\delta) \setminus \Omega_{1,+}, \\
\begin{pmatrix} 1& 0 &0 \\ 0 & 1&e^{\theta_1(sz)-\theta_3(sz)} \\ 0&0&1 \end{pmatrix}, & ~~~z \in \texttt{V} \cap D(1,\delta) \setminus \Omega_{1,-}, \\
I, & \textrm{otherwise},
\end{cases}
\end{equation}
where the regions $\texttt{II}$ and $\texttt{V}$ are illustrated in Figure \ref{fig:X}. Recall the definition of $\theta_k(z)$, $k=1,2,3$, in \eqref{eq: theta-k-def}, it is readily seen that $e^{\theta_2(sz)-\theta_3(sz)}$ and $e^{\theta_1(sz)-\theta_3(sz)}$ are analytic in $\mathbb{C} \setminus (-\infty, 0]$ and both of them are exponentially small as $s \to +\infty$, uniformly for $z \in D(1,\delta)$. Hence, $\widehat{P}^{(1)}(z)$ satisfies a similar RH problem as $ P^{(1)}(z)$, but with the jump condition \eqref{eq:P1-jump} replaced by
\begin{equation}\label{eq:hatP1-jump}
\widehat P^{(1)}_+(z)=\widehat P^{(1)}_-(z)\widehat J^{(1)}(z), \qquad z\in \Sigma_S \cap D(1,\delta),
\end{equation}
where
\begin{equation}\label{def:JP1}
\widehat J^{(1)}(z):=\left\{
 \begin{array}{ll}
          \begin{pmatrix} 0 & 1 &0 \\ -1 & 0 &0 \\ 0&0&1 \end{pmatrix}, & \qquad \hbox{$z\in \Sigma_0^{(1)}\cap D(1,\delta)$,} \\
          \begin{pmatrix} 1&0&0 \\ e^{\theta_2(sz)-\theta_1(sz)}&1&0\\ 0&0&1 \end{pmatrix},  & \qquad  \hbox{$z\in \Sigma_1^{(1)}\cap D(1,\delta)$,} \\
          \begin{pmatrix} 1&0&0 \\ e^{\theta_1(sz)-\theta_2(sz)}&1&0\\ 0&0&1 \end{pmatrix}, & \qquad  \hbox{$z\in \Sigma_5^{(1)}\cap D(1,\delta)$,} \\
           \begin{pmatrix} 1&0&0\\ \frac{ e^{\theta_1(sz)-\theta_2(sz)}}{1-\gamma}&1&0 \\ 0&0&1 \end{pmatrix}, & \qquad  \hbox{$z\in \partial \Omega_{1,-}\cap D(1,\delta)$,} \\
    \begin{pmatrix} 1&0&0\\ \frac{e^{\theta_2(sz)-\theta_1(sz)}}{1-\gamma} &1&0 \\ 0&0&1 \end{pmatrix}, & \qquad  \hbox{$z\in \partial \Omega_{1,+}\cap D(1,\delta)$,} \\
    \begin{pmatrix} 0&1-\gamma&0\\ \frac{1}{\gamma-1}&0&0 \\ 0&0&1 \end{pmatrix}, & \qquad  \hbox{$z\in  (1-\delta,1)$.}
        \end{array}
      \right.
 \end{equation}


We can construct $\widehat{P}^{(1)}(z)$ explicitly by using the confluent hypergeometric parametrix $\Phi^{(\CHF)}$ introduced in Appendix \ref{sec:CHF}. To see this, we define the following local conformal mapping near $z=1$:
\begin{align}\label{eq:f}
f(z) &= -is^{-\frac43}[(\theta_2(sz)-\theta_1(sz))-(\theta_2(s)-\theta_1(s))]
\nonumber \\
& = \frac{3\sqrt{3}}{4} (z^{\frac43}-1)-\frac{ \sqrt{3} \rho}{2s^{\frac23}}(z^{\frac23}-1), \qquad z\in D(1,\delta).
\end{align}
Obviously, we have
\begin{equation} \label{f1-f'1}
f(1)=0, \qquad   f'(1)=\sqrt{3}-\frac{\sqrt{3}}{3}\rho s^{-\frac23}, \qquad  f''(1) = \frac{\sqrt{3}}{3} +  \frac{\sqrt{3}}{9} \rho s^{-\frac23},
\end{equation}
where $f'(1)$ is positive when  $s$ is large enough. Let $\Phi^{(\CHF)}(z;\beta)$ be the confluent hypergeometric parametrix with $\beta$ given in \eqref{def:alpha} (see Appendix \ref{sec:CHF} below), we set, for $z\in D(1,\delta)\setminus \Sigma_S$,
\begin{align}\label{eq:P1Solution}
\widehat{P}^{(1)}(z)&=E_1(z) \diag \left( \Phi^{(\CHF)}(s^{\frac43}f(z);\beta)e^{-\frac\beta2 \pi i\sigma_3},1 \right)
\nonumber\\
&~~~\times\left\{
\begin{array}{ll}
\diag\left(e^{\frac{\theta_2(sz)-\theta_1(sz)}{2}\sigma_3},1\right), & \quad  \Im z>0,\\
\diag\left(e^{\frac{\theta_1(sz)-\theta_2(sz)}{2}\sigma_3},1\right), & \quad  \Im z<0,
\end{array}
\right.
\end{align}
where $\diag \left(A,1\right)$ stands for the $3 \times 3$ block diagonal matrix $\begin{pmatrix}(A)_{11} & (A)_{12} & 0
\\
(A)_{21} & (A)_{22} & 0
\\
0 & 0 & 1
\end{pmatrix}$
with $A$ being a $2\times 2$ matrix,
$\sigma_3=
\begin{pmatrix}
0 & 1
\\
-1 & 0
\end{pmatrix}$, $f(z)$ is defined in \eqref{eq:f} and
\begin{equation}\label{def:E1}
E_1(z)=\left\{
         \begin{array}{ll}
           N(z)\diag\left(\left(e^{\frac{\theta_1(s)-\theta_2(s)+\beta \pi i}{2}}s^{\frac{4\beta}{3}}f(z)^{\beta}\right)^{\sigma_3},1\right), & \quad \hbox{$\Im z>0$}, \\
           N(z)\diag\left(
\begin{pmatrix}
0 & 1
\\
-1 & 0
\end{pmatrix}
\left(e^{\frac{\theta_1(s)-\theta_2(s)+\beta \pi i}{2}}s^{\frac{4\beta}{3}}f(z)^{\beta}\right)^{\sigma_3},1 \right), & \quad \hbox{$\Im z<0$},
         \end{array}
       \right.
\end{equation}
with $N(z)$ given in \eqref{def:N}.

\begin{proposition}\label{prop:P1}
The local parametrix $P^{(1)}(z)$  defined in \eqref{P1:P1-tilde} and \eqref{eq:P1Solution} solves the RH problem \ref{rhp:P-1} for large positive $s$.
\end{proposition}
\begin{proof}
By \eqref{HJumps} and \eqref{def:alpha}, it is readily seen that $P^{(1)}(z)$ satisfies the jump condition \eqref{eq:P1-jump} if the prefactor $E_1(z)$ is analytic in $D(1,\delta)$.
In view of \eqref{def:E1} and \eqref{eq:N-jump}, it is immediate that $E_{1,+}(x)=E_{1,-}(x)$ for $x\in(1,1+\delta)$. If $x\in(1-\delta,1)$, we observe from \eqref{eq:f} that $f_+(x)^{\beta}=f_-(x)^{\beta}e^{2\beta \pi i}$ for large $s$. Thus, it follows from \eqref{eq:N-jump} and \eqref{def:alpha} that
\begin{align}
&N_+(x)\diag\left(f_+(x)^{\beta\sigma_3},1\right)
\nonumber
\\
&=N_-(x)\diag\left(\begin{pmatrix}
0 & e^{2\beta \pi i}
\\
-e^{-2\beta \pi i} & 0
\end{pmatrix}e^{2\beta \pi i \sigma_3}f_-(x)^{\beta\sigma_3},1\right)
\nonumber
\\
&= N_-(x)\diag\left(\begin{pmatrix}
0 & 1
\\
-1 & 0
\end{pmatrix}f_-(x)^{\beta\sigma_3},1\right),
\end{align}
which implies $E_{1,+}(x)=E_{1,-}(x)$ for $x\in(1-\delta,1)$. Hence, $E_1(z)$ is analytic in the punctured disk $D(1,\delta)\setminus \{1\}$. Recall the definition of $N$ in \eqref{def:N}, we note that, if $z\to 1$ with $\Im  z>0$,
\begin{equation}\label{eq:CNdiag}
C_N\diag\left(z^{-\frac13},1,z^{\frac13} \right)L_{+}=C_N L_{+}\left(I+\frac{i}{3\sqrt{3}}\begin{pmatrix} 0 &1&-1 \\ -1&0&-1 \\ 1&1&0 \end{pmatrix}(z-1)+\Boh((z-1)^2)\right),
\end{equation}
and
\begin{align}\label{eq:dk-1}
&\diag\left(\lambda(z^{\frac13}), \lambda(\omega^{-1}z^{\frac13}),\lambda(\omega z^{\frac13})\right)
\nonumber
\\
&=
\diag \left( \left(\frac{3(1-\omega)}{2(z-1)} \right)^{\beta},  \left(\frac
{2(z-1)}{3(1-\omega^2)} \right)^{\beta}, (1+\omega)^{\beta} \right)
\nonumber\\
&~~~ \times\left( I+ \diag \left( \frac{3\sqrt{3}+2 i}{6 \sqrt{3}}\beta, \frac{-3\sqrt{3} + 2i}{6 \sqrt{3}} \beta, -\frac{2\beta i}{3 \sqrt{3}}\right)(z-1) +\Boh((z-1)^2)\right).
\end{align}
Hence, a combination of the above two formulas, \eqref{def:N}, \eqref{def:E1} and \eqref{f1-f'1} shows that $E_1(z)=\Boh(1)$ if $z\to 1$ with $\Im  z>0$. The conclusion also holds if $z\to 1$ with $\Im z<0$
by similar arguments. As a consequence, the singularity of $E_1(z)$ at $1$ is removable, as required.

As for the matching condition \eqref{eq:MatchingCondP1}, one can check directly from \eqref{eq:f} and the large $z$ behavior of $\Phi^{(\CHF)}(z)$ given in \eqref{H at infinity}.

This completes the proof of Proposition \ref{prop:P1}.
\end{proof}

For later use, we collect the local behavior of $E_1(z)$ near $z = 1$ in the following proposition, which can be verified from a direct calculation with the aid of \eqref{def:E1}, \eqref{eq:CNdiag} and \eqref{eq:dk-1}.
\begin{proposition}
With the function $E_1(z)$ defined in \eqref{def:E1}, we have
\begin{equation}\label{eq:E-expand}
E_1(z)=E_1(1)\left(I+E_1(1)^{-1}E_1'(1)(z-1)+\Boh((z-1)^2) \right), \qquad  z \to 1,
\end{equation}
where
\begin{equation}\label{eq:E-expand-coeff-1}
E_1(1)=C_NL_{+}\diag\Big(c_1(s), c_1(s)^{-1}e^{-\frac{\beta\pi i }{3}}, e^{\frac{\beta \pi i }{3}} \Big),
\end{equation}
and
\begin{align}
& E_1(1)^{-1}E_1'(1) \nonumber \\
& = \frac{\sqrt{3} i}{9}\begin{pmatrix} -\frac{3\sqrt{3}\beta  i f''(1)}{2f'(1)} - 2 \beta \omega - \frac{\sqrt{3}\beta i}{2} &c_1(s)^{-2} e^{-\frac{\beta \pi i }{3}} &-c_1(s)^{-1} e^{\frac{\beta \pi i }{3}} \\
-c_1(s)^{2} e^{\frac{\beta \pi i }{3}} & \frac{3\sqrt{3}\beta i f''(1)}{2f'(1)}- 2\beta\omega^2 +\frac{\sqrt{3}\beta i}{2} &-c_1(s) e^{\frac{2\beta \pi i }{3}} \\
c_1(s) e^{-\frac{\beta \pi i }{3}} &c_1(s)^{-1} e^{-\frac{2 \beta \pi i }{3}}& -2\beta \end{pmatrix}, \label{eq:E-expand-coeff-2}
\end{align}
and where $L_+$ and $C_N$ are given in \eqref{def:Lpm} and \eqref{eq:CN}, respectively, and
\begin{equation}\label{eq:c-1}
c_1(s)=\left(\frac{3 \sqrt{3}}{2}\right)^{\beta} e^{\frac{\beta\pi i}{3}}s^{\frac{4\beta}{3}}f'(1)^{\beta}e^{\frac{\theta_1(s)-\theta_2(s)}{2}}.
\end{equation}
\end{proposition}

%


\subsection{Local parametrix near $-1$}
Near $z=-1$, we consider the following local parametrix.
\begin{rhp}\label{rhp:Pnegative1}
We look for a $3 \times 3$ matrix-valued function  $P^{(-1)}(z)$ satisfying
\begin{enumerate}
\item[\rm (1)]  $P^{(-1)}(z)$ is defined and analytic in $D(-1,\delta)\setminus \Sigma_S$.

\item[\rm (2)]  $P^{(-1)}(z)$ satisfies the jump condition
\begin{equation}\label{eq:Pneg1-jump}
P^{(-1)}_+(z)=P^{(-1)}_-(z)J_S(z), \qquad z\in \Sigma_S \cap D(-1,\delta),
\end{equation}
where $J_S(z)$ is given in \eqref{def:JS}.

\item[\rm (3)]
As $ s \to +\infty$, we have the matching condition
\begin{equation}\label{eq:MatchingCondPneg1}
P^{(-1)}(z)=\left(I+\Boh(s^{-\frac43}) \right)N(z),\qquad z\in \partial D(-1,\delta).
\end{equation}
\end{enumerate}
\end{rhp}

Analogously to the construction of $P^{(1)}(z)$ presented in the previous section, we first introduce the following transformation to remove the $(3,2)$ entry of $J_S(z)$ in \eqref{def:JS}:
\begin{equation} \label{P-1:P-1-tilde}
P^{(-1)}(z) = \widehat{P}^{(-1)}(z)  \begin{cases}  \begin{pmatrix} 1& 0 &0 \\ 0 & 1& 0 \\ 0& e^{\theta_3(sz)-\theta_2(sz)} &1 \end{pmatrix}, & z \in \texttt{II}\cap D(-1,\delta) \setminus \Omega_{-1,+}, \\
\begin{pmatrix} 1& 0 &0 \\ 0 & 1&0 \\ 0& e^{\theta_3(sz)-\theta_1(sz)} &1 \end{pmatrix}, & z \in \texttt{V}\cap D(-1,\delta) \setminus \Omega_{-1,-}, \\
I, & \textrm{otherwise}.
\end{cases}
\end{equation}
Then, $\widehat{P}^{(-1)}(z)$ satisfies an RH problem similar to that for $P^{(-1)}$, but with the jump condition \eqref{eq:Pneg1-jump} replaced by
\begin{equation}\label{eq:hatPneg1-jump}
\widehat{P}^{(-1)}_+(z)=\widehat{P}^{(-1)}_-(z)J^{(-1)}(z), \qquad z\in \Sigma_S \cap D(-1,\delta),
\end{equation}
where
\begin{equation}\label{def:JP-1}
J^{(-1)}(z):=\left\{
 \begin{array}{ll}
          \begin{pmatrix} 1 & 0 &0 \\ 0 & 1 &0 \\ e^{\theta_3(sz)-\theta_1(sz)} & 0 & 1 \end{pmatrix}, & \qquad \hbox{$z\in \Sigma_2^{(1)}\cap D(-1,\delta)$,} \\
          \begin{pmatrix} 0&0&1 \\ 0 & 1 & 0\\ -1&0&0 \end{pmatrix},  & \qquad  \hbox{$z\in \Sigma_3^{(1)}\cap D(-1,\delta)$,} \\
          \begin{pmatrix} 1&0&0 \\ 0 &1&0\\ e^{\theta_3(sz)-\theta_2(sz)}&0&1 \end{pmatrix}, & \qquad  \hbox{$z\in \Sigma_4^{(1)}\cap D(-1,\delta)$,} \\
           \begin{pmatrix}
1&0&0
\\
0 & 1 & 0
\\
\frac{ e^{\theta_3(sz)-\theta_2(sz)}}{1-\gamma}&0&1  \end{pmatrix}, & \qquad  \hbox{$z\in \partial \Omega_{-1,-}\cap D(-1,\delta)$,} \\
\begin{pmatrix} 1&0&0
\\ 0 &1&0 \\ \frac{e^{\theta_3(sz)-\theta_1(sz)}}{1-\gamma}&0&1 \end{pmatrix}, & \qquad  \hbox{$z\in \partial \Omega_{-1,+}\cap D(-1,\delta)$,} \\
    \begin{pmatrix} 0& 0 & 1-\gamma\\ 0&1&0 \\\frac{1}{\gamma-1} &0&0 \end{pmatrix}, & \qquad  \hbox{$z\in  (-1,-1+\delta)$.}
        \end{array}
      \right.
 \end{equation}

Again, one can construct $\widehat P^{(-1)}(z)$ by using the confluent hypergeometric parametrix $\Phi^{(\CHF)}$. More precisely, we define
\begin{equation}\label{def:tildef}
\widetilde f (z)=-i s^{-\frac43}\left\{
                   \begin{array}{ll}
                     \theta_3(sz)-\theta_1(sz)-( \theta_{3,+}(-s)-\theta_{1,-}(-s)), & ~~~\hbox{$\Im z>0$,} \\
                     \theta_2(sz)-\theta_3(sz)-( \theta_{2,-}(-s)-\theta_{3,-}(-s)), & ~~~\hbox{$\Im z<0$.}
                   \end{array}
                 \right.
\end{equation}
In view of the relations \eqref{eq:thetarelations}, it is easily seen that $\widetilde f (z)$ is analytic in $D(-1,\delta)$ with
$$\widetilde f(-1)=0,\qquad \widetilde f'(-1)=\sqrt{3}-\frac{\sqrt{3}}{3}\rho s^{-\frac23},$$
which is actually a local conformal mapping near $z=-1$ for large positive $s$. We then set, for $z\in D(-1,\delta)\setminus \Sigma_S$,
\begin{align}\label{eq:P-1Solution}
& \widehat{P}^{(-1)}(z)
\nonumber
\\
&=E_{-1}(z)
\begin{pmatrix} (\Phi^{(\CHF)}(s^{\frac43}\widetilde f(z); \beta))_{11} & 0 & (\Phi^{(\CHF)}(s^{\frac43}\widetilde f(z); \beta))_{12}
 \\0&1&0
\\(\Phi^{(\CHF)}(s^{\frac43}\widetilde f(z); \beta))_{21} & 0 & (\Phi^{(\CHF)}(s^{\frac43}\widetilde f(z); \beta))_{22}
\end{pmatrix}
\nonumber
\\
&~~~\times
\diag\left(e^{-\frac{\beta \pi i}{2}}, 1, e^{\frac{\beta \pi i}{2}}\right)
\left\{
\begin{array}{ll}
\diag\left(e^{\frac{\theta_3(sz)-\theta_1(sz)}{2}},1,e^{\frac{\theta_1(sz)-\theta_3(sz)}{2}}\right), & \quad  \Im z>0,\\
\diag\left(e^{\frac{\theta_3(sz)-\theta_2(sz)}{2}},1,e^{\frac{\theta_2(sz)-\theta_3(sz)}{2}}\right), & \quad  \Im z<0,
\end{array}
\right.
\end{align}
where $\beta$ is given in \eqref{def:alpha}, $\widetilde f(z)$ is defined in \eqref{def:tildef}, and
\begin{align}\label{def:E-1}
&E_{-1}(z)=N(z)
\\
&\times \begin{cases}
\diag\left(e^{\frac{\theta_{1,+}(-s)-\theta_{3,+}(-s)+\beta \pi i}{2}}s^{\frac{4\beta}{3}}\widetilde f(z)^{\beta},1,e^{\frac{\theta_{3,+}(-s)-\theta_{1,+}(-s)-\beta \pi i}{2}}s^{-\frac{4\beta}{3}}\widetilde f(z)^{-\beta}\right),  \\
 \hspace{13cm} \hbox{$\Im z>0$}, \\
\begin{pmatrix}
0 & 0 & 1
\\
0 & 1 & 0
\\
-1 & 0 & 0
\end{pmatrix}
\diag\left(e^{\frac{\theta_{3,+}(-s)-\theta_{2,+}(-s)+\beta \pi i}{2}}s^{\frac{4\beta}{3}}\widetilde f(z)^{\beta},1,e^{\frac{\theta_{2,+}(-s)-\theta_{3,+}(-s)-\beta \pi i}{2}}s^{-\frac{4\beta}{3}}\widetilde f(z)^{-\beta}\right),  \\  \hspace{13cm}  \hbox{$\Im z<0$},
\end{cases} \nonumber
\end{align}
with $N(z)$ given in \eqref{def:N}. As in the proof of Proposition \ref{prop:P1}, it is straightforward to show that $E_{-1}(z)$ is analytic in $D(-1,\delta)$, which leads to the following proposition.
\begin{proposition}\label{prop:P-1}
The local parametrix ${P}^{(-1)}(z)$  defined in \eqref{P-1:P-1-tilde} and \eqref{eq:P-1Solution} solves the RH problem \ref{rhp:Pnegative1} for large positive $s$.
\end{proposition}

\subsection{Local parametrix near the origin}

The local parametrix near the origin reads as follows.

\begin{rhp}\label{rhp:P-0}
We look for a $3\times 3$ matrix-valued function $P^{(0)}(z)$ satisfying
\begin{enumerate}
\item[\rm (1)]  $P^{(0)}(z)$ is defined and analytic in $D(0,\delta)\setminus \Sigma_S$.

\item[\rm (2)]  $P^{(0)}(z)$ satisfies the jump condition
\begin{equation}\label{eq:P0-jump}
P^{(0)}_+(z)=P^{(0)}_-(z)J_S(z), \qquad z\in D(0,\delta) \cap \Sigma_S,
\end{equation}
where $J_S(z)$ is defined in \eqref{def:JS}.
%
\item[\rm (3)]As $s\to \infty$, we have the matching condition
\begin{equation}\label{eq:MatchingCond0}
P^{(0)}(z)=\left(I+\Boh(s^{-\frac23}) \right) N(z), \qquad  z\in \partial D(0,\delta).
\end{equation}
\end{enumerate}
\end{rhp}

This local parametrix can be constructed in terms of the solution $\Psi(z)$ of the RH problem \ref{rhp: Pearcey}, i.e., the Pearcey parametrix, in the following way:
 \begin{equation}\label{eq:P0Solution}
P^{(0)}(z)=E_0(z)\Psi(sz)e^{-\Theta(sz)} \diag\left((1-\gamma)^{-\frac23},(1-\gamma)^{\frac13},(1-\gamma)^{\frac13}\right),
\end{equation}
where $\Theta(z)$ is defined in \eqref{def:Theta} and
\begin{multline}\label{eq:E0}
E_0(z)= -\sqrt{\frac{3}{2\pi}}e^{-\frac{\rho^2}{6}}i N(z) \diag\left((1-\gamma)^{\frac23},(1-\gamma)^{-\frac13},(1-\gamma)^{-\frac13}\right)
 \\ \times L_{\pm}^{-1} \diag((sz)^{\frac13}, 1,(sz)^{-\frac13}) \Psi_0^{-1},
\end{multline}
for $\pm \Im z>0$, with $L_{\pm}$ and $\Psi_0$ given in \eqref{def:Lpm} and \eqref{asyPsi:coeff}, respectively.

\begin{proposition}\label{prop:P0}
The local parametrix $P^{(0)}(z)$  defined in \eqref{eq:P0Solution} solves the RH problem \ref{rhp:P-0}.
\end{proposition}
\begin{proof}
It is straightforward to check the matching condition \eqref{eq:MatchingCond0}. Indeed, for $z\in \partial D(0,\delta)$, by inserting \eqref{eq:asyPsi} into \eqref{eq:P0Solution}, it follows that
\begin{equation}\label{eq:MatchingCondExpand}
P^{(0)}(z)N(z)^{-1}=\widehat{E}_0(z)\left(I+\frac{\widehat{\Psi}_1}{z s^{\frac23}}+ \Boh(s^{-\frac43}) \right)\widehat{E}_0(z)^{-1},
\end{equation}
where
\begin{equation}\label{eq:hatE0}
\widehat{E}_0(z)=N(z) \diag\{(1-\gamma)^{\frac23},(1-\gamma)^{-\frac13},(1-\gamma)^{-\frac13}\}L_{\pm}^{-1} \diag(z^{\frac13}, 1,z^{-\frac13}), \qquad \pm \Im z>0, 
\end{equation}
and
\begin{equation}\label{eq:hatPsi1}
\widehat{\Psi}_1= \begin{pmatrix}
0 & \kappa_3(\rho) & 0 \\
0 & 0 & \kappa_3(\rho) + \frac{\rho}{3} \\
0 &0  & 0
\end{pmatrix},
\end{equation}
with  $\kappa_3(\rho)$ defined in \eqref{poly:kappa3}. Since $\widehat{E}_0(z)$ is independent of $s$, the formula \eqref{eq:MatchingCondExpand} gives us \eqref{eq:MatchingCond0}.

To see the jump condition \eqref{eq:P0-jump}, it suffices to show that $E_0(z)$ is analytic in $D(0,\delta)$. To that end, we observe from \eqref{eq:E0} and \eqref{eq:N-jump} that, if $z>0$,
\begin{align*}
&E_{0,+}(z)E_{0,-}(z)^{-1}
\\
&=L_- \diag\left((1-\gamma)^{-\frac23},(1-\gamma)^{\frac13},(1-\gamma)^{\frac13}\right)
\begin{pmatrix}
0 & 1-\gamma & 0
\\
\frac{1}{\gamma-1} & 0 & 0
\\
0 & 0 & 1
\end{pmatrix}
\\
&\quad \times \diag\left((1-\gamma)^{\frac23},(1-\gamma)^{-\frac13},(1-\gamma)^{-\frac13}\right)L_+^{-1}
=L_-\begin{pmatrix}
0 & 1 & 0
\\
-1 & 0 & 0
\\
0 & 0 & 1
\end{pmatrix}L_+^{-1}=I,
\end{align*}
while for $z<0$,
\begin{align*}
&E_{0,+}(z)E_{0,-}(z)^{-1}
\\
&=\diag\left((sz)_-^{-\frac13}, 1,(sz)_-^{\frac13}\right)L_- \diag\left((1-\gamma)^{-\frac23},(1-\gamma)^{\frac13},(1-\gamma)^{\frac13}\right)
\begin{pmatrix}
0 & 0 & 1-\gamma
\\
0 & 1 & 0
\\
\frac{1}{\gamma-1} & 0 & 0
\end{pmatrix}
\\
&\quad \times \diag\left((1-\gamma)^{\frac23},(1-\gamma)^{-\frac13},(1-\gamma)^{-\frac13}\right)L_+^{-1}\diag\left((sz)_+^{\frac13}, 1,(sz)_+^{-\frac13}\right)
\\
&=\diag\left((sz)_-^{-\frac13}, 1,(sz)_-^{\frac13}\right)L_-
\begin{pmatrix}
0 & 0 & 1
\\
0 & 1 & 0
\\
-1 & 0 & 0
\end{pmatrix}
L_+^{-1}\diag\left((sz)_+^{\frac13}, 1,(sz)_+^{-\frac13}\right)
\\
&=\diag\left((sz)_-^{-\frac13}, 1,(sz)_-^{\frac13}\right)
\begin{pmatrix}
\omega^2 & 0 & 0
\\
0 & 1 & 0
\\
0 & 0 & \omega
\end{pmatrix}
\diag\left((sz)_+^{\frac13}, 1,(sz)_+^{-\frac13}\right)=I.
\end{align*}
Hence, $E_0(z)$ is analytic in the punctured disk $D(0,\delta)\setminus\{0\}$. In view of \eqref{def:N} and \eqref{eq:E0}, it easily seen that $E_0(z)=\Boh(z^{-\frac23})$, which implies that the singularity of $E_0(z)$ at the origin is removable, as expected.

This completes the proof of Proposition \ref{prop:P0}.
\end{proof}


\subsection{Final transformation}

The final transformation is defined by
\begin{equation}\label{def:StoR}
 R(z)=\left\{
                \begin{array}{ll}
 S(z)N(z)^{-1}, &\qquad  \hbox{$z\in \mathbb{C}\setminus \{D(-1,\delta) \cup D(0,\delta)\cup D(1,\delta)\cup  \Sigma_S\}$,} \\
        S(z) P^{(-1)}(z)^{-1},        &\qquad  \hbox{$z \in D(-1,\delta)$,} \\
        S(z)P^{(0)}(z)^{-1},  &\qquad  \hbox{$z \in D(0,\delta)$,}\\
        S(z)P^{(1)}(z)^{-1},  &\qquad  \hbox{$z \in D(1,\delta)$.}
                \end{array}
              \right.
\end{equation}
It is then easily seen that $R$ satisfies the following RH problem.
\begin{rhp}\label{rhp:R}
The $3\times 3$ matrix-valued function $R(z)$ defined in \eqref{def:StoR} has the following properties:
\begin{enumerate}
\item[\rm (1)]  $R(z)$ is analytic in $\mathbb{C} \setminus \Sigma_{R}$,
where
\begin{equation}
\Sigma_R:=\Sigma_S\cup \partial D(-1,\delta) \cup \partial D(0,\delta) \cup \partial D(1,\delta) \setminus \{\mathbb{R} \cup D(-1,\delta)
\cup D(0,\delta) \cup D(1,\delta) \};
\end{equation}
see Figure \ref{fig:ContourR} for an illustration.

\item[\rm (2)]  $R(z)$ satisfies the jump condition  $$ R_+(z)=R_-(z)J_R(z), \qquad z\in \Sigma_R,$$
where
  \begin{equation}\label{def:JR}
                     J_{R}(z)=\left\{
                                      \begin{array}{ll}
                                        P^{(-1)}(z) N(z)^{-1}, & \hbox{ $z \in \partial D(-1, \delta)$,} \\
                                        P^{(0)}(z) N(z)^{-1}, & \hbox{ $z \in \partial D(0, \delta)$,} \\
                                        P^{(1)}(z) N(z)^{-1}, & \hbox{ $z \in \partial D(1, \delta)$,} \\
                                        N(z) J_S(z) N(z)^{-1}, & \hbox{ $ z \in \Sigma_{R} \setminus \{ D(-1,\delta)
\cup D(0,\delta) \cup D(1,\delta) \}$.}
                                      \end{array}
\right.
  \end{equation}
\item[\rm (3)] As $z \to \infty$, we have
$$R(z)=I+\Boh(z^{-1}).$$
\end{enumerate}
\end{rhp}

\begin{figure}[h]
\centering
\includegraphics[width=9.2cm]{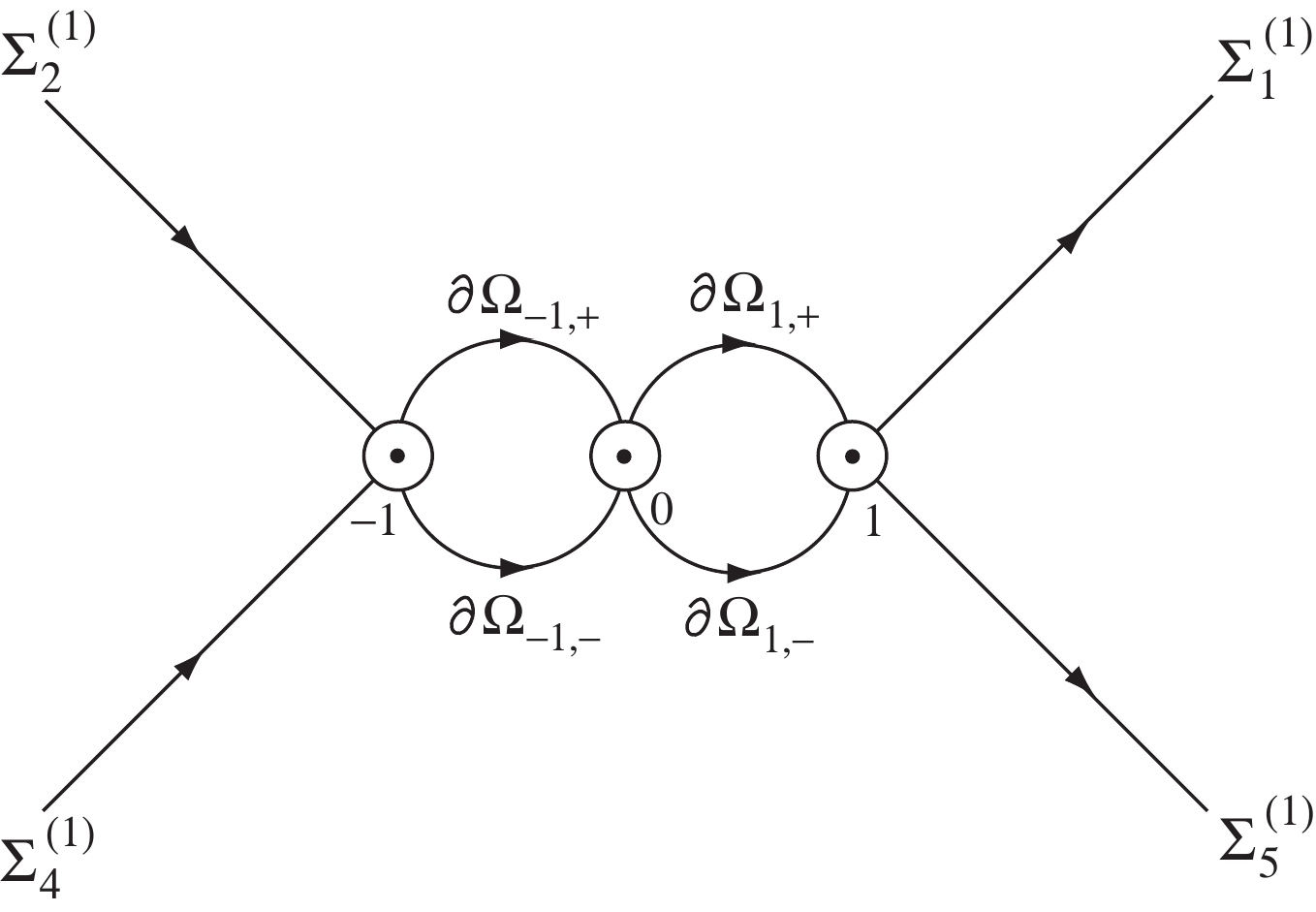}
\caption{Contour  $\Sigma_R$ for the RH problem \ref{rhp:R} for $R$. }
 \label{fig:ContourR}
\end{figure}

For $z\in \Sigma_{R} \setminus \{ D(-1,\delta)
\cup D(0,\delta) \cup D(1,\delta) \} $, it is readily seen from \eqref{def:JR}, \eqref{def:N} and \eqref{def:JS} that there exists some constant $c>0$ such that
\begin{equation}\label{eq:JRest1}
J_R(z)=I+O\left(e^{-c s^{\frac43}}\right), \qquad  s\to+\infty,
\end{equation}
uniformly valid for $\gamma$ and $\rho$ in any compact subset of $[0,1)$ and $\mathbb{R}$, respectively. For $z\in\partial D(0,\delta)$, we see from \eqref{def:JR} and \eqref{eq:MatchingCondExpand} that
\begin{equation}\label{eq:R-Jump-est}
J_{R}(z)=I+\frac{J_1(z)}{s^{\frac23}}+\Boh(s^{-\frac43}), \qquad s\to +\infty,
\end{equation}
where
\begin{equation}\label{def:J1}
J_1(z)= \frac{1}{z} \widehat E_0(z) \widehat \Psi_1 \widehat E_0(z)^{-1},
\end{equation}
with $\widehat E_0(z)$ and $\widehat \Psi_1$ given in \eqref{eq:hatE0} and \eqref{eq:hatPsi1}. Combining \eqref{eq:JRest1} and \eqref{eq:R-Jump-est} with the matching conditions \eqref{eq:MatchingCondP1} and \eqref{eq:MatchingCondPneg1} shows that $R(z)$ admits an asymptotic expansion of the following form
\begin{equation}\label{eq:R-est}
R(z)=I+\frac{R_1(z)}{s^{\frac23}}+\Boh(s^{-\frac43}), \qquad s \to +\infty;
\end{equation}
see the standard analysis in \cite{DeiftBook,DZ93}.  Moreover, from the RH problem \ref{rhp:R} for $R$, one can readily verify that $R_1(z)$ satisfies the following RH problem.
\begin{rhp}
The $3\times 3$ matrix-valued function $R_1(z)$ appearing in \eqref{eq:R-est} has the following properties:
\begin{enumerate}
\item[\rm (1)]  $R_1(z)$ is analytic in $\mathbb{C}\setminus \partial D(0, \delta)$.

\item[\rm (2)]   For $z \in \partial D(0, \delta)$, we have
\begin{equation}\label{eq:R1jump}
R_{1,+}(z)-R_{1,-}(z)=J_1(z),
\end{equation}
where $J_1(z)$ is given in \eqref{def:J1}.

\item[\rm (3)] As $z \to \infty$, we have $R_1(z)=\Boh(z^{-1})$.
\end{enumerate}
\end{rhp}
On account of \eqref{eq:E0}, \eqref{eq:hatE0} and the fact that $E_0(z)$ is analytic near the origin, we have that $\widehat E_0(z)$ is analytic at $z=0$ as well. Note that $\widehat{E}_0(0)=C_NC_N^t$, where $C_N$ is given in \eqref{eq:CN}, it is readily seen from \eqref{eq:R1jump}, \eqref{def:J1} and Cauchy's residue theorem that
\begin{align}\label{eq:R1}
  R_1(z) = \frac{1}{2\pi i} \oint_{\partial D(0, \delta)} \frac{J_1(\zeta)}{\zeta-z} \ud \zeta
=\left\{
    \begin{array}{ll}
       \frac{C_N \widehat{\Psi}_1C_N^{-1}}{z} - J_1(z), &~~~ \hbox{ $z\in D(0,\delta)$,} \vspace{3pt} \\
      \frac{C_N \widehat{\Psi}_1C_N^{-1}}{z} , &~~~ \hbox{ $|z|> \delta.$}
    \end{array}
  \right.
\end{align}
Moreover, the asymptotic expansion \eqref{eq:R-est} also gives us the following asymptotics for the derivatives of $R$ with respect to the parameter $\beta$:
\begin{equation} \label{eq:R-est-beta'}
\frac{\partial^k}{\partial \beta^k} R(z) =  \frac{\partial^k}{\partial \beta^k} R_1(z) s^{-\frac23} +\Boh((\ln s)^ks^{-\frac43}), \qquad s \to +\infty,\qquad k=1,2,\ldots,
\end{equation}
uniformly for $z \in \mathbb{C} \setminus \Sigma_R$; see similar analysis in \cite[Section 3.5]{CL21}.

\section{Asymptotic analysis  of the RH problem for $\Phi(z;s)$ as  $s\to 0^{+}$}\label{sec:AsyPhi0}

The asymptotic analysis for $\Phi(z;s)$ as  $s\to 0^{+}$ is much simpler than the case when $s \to +\infty$. As $s\to 0^+$, the interval  $(-s,s)$  shrinks to the origin. Comparing the RH problem for $\Phi(z;s)$ given in Proposition \ref{rhp:X} with the RH problem \ref{rhp: Pearcey} for $\Psi(z)$, it is easy to see that $\Phi(z;s)$ can be approximated by $\Psi(z)$ when $z$ is bounded away from $(-s,s)$ as $s\to 0^+$. In another word, the global parametrix  is given by
\begin{equation}
\frac{\Psi_0^{-1}}{\sqrt{\frac{2\pi}{3}} e^{\frac{\rho^2}{6}} i}\Psi(z) \begin{cases}
J_1, &  \arg z < \frac{\pi}{4} \textrm{ and } \arg (z -s) > \frac{\pi}{4}, \\
J_2, &  \arg z > \frac{3\pi}{4} \textrm{ and } \arg (z+s) < \frac{3\pi}{4},\\
J_4^{-1}, & \arg z < - \frac{3\pi}{4} \textrm{ and } \arg (z +s) > -\frac{3\pi}{4}, \\
J_5^{-1}, &  \arg z >- \frac{\pi}{4} \textrm{ and } \arg (z -s) <- \frac{\pi}{4},
\end{cases}
\end{equation}
for $|z|>\delta$, where the constant matrix $\Psi_0$ is given in \eqref{asyPsi:coeff} and $J_k$ denotes the constant jump matrix $J_\Phi(z)$ in \eqref{def:JX} restricted on the contour $\Sigma_k^{(s)}, k=0,1,\cdots,5$.

When $z$ lies in a neighbourhood of $(-s,s)$, i.e., $z\in D(0,\delta)$, we approximate $\Phi(z;s)$ by the following explicit function:
 \begin{equation}\label{eq: localPra-small-s-0}
P^{(s)}(z)=\frac{\Psi_0^{-1}}{\sqrt{\frac{2\pi}{3}} e^{\frac{\rho^2}{6}}i} \widetilde \Psi(z) \left(I+\frac{\gamma }{2\pi i} \ln\left(\frac{ z+s}{z-s}\right) \begin{pmatrix} 0&1&1\\ 0&0&0 \\ 0&0&0 \end{pmatrix}\right)
\begin{cases}
J_1^{-1}, &   z\in \mathtt{I},  \\
I, &   z\in \mathtt{II},  \\
       J_2^{-1},        &   z\in \mathtt{III}, \\
        J_2^{-1} J_3^{-1} ,  & z\in \mathtt{IV},\\
       J_1^{-1}J_0^{-1}J_5^{-1},  &  z\in \mathtt{V},\\
               J_1^{-1}J_0^{-1}                            &  z\in \mathtt{VI},
\end{cases}
 \end{equation}
where the function $\widetilde \Psi(z)$ is defined in \eqref{eq: tilde-psi}, the regions $\mathtt{I}-\mathtt{VI}$ are indicated in Figure \ref{fig:X} and the principal branch is taken for $\ln(z\pm s)$.  It is straightforward to check that $P^{(s)}(z)$ satisfies the same jump as $\Phi(z)$ on $\Sigma_k^{(s)}$, with the aid of the following relation
\begin{equation*}
J_5 J_0 J_1 = J_4 J_3 J_2= \begin{pmatrix}
1 & 1 & 1
\\
0 & 1 & 0
\\
0 & 0 & 1
\end{pmatrix}.
\end{equation*}
For $z \in (-s,s)$, as $\ln_+(z-s) - \ln_-(z-s) = 2 \pi i$, we have
\begin{equation}
 P^{(s)}_-(x)^{-1} P^{(s)}_+(x) =
J_5 J_0 J_1 \begin{pmatrix}
1 & -\gamma & -\gamma \\
0 & 1 & 0 \\ 0 & 0 & 1
\end{pmatrix} =  \begin{pmatrix}
1 & 1-\gamma & 1-\gamma \\
0 & 1 & 0 \\ 0 & 0 & 1
\end{pmatrix},
\end{equation}
which is same as the jump of $\Phi(z)$ on $(-s,s)$.  The endpoint condition of $P^{(s)}(z)$ as $z \to s$ can also be shown to agree with that of $\Phi(z)$ in \eqref{eq:X-near-s} directly. This means $P^{(s)}(z)$ defined in \eqref{eq: localPra-small-s-0} is indeed a desired local parametrix of $\Phi(z;s)$ near the origin.







As usual, we define the final transformation as
\begin{equation}\label{def:tildeR}
 \widetilde R(z)=\left\{
                \begin{array}{ll}
\sqrt{\frac{2\pi}{3}} e^{\frac{\rho^2}{6}}i \Phi(z)\Psi(z)^{-1} \Psi_0 , &\qquad  |z|>\delta, \\
        \Phi(z) P^{(s)}(z)^{-1},       &\qquad z\in D(0,\delta).
        \end{array}
              \right.
\end{equation}
It is easily seen that $\widetilde R$ satisfies the following RH problem.
\begin{rhp}
The $3\times 3$ matrix-valued function $\widetilde R(z)$ defined in \eqref{def:tildeR} has the following properties:
\begin{enumerate}
\item[\rm (1)]  $\widetilde  R(z)$ is analytic in $\mathbb{C} \setminus D(0,\delta) $.

\item[\rm (2)]  For $z \in \partial D(0, \delta)$,  we have $$ \widetilde R_+(z)= \widetilde R_-(z) J_{ \widetilde R}(z), $$  where
\begin{equation}
J_{ \widetilde R}(z) = \sqrt{\frac{2\pi}{3}} e^{\frac{\rho^2}{6}}i P^{(s)}(z) \Psi(z)^{-1} \Psi_0.
\end{equation}

\item[\rm (3)] As $z \to \infty$, we have $\widetilde R(z)=I+\Boh(z^{-1}).$
\end{enumerate}
\end{rhp}

The only $s$-dependent term in the jump matrix $J_{ \widetilde R}(z) $ comes from the term $\ln\left(\frac{z+s}{z-s}\right)$ in \eqref{eq: localPra-small-s-0}, which behaves like $\Boh(s)$ as $s \to 0^+$.  This gives us
\begin{equation}\label{eq:tildeR-Jump-est}
J_{ \widetilde R}(z)=I+\Boh(s), \qquad z \in \partial D(0, \delta),
\end{equation}
uniformly valid for the parameters $\gamma$ and $\rho$ in any compact subset of $[0,1)$ and $\mathbb{R}$. Thus, we obtain the following estimate
\begin{equation}\label{eq:tildeR-est}
\widetilde R(z)=I+\Boh(s),\qquad s\to 0^+,
\end{equation}
uniformly for $z \in \mathbb{C} \setminus \partial D(0,\delta)$.

We have now completed asymptotic analysis of the RH problem for $\Phi$, and are ready to prove our main results in what follows.

\section{Proofs of the main results}\label{sec:proofs}

\subsection{Proof of Theorem \ref{thm:specialsol}}

We have already proved in Proposition \ref{pro:Lax pair} that the functions $p_k(s)$ and $q_k(s)$, $k=0,1,2,3$, defined in \eqref{def: p-0} and \eqref{def: q-k} satisfy the equations \eqref{eq:SystemEq} and \eqref{eq:sum0}, and $p_0(s)$ and $q_0(s)$ satisfy the coupled differential system \eqref{eq:eqforp-0} and \eqref{eq:p0q0-1}. The existence part follows from the solvability of the associated RH problem. It then remains to show asymptotic results of $p_k$ and $q_k$, which are outcomes of asymptotic analysis of the RH problem for $\Phi$ and will be discussed next.

\subsubsection*{Asymptotics of $p_k$ and $q_k$  as $s\to+\infty$}

We first derive the asymptotics of $p_0(s)$ and $q_0(s)$. This comes from their relation with $\mathsf{\Phi}_1 $ (see \eqref{def: p-0}), which is the coefficient of $z^{-1}$-term in the large $z$ expansion of $\Phi(z)$ in \eqref{eq:asyX}. To find the asymptotics of $\mathsf{\Phi}_1 $ as $s \to +\infty$, we trace back  the transformations $\Phi \mapsto T \mapsto S\mapsto R$ in \eqref{def:PhiToT}, \eqref{def:TtoS} and \eqref{def:StoR} and obtain
\begin{equation}\label{eq:phioutside}
\Phi(sz) = \diag\left(s^{-\frac13},1, s^{\frac13}\right)R(z)N(z) e^{\Theta(sz)},\qquad z\in\mathbb{C} \setminus \{D(0,\delta) \cup D(1,\delta) \cup D(-1,\delta)\},
\end{equation}
where $\Theta(z)$ is defined in \eqref{def:Theta}. From \eqref{eq:R-est} and \eqref{eq:R1}, it follows
\begin{equation}\label{eq:Rlargez}
R(z) = I + \frac{\mathsf{R}_1 }{z} +\Boh(z^{-2}),\qquad z \to \infty,
\end{equation}
where
\begin{equation}\label{def:R1}
\mathsf{R}_1  = C_N \widehat{\Psi}_1C_N^{-1} s^{-\frac23} +\Boh(s^{-\frac43}) , \qquad s \to +\infty,
\end{equation}
with $C_N$ and $\widehat{\Psi}_1$ given in \eqref{eq:CN} and \eqref{eq:hatPsi1}, respectively. Combining \eqref{eq:asyX}, \eqref{eq:asyN}, \eqref{eq:phioutside} and  \eqref{eq:Rlargez}, we have
\begin{equation}
\mathsf{\Phi}_1 = s \diag\left(s^{-\frac13},1, s^{\frac13}\right) \left(\mathsf{R}_1  + \mathsf{N}_1  \right) \diag\left(s^{\frac13},1, s^{-\frac13}\right),
\end{equation}
where $\mathsf{N}_1$ is given in \eqref{eq:N1}. This, together with \eqref{def:R1}, implies that
\begin{equation}
(\mathsf{\Phi}_1)_{12}  = \sqrt{3} \beta i  s^{\frac23} +  \frac{\rho^3}{54} - \frac{\rho}{6} +\Boh(s^{-\frac23}), \quad (\mathsf{\Phi}_1)_{23} =  \sqrt{3} \beta i  s^{\frac23} +  \frac{\rho^3}{54} + \frac{\rho}{6} +\Boh(s^{-\frac23}).
\end{equation}
The asymptotics of $p_0(s)$ and $q_0(s)$ given in \eqref{thm:p-0-asy} and \eqref{thm:q-0-asy} then follow directly from the above formula and \eqref{def: p-0}.

	
To show the asymptotics of $p_k(s)$ and $q_k(s)$, $k=1,2,3$, we start with their definitions in \eqref{def: q-k}, in which $\Phi_0^{(0)}(s)$ appears in the local behavior of $\Phi(z)$ as $z \to s$. More precisely, we have from \eqref{eq:X-near-s} and \eqref{eq: Phi-expand-s}  that
\begin{equation} \label{Phi-00-formula}
\Phi_0^{(0)}(s) =  \lim_{z \to 1, \atop z\in \mathtt{II}} \Phi(sz)  \begin{pmatrix}
	1 & \frac{\gamma}{2\pi i} \ln(sz-s) & \frac{\gamma}{2\pi i} \ln(sz-s) \\
	0 & 1 & 0 \\
	0& 0 & 1
	\end{pmatrix} .
\end{equation}
Again, by tracing back the transformations $\Phi \mapsto T \mapsto S \mapsto R $ in \eqref{def:PhiToT}, \eqref{def:TtoS} and \eqref{def:StoR}, it follows that
\begin{equation}\label{eq:phisz1}
\Phi(sz)=\diag\left(s^{-\frac13},1, s^{\frac13}\right)R(z)P^{(1)}(z) e^{\Theta(sz)}, \qquad z\in D(1, \delta) \cap \mathtt{II}.
\end{equation}
As $R(z)$ is analytic at $z =1$ and $\theta_1(z)+ \theta_2(z) + \theta_3(z) = 0$,  it is readily seen from the above two formulas, the definition of $P^{(1)}(z)$ in \eqref{P1:P1-tilde} and \eqref{eq:P1Solution} that
\begin{eqnarray}
&&\Phi_0^{(0)}(s) = \diag\left(s^{-\frac13},1, s^{\frac13}\right)R(1) E_1(1)  \nonumber \\
&&\quad \times \lim_{z \to 1, \atop z\in \mathtt{II}} \left[ \diag \left( \Phi^{(\CHF)}(s^{\frac43}f(z);\beta)e^{-\frac\beta2 \pi i\sigma_3},1 \right) \widetilde\Theta(z) \begin{pmatrix}
	1 & \frac{\gamma}{2\pi i} \ln(sz-s) & \frac{\gamma}{2\pi i} \ln(sz-s) \\
	0 & 1 & 0 \\
	0& 0 & 1
	\end{pmatrix} \right],
\nonumber
\end{eqnarray}
where
\begin{equation}\label{eq:M}
\widetilde\Theta(z)= e^{-\frac{1}{2}\theta_3(sz)} \begin{pmatrix}
	1 & 0& 0 \\
	0 & 1 & 1 \\
	0& 0 & e^{\frac{3}{2}\theta_3(sz)}
	\end{pmatrix}.
\end{equation}	
Recall that $f(1) =0$, $f'(1)>0$ when $s$ is large enough (see \eqref{f1-f'1}), and the local behavior of $\Phi^{(\CHF)}(z)$ near the origin shown in \eqref{eq:H-expand-2}, one can easily verify that the terms involving $\ln(z-1)$ cancel out when we take the limit $z \to 1$. Thus, we arrive at the following expression of $\Phi_0^{(0)}(s)$:
\begin{align}
\Phi_0^{(0)}(s) &= e^{-\frac{1}{2}\theta_3(s)}  \diag\left(s^{-\frac13},1, s^{\frac13}\right)R(1) E_1(1) \diag \left( \Upsilon_0 ,1 \right)  \nonumber \\
&~~~ \times \begin{pmatrix}
	1 & \frac{\gamma}{2\pi i} \left( \ln s  - \ln (e^{-\frac{\pi i}{2}} s^{\frac{4}{3}} f'(1)) \right) & \frac{\gamma}{2\pi i} \left( \ln s  - \ln (e^{-\frac{\pi i}{2}} s^{\frac{4}{3}} f'(1))  \right) \\
	0 & 1 & 1 \\
	0& 0 & e^{\frac{3}{2}\theta_3(s)}
	\end{pmatrix} , \label{eq:Phi-0-0}
\end{align}
where $\Upsilon_0$ is given in \eqref{eq:H-expand-coeff-0}.

We now derive the asymptotics of $q_k(s)$, $k=1,2,3$. Inserting \eqref{eq:Phi-0-0} into \eqref{def: q-k}, it follows that
 \begin{align}\label{eq:q-expand}
\begin{pmatrix}
      q_1(s) \\
      q_2(s) \\
    q_3(s)
    \end{pmatrix} =\Phi_0^{(0)}(s) \begin{pmatrix}
      1\\
    0 \\
   0
    \end{pmatrix}
  &=e^{-\frac{1}{2}\theta_3(s)}\diag\left(s^{-\frac13},1, s^{\frac13}\right)R(1)E_1(1)\begin{pmatrix}
   \Gamma(1-\beta)e^{-\beta\pi i}  \\
      \Gamma(1+\beta) \\
0
    \end{pmatrix} .
\end{align}
With $E_1(1)$ given in \eqref{eq:E-expand-coeff-1}, we have
\begin{equation}
E_1(1)\begin{pmatrix}
   \Gamma(1-\beta)e^{-\beta\pi i}  \\
      \Gamma(1+\beta) \\
0
    \end{pmatrix} = C_N \begin{pmatrix}
-\omega \Gamma(1-\beta)e^{-\beta\pi i}c_1(s)+ \frac{\omega^2 \Gamma(1+\beta) }{e^{\frac{\beta\pi i}{3}} c_1(s) }  \vspace{4pt}  \\
 -\Gamma(1-\beta)e^{-\beta\pi i}c_1(s)+ \frac{\Gamma(1+\beta) }{e^{\frac{\beta\pi i}{3}} c_1(s) } \vspace{4pt}  \\
-\omega^2 \Gamma(1-\beta)e^{-\beta\pi i}c_1(s)+\frac{\omega \Gamma(1+\beta) }{e^{\frac{\beta\pi i}{3}} c_1(s) }
    \end{pmatrix},
\end{equation}
where $C_N$ and $c_1(s)$ are given in \eqref{eq:CN} and \eqref{eq:c-1}.
An important observation here is each entry of the $3 \times 1$  matrix  on the right-hand side of the above formula is the summation of a complex conjugate pair. To see this, note that if $s \in \mathbb{R}$ and $\gamma \in [0,1)$, it follows from \eqref{def:alpha}, \eqref{eq: theta-k-def} and \eqref{f1-f'1} that
\begin{equation} \label{eq:complex-para}
  \Re \beta = 0, \qquad \Re(\theta_1(s)-\theta_2(s)) = 0,  \qquad  \Im f'(1) = 0.
\end{equation}
This, together with \eqref{eq:c-1}, implies that
\begin{align}
\Gamma(1-\beta)e^{-\beta\pi i}c_1(s) & = \Gamma(1-\beta)e^{-\frac{2\beta\pi i}{3}}\left(\frac{3 \sqrt{3}}{2}\right)^{\beta}s^{\frac{4\beta}{3}}f'(1)^{\beta}e^{\frac{\theta_1(s)-\theta_2(s)}{2}}, \label{eq:complex-para-1} \\
\frac{\Gamma(1+\beta) }{e^{\frac{\beta\pi i}{3}} c_1(s) } & = \Gamma(1+\beta)e^{-\frac{2\beta\pi i}{3}} \left(\frac{3 \sqrt{3}}{2}\right)^{-\beta} s^{-\frac{4\beta}{3}}f'(1)^{-\beta}e^{-\frac{\theta_1(s)-\theta_2(s)}{2}}, \label{eq:complex-para-2}
\end{align}
which indeed constitute a complex conjugate pair due to \eqref{eq:complex-para}. With the aid of the large $s$ approximation of $R(z)$  \eqref{eq:R-est}, it follows from \eqref{eq:q-expand}--\eqref{eq:complex-para-2} that
\begin{align*}
q_1(s) &=  \frac{2i }{e^{\frac{1}{2}\theta_3(s)} s^{\frac{1}{3}}} \Im \left(-\omega \Gamma(1-\beta)e^{-\beta\pi i}c_1(s)  \right) \left(1+\Boh(s^{-\frac23}) \right)
\nonumber
\\
& =  \frac{2 e^{-\frac{2\beta\pi i}{3}} i}{e^{\frac{1}{2}\theta_3(s)} s^{\frac{1}{3}}} |\Gamma(1-\beta)|  \left(1+\Boh(s^{-\frac23}) \right)
\nonumber
\\
&~~~ \times \sin\left( -\frac{\pi }{3} +  \arg \Gamma(1-\beta) -\beta i \left(\frac{4}{3}\ln s+\ln \left(\frac{3\sqrt{3}}{2}\right)+\ln f'(1)\right) + \frac{\theta_1(s)-\theta_2(s)} {2i} \right) .
\end{align*}
Inserting the explicit formulas of $\theta_i(z)$, $i=1,2,3$, and of $f'(1)$ in \eqref{eq: theta-k-def} and \eqref{f1-f'1} into the above formula, we obtain the asymptotics of $q_1(s)$ shown in \eqref{thm:q-1-asy}. The asymptotics of $q_2(s)$ and $q_3(s)$ in \eqref{thm:q-2-asy} and \eqref{thm:q-3-asy} can be derived in a similar manner.



For the asymptotics of $p_k(s)$, $k=1,2,3$, we obtain from \eqref{def: q-k}, \eqref{eq:Phi-0-0} and the fact $(AB)^{-t} = A^{-t} B^{-t}$ for any two invertible matrices $A,B$ that
\begin{align}\label{eq:pasym1}
\begin{pmatrix}
  p_1(s) \\
      p_2(s) \\
    p_3(s)
    \end{pmatrix}&=-\frac{\gamma}{2\pi i}\Phi_0^{(0)}(s)^{-t}\begin{pmatrix}
      0\\
    1 \\
   1
    \end{pmatrix} \nonumber \\
     &= -\frac{\gamma \, e^{\frac{1}{2}\theta_3(s)}}{2\pi i}   \diag\left(s^{\frac13},1, s^{-\frac13}\right)R(1)^{-t} E_1(1)^{-t} \diag \left( \Upsilon_0^{-t} ,1 \right) \begin{pmatrix}
 0 \\
     1 \\
   0
    \end{pmatrix}.
\end{align}
With $E_1(1)$ and $\Upsilon_0$ given in \eqref{eq:E-expand-coeff-1} and \eqref{eq:H-expand-coeff-0}, we have
\begin{equation}\label{eq:pasym2}
E_1(1)^{-t} \diag \left( \Upsilon_0^{-t} ,1 \right) \begin{pmatrix}
 0 \\
     1 \\
   0
    \end{pmatrix} = \frac{C_N^{-t}}{3} \begin{pmatrix}
\omega \Gamma(1-\beta)e^{-\frac{2\beta\pi i}{3}}c_1(s)  + \frac{\omega^2 \Gamma(1+\beta)}{c_1(s)} \vspace{4pt}  \\
 \Gamma(1-\beta) e^{-\frac{2\beta\pi i}{3}} c_1(s)  + \frac{\Gamma(1+\beta)}{c_1(s)}  \vspace{4pt}  \\
\omega^2 \Gamma(1-\beta) e^{-\frac{2\beta\pi i}{3}} c_1(s)  + \frac{\omega \Gamma(1+\beta)}{c_1(s)}
    \end{pmatrix}.
\end{equation}
In view of \eqref{eq:complex-para-1} and \eqref{eq:complex-para-2}, it is readily seen that $\Gamma(1-\beta) e^{-\frac{2\beta\pi i}{3}} c_1(s)$ and $\frac{\Gamma(1+\beta)}{c_1(s)}$ in the above formula form a complex conjugates pair. The asymptotics of $p_k(s)$, $k=1,2,3$, in \eqref{thm:p-1-asy}--\eqref{thm:p-3-asy} then follow directly from the above two formulas. For the convenience of the reader, we include the derivation of \eqref{thm:p-2-asy}. By \eqref{eq:R-est}, \eqref{eq:pasym1} and \eqref{eq:pasym2}, it is easily seen that
\begin{align*}
p_2(s) &=  -\frac{\gamma e^{\frac{1}{2}\theta_3(s)}}{3\pi i} \Re \left(\Gamma(1-\beta) e^{-\frac{2\beta\pi i}{3}} c_1(s) \right) \left(1+\Boh(s^{-\frac23}) \right) \\
& =  -\frac{\gamma e^{\frac{1}{2}\theta_3(s) -\frac{\beta\pi i}{3}}  }{3\pi i} |\Gamma(1-\beta)|  \left(1+\Boh(s^{-\frac23}) \right) \\
& ~~~ \times \cos\left(  \arg \Gamma(1-\beta) -\beta i \left(\frac{4}{3}\ln s+\ln \left(\frac{3\sqrt{3}}{2} \right) + \ln f'(1)\right) + \frac{\theta_1(s)-\theta_2(s)} {2i} \right) .
\end{align*}
Since
\begin{equation}\label{eq:gammabeta}
\gamma = 1-e^{2\beta \pi i } = -2i e^{\beta \pi i } \sin( \beta\pi), \qquad \gamma\in[0,1),
\end{equation}
we obtain \eqref{thm:p-2-asy} from the above two formulas and \eqref{f1-f'1}.

\subsubsection*{Asymptotics of $p_k$ and $q_k$  as $s \to 0^+$}
The asymptotics of $p_k(s)$ and $q_k(s)$, $k=0,1,2,3$, as $s \to 0^+$ are the consequence of asymptotic analysis of the RH problem for $\Phi$ as $s \to 0^+$. From \eqref{def:tildeR} and \eqref{eq:tildeR-est}, we obtain, for $|z|>\delta$,
\begin{equation}\label{eq:PhiExpandOutSmalls}
\Phi(z)=\frac{1}{\sqrt{\frac{2\pi}{3}} e^{\frac{\rho^2}{6}}i } \left(I+\Boh(s) \right) \Psi_0^{-1} \Psi(z), \qquad  s \to 0^+.
\end{equation}
This, together with  \eqref{eq:asyPsi} and \eqref{def: p-0}, implies the asymptotics of $p_0(s)$ and $q_0(s)$ given in \eqref{thm:p-0-asy-0} and \eqref{thm:q-0-asy-0}.

To obtain the asymptotics of $p_k(s)$ and $q_k(s)$, $k=1,2,3$, again we rely on their definitions in \eqref{def: q-k} and the asymptotics of $\Phi_0^{(0)}(s)$. For this purpose, we observe from \eqref{eq: localPra-small-s-0} and  \eqref{def:tildeR} that
\begin{equation}
\Phi(z)= \widetilde{R}(z) \frac{\Psi_0^{-1}}{\sqrt{\frac{2\pi}{3}} e^{\frac{\rho^2}{6}}i} \widetilde \Psi(z) \left(I+\frac{\gamma }{2\pi i} \ln\left(\frac{ z+s}{z-s}\right) \begin{pmatrix} 0&1&1\\ 0&0&0 \\ 0&0&0 \end{pmatrix}\right),\qquad z\in D(0,\delta) \cap \mathtt{II}.
\end{equation}
A further appeal to \eqref{Phi-00-formula} shows that
\begin{equation}
  \Phi_0^{(0)}(s) = \widetilde{R}(s) \frac{\Psi_0^{-1}}{\sqrt{\frac{2\pi}{3}} e^{\frac{\rho^2}{6}}i} \widetilde \Psi(s) \left(I+\frac{\gamma \ln(2s)}{2\pi i}  \begin{pmatrix} 0&1&1\\ 0&0&0 \\ 0&0&0 \end{pmatrix}\right).
\end{equation}
Using \eqref{eq: tilde-psi} and \eqref{eq:tildeR-est}, we have from the above formula
\begin{equation} \label{eq:PhiExpandInSmalls}
  \Phi_0^{(0)}(s) = \frac{\Psi_0^{-1}}{\sqrt{\frac{2\pi}{3}} e^{\frac{\rho^2}{6}}i} \left(  \widetilde \Psi(0) + \Boh(s) \right) \left(I+\frac{\gamma \ln(2s)}{2\pi i}  \begin{pmatrix} 0&1&1\\ 0&0&0 \\ 0&0&0 \end{pmatrix}\right), \qquad s \to 0^+.
\end{equation}
Regarding the term $\widetilde \Psi(0)$, from the definition of $\mathcal{P}_j(z)$, $j=0,1,\ldots,5$, given in \eqref{def:pj}, it is immediate that
\begin{equation*}
\mathcal{P}'_0(0) = i \int_{-\infty}^{\infty} t \, e^{-\frac14 t^4-\frac{\rho}{2}t^2}\ud t = 0,  \qquad \mathcal{P}'_1(0) - \mathcal{P}'_4(0)  = -i \int_{-i\infty}^{i\infty} t \, e^{-\frac14 t^4-\frac{\rho}{2}t^2}\ud t = 0.
\end{equation*}
The above formulas, together with the definition of $\widetilde \Psi(z)$ in \eqref{eq: tilde-psi}, yield
\begin{equation}\label{eq:Phi0}
\widetilde \Psi(0) \begin{pmatrix}
      1\\
    0 \\
   0
    \end{pmatrix} = \begin{pmatrix}
	\mathcal{P}_0(0)   \\
	0   \\
	\mathcal{P}''_0(0)
	\end{pmatrix}, \qquad  \widetilde \Psi(0)^{-t} \begin{pmatrix}
	0 \\ 1 \\ 1
	\end{pmatrix}= \begin{pmatrix}
	0 \\ \frac{1}{\mathcal{P}'_1(0)} \\ 0
	\end{pmatrix},
\end{equation}
where note that $\mathcal{P}'_1(0) \neq 0$. Finally, using \eqref{def: q-k}, \eqref{eq:PhiExpandInSmalls} and \eqref{eq:Phi0}, we obtain the asymptotics of $p_k(s)$ and $q_k(s)$, $k=1,2,3$, in \eqref{thm:p-1-asy-0} and \eqref{thm:q-1-asy-0}, respectively.

This completes the proof of Theorem \ref{thm:specialsol}. \qed

\subsection{Proof of Theorem \ref{thm:TW}}
As long as the asymptotics of the functions $p_k(s)$ and $q_k(s)$, $k=0,1,2,3$, are available, the asymptotics of the associated Hamiltonian $H(s)$ follow directly from its definition \eqref{pro:H} and straightforward calculations. In view of \eqref{thm:p-0-asy-0}--\eqref{thm:q-1-asy-0}, this idea gives us the asymptotics of $H(s)$ as $s \to 0^+$ in \eqref{thm:H-asy-0}. The same strategy, however, leads to some messy to derive the asymptotics of $H(s)$ as  $s \to + \infty$, due to the complicated asymptotics of $p_k(s)$ and $q_k(s)$ as $s\to +\infty$ given in \eqref{thm:p-0-asy}--\eqref{thm:q-3-asy}. To overcome this difficulty, we make use of the relation \eqref{eq:H-F}, that is,
\begin{equation} \label{eq:H-Phi10}
H(s) = -\frac{\gamma}{2 \pi i} \begin{pmatrix}
0 & 1 & 1
\end{pmatrix} \Phi_1^{(0)}(s)  \begin{pmatrix}
      1\\
    0 \\
   0
    \end{pmatrix},
\end{equation}
where $\Phi_1^{(0)}(s)$ is given in the local behavior of $\Phi(z)$ near $z = s$. Similar to \eqref{Phi-00-formula},  we have from \eqref{eq:X-near-s} and \eqref{eq: Phi-expand-s}  that
\begin{equation} \label{Phi-10-formula}
\Phi_1^{(0)}(s) = \frac{\Phi_0^{(0)}(s)^{-1}}{s}   \lim_{z \to 1, \atop z\in \mathtt{II}} \left[ \Phi(sz)  \begin{pmatrix}
	1 & \frac{\gamma}{2\pi i} \ln(sz-s) & \frac{\gamma}{2\pi i} \ln(sz-s) \\
	0 & 1 & 0 \\
	0& 0 & 1
	\end{pmatrix} \right]',
\end{equation}
where $'$ denotes the derivative with respect to $z$. From \eqref{eq:Phi-0-0}, it follows that
\begin{equation}
\begin{pmatrix}
      0 & 1 & 1
    \end{pmatrix} \Phi_0^{(0)}(s)^{-1} = e^{\frac{1}{2}\theta_3(s)} \begin{pmatrix}
      0 & 1 & 0
    \end{pmatrix} \diag \left( \Upsilon_0^{-1} ,1 \right) E_1(1)^{-1} R(1)^{-1}    \diag\left(s^{\frac13},1, s^{-\frac13}\right).
\end{equation}
Regarding the limit in \eqref{Phi-10-formula}, we have from \eqref{eq:phisz1}, \eqref{P1:P1-tilde} and \eqref{eq:P1Solution} that
\begin{align}
& \lim_{z \to 1, \atop z\in \mathtt{II}} \left[ \Phi(sz)  \begin{pmatrix}
	1 & \frac{\gamma}{2\pi i} \ln(sz-s) & \frac{\gamma}{2\pi i} \ln(sz-s) \\
	0 & 1 & 0 \\
	0& 0 & 1
	\end{pmatrix} \right]' \begin{pmatrix}
      1\\
    0 \\
   0
    \end{pmatrix}
    \nonumber  \\
& = \diag\left(s^{-\frac13},1, s^{\frac13}\right) \lim_{z \to 1, \atop z\in \mathtt{II}} \left[ R(z) E_1(z) \diag \left( \Phi^{(\CHF)}(s^{\frac43}f(z);\beta)e^{-\frac{\beta \pi i}{2} \sigma_3},1 \right) \widetilde\Theta(z) \right] '  \begin{pmatrix}
      1\\
    0 \\
   0
    \end{pmatrix},
\end{align}
where $\widetilde\Theta(z)$ is defined in \eqref{eq:M}. Similar to the derivation of \eqref{eq:Phi-0-0}, we calculate the above limit and obtain
\begin{align}
& e^{-\frac{1}{2}\theta_3(s)} \diag\left(s^{-\frac13},1, s^{\frac13}\right) \bigg(   R'(1) E_1(1) \diag \left( \Upsilon_0 ,1 \right) + R(1) E_1'(1) \diag \left( \Upsilon_0 ,1 \right) \nonumber \\
 & + s^{\frac43} f'(1) R(1) E_1(1) \diag \left( \Upsilon_0 \Upsilon_1, 0 \right)  -\frac{s \theta_3'(s)}{2}  R(1) E_1(1) \diag \left( \Upsilon_0 ,1 \right)  \bigg) \begin{pmatrix}
      1\\
    0 \\
   0
    \end{pmatrix}, \label{Phi-10-limit-final}
\end{align}
where $\Upsilon_0$ and  $\Upsilon_1$ are constant matrices in \eqref{eq:H-expand-2}. A combination of \eqref{eq:H-Phi10}--\eqref{Phi-10-limit-final} gives us
\begin{align}
H(s) =& \, -\frac{\gamma}{2 \pi i s} \bigg( \diag \left( \Upsilon_0^{-1} ,1 \right) E_1(1)^{-1} R(1)^{-1} R'(1) E_1(1) \diag \left( \Upsilon_0 ,1 \right)
\nonumber \\
 & + \diag \left( \Upsilon_0^{-1} ,1 \right) E_1(1)^{-1} E_1'(1) \diag \left( \Upsilon_0 ,1 \right) + s^{\frac43} f'(1)  \diag \left( \Upsilon_1, 0 \right) -\frac{s \theta_3'(s)}{2} I \bigg)_{21}. \label{H-final-form}
\end{align}
We now derive the large positive $s$ asymptotics of the four terms in the bracket of the above formula one by one. From the estimate of $R(z)$ in \eqref{eq:R-est}, it follows that
\begin{equation}\label{eq:term1}
  \left(\diag \left( \Upsilon_0^{-1} ,1 \right) E_1(1)^{-1} R(1)^{-1} R'(1) E_1(1) \diag \left( \Upsilon_0 ,1 \right) \right)_{21} = \Boh(s^{-\frac23}).
\end{equation}
By the explicit expression of $\Upsilon_0$ in \eqref{eq:H-expand-coeff-0}, we have
\begin{align}\label{eq:term2}
&\left(\diag \left( \Upsilon_0^{-1} ,1 \right) E_1(1)^{-1} E_1'(1) \diag \left( \Upsilon_0 ,1 \right)\right)_{21} \nonumber \\ &=\Gamma(1+\beta)\Gamma(1-\beta)e^{-\beta \pi i} \left( \left(E_1(1)^{-1}E_1'(1)\right)_{22}-\left(E_1(1)^{-1}E_1'(1)\right)_{11} \right)\nonumber\\
&~~~ -\Gamma(1+\beta)^2 \left( E_1(1)^{-1}E_1'(1) \right)_{12} + \Gamma(1-\beta)^2e^{-2\beta \pi i} \left( E_1(1)^{-1}E_1'(1) \right)_{21}.
\end{align}
Since $\Re \beta =0$, it is readily seen from \cite[Equation 5.4.3]{DLMF} that
\begin{equation}\label{eq:gamma1beta}
\Gamma(1+\beta)\Gamma(1-\beta) = |\Gamma(1+\beta)|^2 = \frac{\beta \pi}{\sin(\beta \pi)}.
\end{equation}
This, together with the formula of $ E_1(1)^{-1} E_1'(1)$ given in \eqref{eq:E-expand-coeff-2}, implies that
\begin{align}
   &\Gamma(1+\beta)\Gamma(1-\beta) e^{-\beta \pi i} \left( \left(E_1(1)^{-1}E_1'(1)\right)_{22}-\left(E_1(1)^{-1}E_1'(1)\right)_{11} \right) \nonumber \\
   &= -\frac{\beta^2 \pi}{\sin(\beta \pi)} e^{-\beta \pi i} \left(\frac{f''(1)}{f'(1)}+1 \right) = -\frac{\beta^2 \pi}{\sin(\beta \pi)} e^{-\beta \pi i} \left(\frac{4}{3} + \Boh(s^{-\frac23}) \right), \label{H:derivation-1}
\end{align}
and
\begin{align}
  & -\Gamma(1+\beta)^2 \left( E_1(1)^{-1}E_1'(1) \right)_{12} + \Gamma(1-\beta)^2e^{-2\beta \pi i} \left( E_1(1)^{-1}E_1'(1) \right)_{21} \nonumber \\
  &= -\frac{\sqrt{3} i }{9 } e^{\frac{\beta \pi i}{3}}  \left( \frac{\Gamma(1+\beta)^2}{e^{\frac{2\beta\pi i}{3}} c_1^2(s)} + \Gamma(1-\beta)^2 e^{-2\beta \pi i}c_1^2(s)  \right) \nonumber \\
  &= -\frac{2\sqrt{3} i}{9} \frac{ \beta \pi   }{\sin(\beta \pi )} e^{-\beta \pi i} \cos\left(2\vartheta(s)\right) + \Boh(s^{-\frac23}), \label{H:derivation-2}
\end{align}
where $\vartheta(s)$ is given in \eqref{eq:vtheta}. In the derivation of \eqref{H:derivation-1} and \eqref{H:derivation-2}, we have also made use of the values of $f'(1)$ and $f''(1)$ in \eqref{f1-f'1}, and the complex conjugation property listed in \eqref{eq:complex-para-1} and \eqref{eq:complex-para-2}. The asymptotics of the last two term in \eqref{H-final-form} are simpler. It follows from \eqref{f1-f'1} and \eqref{eq:H-expand-coeff-1} that
\begin{equation} \label{H:derivation-3}
  \left( s^{\frac43} f'(1)  \diag \left( \Upsilon_1, 0 \right) \right)_{21} = \frac{ \beta \pi i \, e^{-\beta \pi i} }{\sin(\beta \pi )} \left(   \sqrt{3} s^{\frac43} -\frac{\sqrt{3}}{3}\rho s^{\frac23} \right),
\end{equation}
and obviously,
\begin{equation}\label{eq:term4}
\left(\frac{s \theta_3'(s)}{2} I \right)_{21} = 0.
\end{equation}
In view of \eqref{eq:gammabeta}, we obtain the asymptotics of $H(s)$ as $s\to +\infty$ in \eqref{thm:H-asy-infty} by combining
\eqref{H-final-form}--\eqref{eq:term2} and \eqref{H:derivation-1}--\eqref{eq:term4}.

To show the integral representation of $F$, it is readily seen from \eqref{eq:derivativeinsX-2} and \eqref{eq:H-F} that
\begin{equation}
\frac{\ud }{\ud s}F(s;\gamma,\rho)=2H(s).
\end{equation}
By taking integration on both sides of the above formula, we obtain \eqref{thm: IntRep}. Moreover, the integral is convergent near the origin in view of \eqref{thm:H-asy-0}.


This completes the proof of Theorem \ref{thm:TW}.
\qed

\subsection{Proof of Theorem \ref{thm:FAsy}}
The proof relies on the differential identities established in Proposition \ref{prop:H-diff}. First, we obtain from \eqref{eq: action-diff} that
\begin{align}\label{eq:pkqkH}
\int_0^{s}H(\tau)d\tau
&= \int_0^{s} \left(\sum_{k=0}^3p_k(\tau)q_k'(\tau)-H(\tau) \right) \ud \tau
\nonumber
\\ &~~~ -\frac{1}{4}\Big[2p_0(\tau)q_0(\tau)+p_2(\tau)q_2(\tau)+2p_3(\tau)q_3(\tau)-3\tau H(\tau)\Big]_{\tau = 0}^s.
\end{align}
To tackle the integral on the right-hand side of the above formula, we make use of \eqref{eq: dH-gamma}, where the differential identity still holds if we replace $\gamma$ by $\beta$. By integrating both sides of \eqref{eq: dH-gamma} with respect to $s$, it follows that
\begin{align}\label{eq:dbeta}
  \frac{\partial}{\partial \beta} \int_{0}^s \left(\sum_{k=0}^3 p_k(\tau)q_k'(\tau)-H(\tau) \right) \ud \tau & = \sum_{k=0}^3p_k(s)\frac{\partial}{\partial \beta}q_k(s) - \sum_{k=0}^3p_k(0)\frac{\partial}{\partial \beta}q_k(0)
  \nonumber
  \\
  & = \sum_{k=0}^3p_k(s)\frac{\partial}{\partial \beta}q_k(s),
\end{align}
where we have made use of the fact that $\sum\limits_{k=0}^3p_k(0)\frac{\partial}{\partial \beta}q_k(0) = 0$ in the second equality, as can be seen from the asymptotics of $p_k(s)$ and $q_k(s)$, $k=0,1,2,3$, as $s\to 0^+$ given in \eqref{thm:p-0-asy-0}--\eqref{thm:q-1-asy-0}.
A key observation now is, if $\beta = 0$ (equivalently, $\gamma = 0$), it is readily seen from \eqref{eq:H-F} and \eqref{def: q-k} that
$H(s)=0$ and $p_1(s)=p_2(s)=p_3(s)=0$. Furthermore, from \eqref{def:fh} and \eqref{eq:Y-jump},  it follows that $Y\equiv I$. This, together with \eqref{def: p-0}, \eqref{X1-formula} and \eqref{asyPsi:coeff}, implies that  $p_0(s) = \frac{\sqrt{2}}{2}\left( \frac{\rho^3}{54}+\frac{\rho}{2} \right)$, $q_0(s) = \frac{\sqrt{2}}{2} \left(-\frac{\rho^3}{54}+\frac{\rho}{2}\right)$. As a consequence,
\begin{equation}
  \sum_{k=0}^3 p_k(s)q_k'(s)-H(s) \equiv 0, \qquad \beta = 0.
\end{equation}
Integrating both sides of \eqref{eq:dbeta} with respect to $\beta$, we obtain from the above formula that
\begin{equation} \label{eq:intpkqkbeta}
  \int_{0}^s \left(\sum_{k=0}^3 p_k(\tau)q_k'(\tau)-H(\tau) \right) \ud \tau = \int_{0}^\beta \left( \sum_{k=0}^3p_k(s)\frac{\partial}{\partial \beta'}q_k(s) \right) \ud \beta',
\end{equation}
where we have used the notations $p_k(s)=p_k(s;\beta')$, $k=0,1,2,3$, for the integral on the right-hand side of the above formula.
Therefore, a combination of \eqref{eq:pkqkH} and \eqref{eq:intpkqkbeta} gives us
\begin{align} \label{eq:H-integral}
  \int_0^{s}H(\tau) \ud \tau & = \int_{0}^\beta \left( \sum_{k=0}^3p_k(s)\frac{\partial}{\partial \beta'}q_k(s) \right) \ud \beta'  \nonumber
  \\
  &~~~-\frac{1}{4}\Big[2p_0(\tau)q_0(\tau)-2p_1(\tau)q_1(\tau)-p_2(\tau)q_2(\tau)-3\tau H(\tau)\Big]_{\tau=0}^s,
\end{align}
where we have replaced $p_3(\tau)q_3(\tau)$ in \eqref{eq:pkqkH} by $-p_1(\tau)q_1(\tau)-p_2(\tau)q_2(\tau)$; see \eqref{eq: const1}.

We next use the asymptotics of $p_k$ and $q_k$, $k=0,1,2,3$, established in Theorem \ref{thm:specialsol} to calculate the asymptotics of the formula on the right-hand side of \eqref{eq:H-integral} as $s\to +\infty$. For
the terms  $p_k(s)\frac{\partial}{\partial \beta}q_k(s) $, $k=1,2,3$, appearing in the definite integral, we rewrite $\frac{\partial}{\partial \beta}q_k(s)$ as $q_k(s)\frac{\partial}{\partial \beta}\ln\left( q_k(s)\right)$ and obtain from \eqref{thm:p-1-asy}--\eqref{thm:p-3-asy} and \eqref{thm:q-1-asy}--\eqref{thm:q-3-asy} that
\begin{align}\label{eq:p1q1beta-asy}
& p_1(s) \frac{\partial}{\partial \beta}q_{1}(s)=p_1(s)q_1(s)  \frac{\partial}{\partial \beta}\ln\left( q_1(s)\right) 
\nonumber\\
& = -\frac{2}{3} \beta i \left( \sin\Big(2\vartheta(s)-\frac{2}{3}\pi\Big)+\sqrt{3}\beta i \sin(2\vartheta(s))-\frac{3}{2}\beta i \right) \nonumber \\
& ~~~ \times \left(-\frac{2}{3}\pi i+\frac{\partial}{\partial \beta}\ln | \Gamma(1-\beta)|+\cot \left(\vartheta(s)-\frac{\pi}{3} \right) \frac{\partial}{\partial \beta}\vartheta(s) \right) \left(1+\Boh\left(\frac{\ln s}{s^{\frac23}}\right) \right),
\end{align}
where the relation $|\Gamma(1-\beta)|^2 = \frac{\beta \pi}{\sin(\beta \pi)}$ is used again; see \eqref{eq:gamma1beta}.
Similarly, we have
\begin{align}\label{eq:p2q2beta-asy}
& p_2(s) \frac{\partial}{\partial \beta}q_{2}(s)= p_2(s)q_2(s)  \frac{\partial}{\partial \beta}\ln\left( q_2(s)\right)
\nonumber
\\
&=-\frac{2}{3} \beta i \sin(2\vartheta(s)) \left(-\frac{2}{3}\pi i+\frac{\partial}{\partial \beta}\ln | \Gamma(1-\beta)|+\cot (\vartheta(s)) \frac{\partial}{\partial \beta}\vartheta(s) \right) \left(1+\Boh\left(\frac{\ln s}{s^{\frac23}}\right) \right)
\end{align}
and
\begin{align}\label{eq:p3q3beta-asy}
& p_3(s) \frac{\partial}{\partial \beta}q_{3}(s) =p_3(s)q_3(s)\frac{\partial}{\partial \beta}\ln\left( q_3(s)\right)
\nonumber
\\
&=\biggl[ -\frac{2}{3} \beta i\sin \left( 2\vartheta(s)+\frac{2}{3}\pi \right) \left( -\frac{2}{3}\pi i+\frac{\partial}{\partial \beta}\ln | \Gamma(1-\beta)|+\cot \left(\vartheta(s)+\frac{\pi}{3} \right) \vartheta_{\beta}(s) \right) \nonumber\\
&~~~ -\beta^2 \left( \frac{2}{\sqrt{3}}  \sin(2\vartheta(s))-  1 \right) \left(-\frac{2}{3}\pi i+\frac{1}{\beta}+\frac{\partial}{\partial \beta}\ln | \Gamma(1-\beta)|+\cot \left(\vartheta(s)-\frac{\pi}{3} \right) \frac{\partial}{\partial \beta}\vartheta(s) \right) \biggr] \nonumber \\
& ~~~ \times
\left(1+\Boh\left(\frac{\ln s}{s^{\frac23}}\right) \right).
\end{align}
A combination of \eqref{eq:p1q1beta-asy}--\eqref{eq:p3q3beta-asy} and the fact that
\begin{align*}
\sin(2\vartheta)+\sin\left(2\vartheta+\frac{2}{3}\pi\right)+ \sin\left(2\vartheta-\frac{2}{3}\pi\right)&=0,
\\
\cos ^2 \vartheta +\cos ^2 \left(\vartheta+\frac{\pi}{3} \right)+\cos ^2 \left(\vartheta-\frac{\pi}{3} \right)&=\frac{3}{2},
\end{align*}
shows that
\begin{multline}\label{eq:pkqkbeta}
\int_0^\beta \left(\sum_{k=1}^3 p_k(s) \frac{\partial}{\partial \beta'}q_{k}(s) \right) \ud \beta'
\\
=\int_0^\beta \beta' \left(1-\frac{2 }{\sqrt{3}} \sin(2\vartheta(s))   -2i \frac{\partial}{\partial \beta'}\vartheta(s) \right) \ud \beta'  +\Boh \left(\frac{\ln s}{s^{\frac23}}\right).
\end{multline}
Comparing the above formula with \eqref{eq:H-integral}, it remains to derive the asymptotics of $p_0(s) \frac{\partial}{\partial \beta}q_{0}(s)$. Theoretically, this can be done by using the asymptotics of $p_0(s)$ and $q_0(s)$ in \eqref{thm:p-0-asy} and \eqref{thm:q-0-asy}, but then the $\Boh(s^{-\frac23})$-term therein has to be calculated explicitly. Alternatively, we refer to \eqref{eq: const2} and rewrite $p_0(s) \frac{\partial}{\partial \beta}q_{0}(s)$ as
\begin{equation}
  p_0(s) \frac{\partial}{\partial \beta}q_{0}(s) = - \sqrt{2} p_3(s) q_1(s)\frac{\partial}{\partial \beta}q_{0}(s) - \left(q_0(s) - \frac{\rho}{\sqrt{2}} \right) \frac{\partial}{\partial \beta}q_{0}(s).
\end{equation}
This gives us
\begin{equation}
  \int_0^{\beta} p_0(s) \frac{\partial}{\partial \beta'}q_{0}(s) \ud\beta'
=-\sqrt{2}\int_0^{\beta}p_3(s)q_1(s)\frac{\partial}{\partial \beta'}q_{0}(s) \ud \beta' - \left[\frac{1}{2}q_0(s)^2-\frac{\rho}{\sqrt{2}}q_{0}(s)\right]_{\beta'=0}^{\beta}.
\end{equation}
Inserting the asymptotics of $p_3(s)$, $q_0(s)$ and $q_1(s)$ given in \eqref{thm:p-3-asy}--\eqref{thm:q-1-asy} into the above formula, it follows that
\begin{multline}\label{eq:p0q0}
\int_0^{\beta} p_0(s) \frac{\partial}{\partial \beta'}q_{0}(s) \ud\beta'
= \int_0^{\beta} \beta' \left(\frac{2}{\sqrt{3}} \sin(2\vartheta(s))- 1 \right)\ud \beta'  - \frac{1}{2}q_0(s)^2 + \frac{\rho}{\sqrt{2}}q_{0}(s)  \\
 -\frac{1}{4}\left(\frac{\rho^3}{54}+\frac{3\rho}{2}\right)\left(-\frac{\rho^3}{54} + \frac{\rho}{2} \right) +\Boh \left(\frac{\ln s}{s^{\frac23}}\right) .
\end{multline}
This, together with \eqref{eq:pkqkbeta}, implies that
\begin{multline} \label{eq:pkqk-final-sum}
\int_0^\beta \left(\sum_{k=0}^3 p_k(s) \frac{\partial}{\partial \beta'}q_{k}(s) \right) \ud \beta'
= -2i \int_0^\beta \beta' \frac{\partial}{\partial \beta'}\vartheta(s)  \ud \beta'  -\frac{1}{2}q_0(s)^2+\frac{\rho}{\sqrt{2}}q_{0}(s) \\ -\frac{1}{4}\left(\frac{\rho^3}{54}+\frac{3\rho}{2}\right)\left(-\frac{\rho^3}{54} + \frac{\rho}{2} \right) +\Boh \left(\frac{\ln s}{s^{\frac23}}\right) .
\end{multline}
Recall the definition of $\vartheta(s)$ in \eqref{eq:vtheta}, we have
\begin{align}\label{vartheta-log}
  & -2i \int_0^\beta \beta' \frac{\partial}{\partial \beta'}\vartheta(s)  \ud \beta'
   \nonumber
   \\
   &= -2i  \int_0^{\beta}\beta' \left(\frac{\partial}{\partial \beta'} \arg\Gamma(1-\beta') -i \left (\frac{4}{3}\ln s+\ln \left(\frac{9}{2} \right) \right) \right) \ud\beta'
   \nonumber \\
   & = 2i \int_0^{\beta}\arg\Gamma(1-\beta')\ud\beta' -2\beta i\arg\Gamma(1-\beta)  - \beta^2\left(\frac{4}{3}\ln s+\ln \left(\frac{9}{2}\right)\right)
    \nonumber \\
   & = \beta\left(\ln \Gamma(1+\beta)-\ln \Gamma(1-\beta) \right)
 \nonumber \\
& ~~~ -\int_0^{\beta} \left( \ln \Gamma(1+\beta')-\ln \Gamma(1-\beta') \right) \ud\beta'  - \beta^2\left(\frac{4}{3}\ln s+\ln \left(\frac{9}{2}\right)\right),
\end{align}
where we have made use of the fact that $2i\arg\Gamma(1-\beta) =\ln\Gamma(1-\beta) - \ln\Gamma(1+\beta)$ if $\Re \beta =0$ in the last equality.
From the definition of the Barnes G-function given in \eqref{eq:G}, we further obtain from \eqref{vartheta-log} that
\begin{equation} \label{vartheta-barnes}
  -2i \int_0^\beta \beta' \frac{\partial}{\partial \beta'}\vartheta(s)  \ud \beta' = \ln\left( G(1+\beta) G(1-\beta) \right) + \beta^2 \left( 1- \frac{4}{3}\ln s-\ln \left(\frac{9}{2}\right)\right).
\end{equation}
Inserting \eqref{vartheta-barnes} into \eqref{eq:pkqk-final-sum}, it then follows from \eqref{eq:H-integral} and \eqref{thm:p-0-asy-0}--\eqref{thm:q-1-asy-0} that
\begin{align}
     \int_0^{s}H(\tau)\ud\tau &= \ln\left( G(1+\beta) G(1-\beta) \right) + \beta^2 \left( 1- \frac{4}{3}\ln s-\ln \left(\frac{9}{2}\right)\right)  \nonumber \\
      &~~~ -\frac{1}{2}q_0(s)^2+\frac{\rho}{\sqrt{2}}q_{0}(s)   + \frac{1}{4}\left(\frac{\rho^4}{54}-\frac{\rho^2}{2}\right) \nonumber  \\ 
      &~~~ -\frac{1}{4}\left(2p_0(s)q_0(s)-3sH(s)\right)   + \frac14 (2p_1(s)q_1(s)+p_2(s)q_2(s))+\Boh \left(\frac{\ln s}{s^{\frac23}}\right)
      \nonumber
      \\
    & = \ln\left( G(1+\beta) G(1-\beta) \right) + \beta^2 \left( 1- \frac{4}{3}\ln s-\ln \left(\frac{9}{2}\right)\right)+ \frac{1}{4}\left(\frac{\rho^4}{54}-\frac{\rho^2}{2}\right) \nonumber  \\ 
      &~~~+ \frac{q_0(s)}{\sqrt{2}}\left(p_3(s)q_1(s)+\frac{\rho}{2}\right)+\frac34 sH(s)
      \nonumber
      \\
      &~~~+ \frac14 (2p_1(s)q_1(s)+p_2(s)q_2(s))+\Boh \left(\frac{\ln s}{s^{\frac23}}\right),
\end{align}
where we have made use of the relation \eqref{eq: const2} in the second equality.

Finally, substituting \eqref{thm:p-1-asy}--\eqref{thm:q-2-asy} and \eqref{thm:H-asy-infty} into the above formula, we obtain the large gap asymptotic formula \eqref{main:F-asy} from \eqref{thm: IntRep}. The fact that the error term in \eqref{main:F-asy}  is $\Boh \left(s^{-\frac23}\right) $ instead of $\Boh \left(\frac{\ln s}{s^{\frac23}}\right) $  follows  directly  from substituting \eqref{thm:H-asy-infty} in \eqref{thm: IntRep}.

This completes the proof of Theorem \ref{thm:FAsy}.
\qed

\begin{appendices}

\section{Confluent hypergeometric parametrix}\label{sec:CHF}
The confluent hypergeometric parametrix $\Phi^{(\CHF)}(z)=\Phi^{(\CHF)}(z;\beta)$ with $\beta$ being a parameter is a solution of the following RH problem.

\subsection*{RH problem for $\Phi^{(\CHF)}$}
 \begin{description}
  \item(a)   $\Phi^{(\CHF)}(z)$ is analytic in $\mathbb{C}\setminus \{\cup^6_{j=1}\widehat\Sigma_j\cup\{0\}\}$, where the contours $\widehat\Sigma_j$, $j=1,\ldots,6,$ are indicated in Fig. \ref{fig:jumps-Phi-C}.

  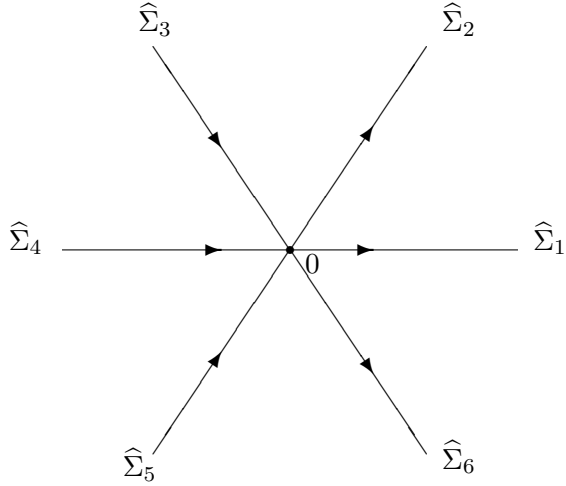
\begin{figure}[h]
\begin{center}
   \setlength{\unitlength}{1truemm}
   \begin{picture}(100,70)(-5,2)
       \put(40,40){\line(-2,-3){18}}
       \put(40,40){\line(-2,3){18}}
       \put(40,40){\line(-1,0){30}}
       \put(40,40){\line(1,0){30}}
  \put(40,40){\line(2,-3){18}}
    \put(40,40){\line(2,3){18}}

       \put(30,55){\thicklines\vector(2,-3){1}}
       \put(30,40){\thicklines\vector(1,0){1}}
       \put(50,40){\thicklines\vector(1,0){1}}
       \put(30,25){\thicklines\vector(2,3){1}}
      \put(50,25){\thicklines\vector(2,-3){1}}
       \put(50,55){\thicklines\vector(2,3){1}}


       \put(42,36.9){$0$}
         \put(72,40){$\widehat \Sigma_1$}
           \put(60,69){$\widehat \Sigma_2$}
             \put(20,69){$\widehat \Sigma_3$}
              \put(3,40){$\widehat \Sigma_4$}
       \put(18,10){$\widehat \Sigma_5$}
          \put(60,11){$ \widehat \Sigma_6$}

%

       \put(40,40){\thicklines\circle*{1}}
\end{picture}
   \caption{The jump contours for the RH problem for $\Phi^{(\CHF)}$.}
   \label{fig:jumps-Phi-C}
\end{center}
\end{figure}

  \item(b) $\Phi^{(\CHF)}$ satisfies the following jump condition:
  \begin{equation}\label{HJumps}
  \Phi^{(\CHF)}_+(z)=\Phi^{(\CHF)}_-(z) \widehat J_i(z), \quad z \in \widehat\Sigma_i,\quad j=1,\ldots,6,
  \end{equation}
  where
  \begin{equation*}
 \widehat J_1(z) = \begin{pmatrix}
    0 &   e^{-\beta \pi i} \\
    -  e^{\beta \pi i} &  0
    \end{pmatrix}, \qquad \widehat J_2(z) = \begin{pmatrix}
    1 & 0 \\
    e^{ \beta \pi i } & 1
    \end{pmatrix}, \qquad
    \widehat J_3(z) = \begin{pmatrix}
    1 & 0 \\
    e^{ -\beta\pi i} & 1
    \end{pmatrix},                                                         
  \end{equation*}
  \begin{equation*}
  \widehat J_4(z) = \begin{pmatrix}
    0 &   e^{\beta\pi i} \\
     -  e^{-\beta\pi i} &  0
     \end{pmatrix}, \qquad
      \widehat J_5(z) = \begin{pmatrix}
     1 & 0 \\
     e^{- \beta\pi i} & 1
     \end{pmatrix},\qquad
     \widehat J_6(z) = \begin{pmatrix}
   1 & 0 \\
   e^{\beta\pi i} & 1
   \end{pmatrix}.
  \end{equation*}

  \item(c) $\Phi^{(\CHF)}(z)$ satisfies the following asymptotic behavior at infinity:
  \begin{multline}\label{H at infinity}
 \Phi^{(\CHF)}(z)=(I + \Boh(z^{-1})) z^{-\beta \sigma_3}e^{-\frac{iz}{2}\sigma_3}
  \left\{\begin{array}{ll}
                         I, & ~0< \arg z <\pi,
                         \\
                        \begin{pmatrix}
                                                             0 &   -e^{\beta\pi i} \\
                                                            e^{-\beta\pi i } &  0
                        \end{pmatrix}, &~ \pi< \arg z<\frac{3\pi}{2},
                        \\
                        \begin{pmatrix}
                        0 &   -e^{-\beta\pi i} \\
                        e^{\beta\pi i} &  0
                        \end{pmatrix}, & -\frac{\pi}{2}<\arg z<0.
 \end{array}\right.
\end{multline}
\item(d) As $z\to 0$, we have $\Phi^{(\CHF)}(z)=\Boh(\ln |z|)$.

\end{description}

From \cite{IK}, it follows that the above RH problem can be solved explicitly in the following way. For $z$ belonging to the region bounded by the rays $\widehat \Sigma_1$ and $\widehat \Sigma_2$,
\begin{equation}\label{Hsolution}
\Phi^{(\CHF)}(z)=C_1\left(\begin{array}{ll}
\psi(\beta,1,e^{\frac{\pi i}{2}}z)e^{2 \beta \pi i}e^{-\frac{iz}{2}}&-\frac{\Gamma(1-\beta)}{\Gamma(\beta)}\psi(1-\beta,1,e^{-\frac{\pi i}{2}}z)e^{\beta \pi i}e^{\frac{iz}{2}}\\
-\frac{\Gamma(1+\beta)}{\Gamma(-\beta)}\psi(1+\beta,1,e^{\frac{\pi i}{2}}z)e^{\beta \pi i}e^{-\frac{iz}{2}}
&\psi(-\beta,1,e^{-\frac{\pi i}{2}}z)e^{\frac{iz}{2}}\end{array}\right),
\end{equation}
where the confluent hypergeometric function $\psi(a,b;z)$ is the unique solution to the Kummer's equation
\begin{equation} \label{Kummer-equation}
z\frac{\ud^2y}{\ud z^2}+(b-z)\frac{\ud y}{\ud z}-ay=0
\end{equation}
satisfying the boundary condition $\psi(a,b,z)\sim  z^{-a}$ as $z\to \infty$ and $-\frac{3\pi }{2} < \arg z < \frac{3\pi}{2}$;
see \cite[Chapter 13]{DLMF}.  The branches of the multi-valued functions are chosen such that $-\frac{\pi}{2}<\arg z<\frac{3\pi}{2}$ and
$$
C_1=
\begin{pmatrix}
e^{-\frac32 \beta \pi i} & 0
\\
0 & e^{\frac12 \beta \pi i}
\end{pmatrix}
$$
is a constant matrix. The explicit formula of $\Phi^{(\CHF)}(z)$ in the other sectors is then determined by using the jump condition \eqref{HJumps}.

We conclude this appendix by the detailed local behavior of $\Phi^{(\CHF)}(z)$ near the origin. For this purpose, let us recall the following properties of the confluent hypergeometric functions (cf. \cite[Equations 13.2.41 and 13.2.9]{DLMF}):
\begin{eqnarray}
\frac{1}{\Gamma\left(b\right)}\phi\left(a,b,z\right)=\frac{e^{\mp a\pi i}
}{\Gamma\left(b-a\right)}\psi\left(a,b,z\right)+\frac{e^{\pm(b-a)\pi i}}{
\Gamma\left(a\right)}e^{z}\psi\left(b-a,b,e^{\pm\pi i}z\right), \label{eq:rec-1} \\
\psi\left(a,1,z\right)=-\frac{\ln z}{\Gamma\left(a\right)}\phi\left(a,1,z\right) -\frac{1}{\Gamma\left(a\right)}\sum_{k=0}^{\infty}\frac{{\left(a\right)_{k}}}{(k!)^2} \left( \frac{\Gamma'\left(a+k\right)}{\Gamma\left(a+k\right)} - \frac{2 \, \Gamma'(1+k)}{\Gamma(1+k)} \right)z^{k}, \label{eq:rec-2}
\end{eqnarray}
where $\phi\left(a,b,z\right)$ is another solution to \eqref{Kummer-equation} defined by
\begin{equation} \label{phi-z-def}
\phi\left(a,b,z\right)=\sum_{k=0}^{\infty}\frac{{\left(a\right)_{k}}}{\left(b\right)_k}\frac{z^{k}}{k!}
\end{equation}
with $(a)_k=\frac{\Gamma(a+k)}{\Gamma(a)}=a(a+1)\cdots(a+k-1)$ being the Pochhammer symbol.
The function $\phi\left(a,b,z\right)$ is an entire function and satisfies the following relation
\begin{equation}\label{eq:relation}
\phi\left(a,b,z\right)=\phi\left(b-a,b,-z\right)e^{z}.
\end{equation}
It then follows from \eqref{eq:rec-1}--\eqref{eq:relation}, \eqref{HJumps} and \eqref{Hsolution} that
\begin{align}\label{eq:HsolutionExpand}
&\Phi^{(\CHF)}(z)
\nonumber
\\
&= \left(\begin{array}{ll}
\Gamma\left(1-\beta\right) e^{- \frac{\beta \pi i}{2} - \frac{iz}{2}} \phi(\beta,1,e^{\frac{\pi i}{2}}z)&\frac{e^{-\frac{\beta \pi i}{2} + \frac{iz}{2} } }{\Gamma(\beta)} \sum\limits_{k=0}^{\infty}\frac{{\left(1-\beta\right)_{k}}}{(k!)^2} \left( \frac{\Gamma'\left(1-\beta+k\right)}{\Gamma\left(1-\beta+k\right)} - \frac{2 \, \Gamma'(1+k)}{\Gamma(1+k)} \right) (-iz)^{k} \vspace{5pt}\\
\Gamma\left(1+\beta\right) e^{\frac{\beta \pi i}{2} + \frac{iz}{2}} \phi(-\beta,1,e^{-\frac{\pi i}{2}}z)
&-\frac{e^{\frac{\beta \pi i}{2} + \frac{iz}{2}}  }{\Gamma(-\beta)} \sum\limits_{k=0}^{\infty}\frac{{\left(-\beta\right)_{k}}}{(k!)^2}\left( \frac{\Gamma'\left(-\beta+k\right)}{\Gamma\left(-\beta+k\right)} - \frac{2 \, \Gamma'(1+k)}{\Gamma(1+k)} \right)(-iz)^{k} \end{array}\right)
\nonumber
\\
&~~~ \times  \begin{pmatrix} 1 & \frac{\sin(\beta \pi )}{\pi} \ln (e^{-\frac{\pi i}{2}}z) \\ 0 & 1  \end{pmatrix} ,
\end{align}
for $z$ belonging to the region bounded by the rays $\widehat \Sigma_2$ and $\widehat \Sigma_3$. This, together with \eqref{phi-z-def}, implies that
\begin{equation}\label{eq:H-expand-2}
\Phi^{(\CHF)}(z) e^{-\frac{\beta \pi i}{2} \sigma_3} = \Upsilon_0\left( I+ \Upsilon_1z+\Boh(z^2) \right) \begin{pmatrix} 1 & -\frac{\gamma}{2\pi i} \ln (e^{-\frac{\pi i}{2}}z) \\
0 & 1  \end{pmatrix}, \qquad  z \to 0,
\end{equation}
for $z$ belonging to the region bounded by the rays $\widehat \Sigma_2$ and $\widehat \Sigma_3$, where $\gamma=1-e^{2\beta \pi i}$,
\begin{align}\label{eq:H-expand-coeff-0}
\Upsilon_0
=\begin{pmatrix} \Gamma\left(1-\beta\right) e^{-\beta \pi i} &\frac{1}{\Gamma(\beta)} \left( \frac{\Gamma'\left(1-\beta\right)}{\Gamma\left(1-\beta\right)} +2\gamma_{\textrm{E}} \right) \vspace{5pt} \\
\Gamma\left(1+\beta\right) & -\frac{e^{\beta \pi i}}{\Gamma(-\beta)} \left( \frac{\Gamma'\left(-\beta\right)}{\Gamma\left(-\beta\right) } +2\gamma_{\textrm{E}}\right) \end{pmatrix}
\end{align}
with $\gamma_{\textrm{E}}$ being the Euler's constant,
and
 \begin{equation}\label{eq:H-expand-coeff-1}
(\Upsilon_1)_{21}=\frac{ \beta \pi i \, e^{-\beta \pi i} }{\sin(\beta \pi )}.
\end{equation}

\end{appendices}

\section*{Acknowledgements}
We thank Christophe Charlier for helpful comments on this work. Dan Dai was partially supported by grants from the City University of Hong Kong (Project No. 7005597 and 7005252), and grants from the Research Grants Council of the Hong Kong Special Administrative Region, China (Project No. CityU 11303016 and CityU 11300520). Shuai-Xia Xu was partially supported by National Natural Science Foundation of China under grant numbers 11971492, 11571376 and 11201493, and Natural Science Foundation for Distinguished Young Scholars of Guangdong Province of China. Lun Zhang was partially supported by National Natural Science Foundation of China under grant numbers 11822104 and 11501120, ``Shuguang Program'' supported by Shanghai Education Development Foundation and Shanghai Municipal Education Commission, and The Program for Professor of Special Appointment (Eastern Scholar) at Shanghai Institutions of Higher Learning.


\end{document}